\documentclass[journal,10pt,twocolumn]{IEEEtran}
\usepackage{amsmath,amsfonts,amsthm,amssymb,bbm}
\usepackage{cite}
\usepackage{balance}
\usepackage{mathrsfs}
\usepackage{xcolor}
\usepackage{graphicx}

\allowdisplaybreaks
\DeclareMathOperator*{\argmax}{arg\,max}

\theoremstyle{plain} 
\newtheorem{theorem}{Theorem}
\newtheorem{corollary}{Corollary}
\newtheorem{definition}{Definition}
\newtheorem{lemma}{Lemma}
\newtheorem{proposition}{Proposition}
\theoremstyle{definition} \newtheorem{remark}{Remark}
\theoremstyle{definition} \newtheorem{example}{Example}

\newcommand{\blue}{\color{black}}
\newcommand{\nblue}{\color{black}}

\title{An Operational Approach to Information Leakage via Generalized Gain Functions}
\author{}

\begin{document}
\author{Gowtham R. Kurri, \IEEEmembership{Member,~IEEE}, Lalitha Sankar, \IEEEmembership{Senior Member,~IEEE}, Oliver Kosut, \IEEEmembership{Senior Member,~IEEE}\thanks{This article was presented in part at the 2021 International Symposium on Information Theory (ISIT) and 2022 ISIT. This work is supported in part by NSF grants CIF-1901243, CIF-1815361, CIF-2007688, CIF-2134256, CIF-2031799, and CIF-2312666.

{\blue Gowtham R. Kurri was with the School of Electrical, Computer and Energy Engineering at Arizona State University. He is now with the Signal Processing and Communications Research Centre at International Institute of Information Technology, Hyderabad, Telangana - 500032, India (e-mail: {gowtham.kurri@iiit.ac.in}).}

Lalitha Sankar and Oliver Kosut are with the School of Electrical, Computer and Energy Engineering, Arizona State University, Tempe, AZ 85281 USA (email: lalithasankar@asu.edu; okosut@asu.edu). 
}}
\maketitle
\begin{abstract}
    We introduce a \emph{gain function} viewpoint of information leakage by proposing \emph{maximal $g$-leakage}, a rich class of operationally meaningful leakage measures that subsumes recently introduced leakage measures {\blue{---}} {maximal leakage} and {maximal $\alpha$-leakage}. In maximal $g$-leakage, the gain of an adversary in guessing an unknown random variable is measured using a {gain function} applied to the probability of correctly guessing. In particular, maximal $g$-leakage captures the multiplicative increase, upon observing $Y$, in the expected gain of an adversary in guessing a randomized function of $X$, maximized over all such randomized functions. {\blue{We also consider the scenario where an adversary can make multiple attempts to guess the randomized function of interest. We show that maximal leakage is an upper bound on maximal $g$-leakage under multiple guesses, for any non-negative gain function $g$. We obtain a closed-form expression for maximal $g$-leakage under multiple guesses for a class of concave gain functions. We also study maximal $g$-leakage measure for a specific class of gain functions related to the $\alpha$-loss, that interpolates log-loss ($\alpha=1$) and (soft) 0-1 loss ($\alpha=\infty$). In particular, we first completely characterize the minimal expected $\alpha$-loss under multiple guesses and analyze how the corresponding leakage measure is affected with the number of guesses. We show that a new measure of divergence that belongs to the class of Bregman divergences captures the relative performance of an arbitrary adversarial strategy with respect to an optimal strategy in minimizing the expected $\alpha$-loss. Finally, we study two variants of maximal $g$-leakage depending on the type of adversary and obtain closed-form expressions for them, which do not depend on the particular gain function considered as long as it satisfies some mild regularity conditions. We do this by developing a variational characterization for the R\'{e}nyi divergence of order infinity which naturally generalizes the definition of pointwise maximal leakage to incorporate arbitrary gain functions.}} 
\end{abstract}

\section{Introduction}
A fundamental question in many privacy and secrecy problems is --- how much information does a random variable $Y$ that represents observed data by an adversary leak about a correlated random variable $X$ that represents sensitive data? For example, $X$ may be a secret that must be kept confidential and the observation $Y$ could be an inevitable consequence of a system design, for instance, in exchange for certain services. A number of approaches to quantify such information leakage have been proposed in both computer science~\cite{Dwork06,smith2009foundations,braun2009quantitative,kasiviswanathan2011can,Alvimetal12,duchi2013local,Alvimetal14} and information theory~\cite{MerhavA99,CalmonF12,SchielerC14,IssaW15,Asoodeh15,Issaetal,Rassouli20,Liaoetal,saeidian2022pointwise}.

Leakage measures with an associated operational meaning are of specific interest in the literature since an upper bound on such a leakage measure allows the system designer to ensure certain guarantees on the system. Recently, Issa~\emph{et al.}~\cite{Issaetal} proposed one such leakage measure in the \emph{guessing} framework. In particular, Issa \emph{et al.}~\cite{Issaetal} consider an adversary interested in a (possibly randomized) function of $X$ and study the logarithm of the multiplicative increase, upon observing $Y$, of the probability of correctly guessing a  randomized function of $X$, say, $U$. Moreover, this quantity is maximized over all the random variables $U$ such that $U-X-Y$ forms a Markov chain capturing the scenario that the function of interest $U$ is unknown to the system designer. The resulting quantity is referred to as \emph{maximal leakage} (MaxL) and an upper bound on it limits the amount of information leakage of any arbitrary randomized function of $X$ through $Y$. Liao~\emph{et al.}~\cite{Liaoetal} later generalized MaxL to a family of leakages, \emph{maximal $\alpha$-leakage} (Max-$\alpha$L), for $\alpha\in(1,\infty)$, that allows tuning the measure to specific applications. In particular, similar to MaxL, Max-$\alpha$L quantifies the maximal logarithmic increase in a monotonically increasing power function (dependent on $\alpha$) applied to the probability of correctly guessing. We remark that Max-$\alpha$L provides an operational interpretation to mutual information (for $\alpha=1$) in the context of privacy leakage, which was an open problem earlier~\cite{Issa2017,Issaetal}. Saeidian~\emph{et al.}~\cite{saeidian2022pointwise} introduced a variant of maximal leakage, called \emph{pointwise maximal leakage}, capturing the amount of information leaked about $X$ due to disclosing a single outcome $Y=y$ rather than focusing on the \emph{average outcome} as in maximal leakage. These operationally motivated leakage measures find applications in many areas such as in privacy-utility trade-offs~\cite{Liaoetal}, private information retrieval~\cite{QianZTL22}, hypothesis testing~\cite{LiaoSCT17}, source coding~\cite{Liuetal21}, membership inference~\cite{SaedianCOS21}, and age of information~\cite{NityaYSM22}.

We extend the aforementioned line of work by focusing on arbitrary \emph{gain functions} $g:[0,1]\rightarrow [0,\infty)$ applied to the probability of correctness (see Fig.~\ref{Figure-gain}). In particular, we define \emph{maximal $g$-leakage} as the maximal logarithmic increase in the gain (applied to the probability of correctness) of an adversary and study its properties for arbitrary $g$. {\blue{We also consider \emph{maximal $g$-leakage under multiple guesses} where the adversary is allowed to make multiple guesses.}} Further, we introduce and study variants of maximal $g$-leakage depending on the type of adversary by developing a new variational characterization for R\'{e}nyi divergence of order $\infty$~\cite{renyi1961measures}. A variational characterization for a divergence transforms its definition into an optimization problem. Variational characterizations for R\'{e}nyi divergences of order $\alpha\in\mathbb{R}\setminus\{0\}$~\cite{renyi1961measures} are studied in the literature~\cite{Shayevitz11,VanH14,Sason16,Anantharam,BirrellDKRW21}. We also study information leakage when the adversary is allowed to make multiple attempts to guess the randomized function of interest by focusing on a specific gain function related to the $\alpha$-loss~\cite{arimoto1971information,LiaoKS20,SypherdDSK19,Liaoetal} interpolating log-loss ($\alpha=1$) and (soft) $0$-$1$ loss ($\alpha=\infty$). 
\begin{figure*}[htbp]
    \centering
    \includegraphics[scale=1.2]{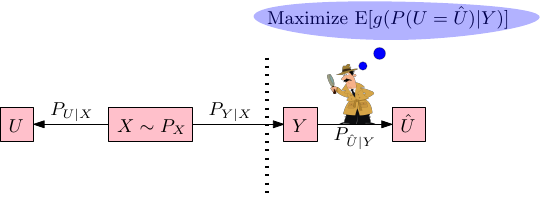}
    \caption{A gain function viewpoint of leakage. An adversary observes $Y$ and wants to maximize, on an average, the gain function applied to the probability of correctly guessing a randomized function of $X$, denoted by $U$.}
    \label{Figure-gain}
\end{figure*}

\subsection{Related Work}\label{subsection:related-work}
Most of the approaches in the literature start with a particular metric developed in other contexts and study the properties that follow from the definition. It is difficult to associate an operational meaning to such a leakage measures. In fact, such approaches can label evidently insecure systems as secure (see~\cite[Section~3.6]{Issa2017}). An alternative approach is via a specific threat model of a \emph{guessing} adversary giving an operationally meaningful interpretation to information leakage~\cite{smith2009foundations,braun2009quantitative,Alvimetal12,Issaetal,Liaoetal,saeidian2022pointwise}.  Smith~\cite{smith2009foundations} defines \emph{min-entropy leakage} as the logarithm of the multiplicative increase, upon observing $Y$, of the probability of correctly guessing $X$. Braun \emph{et al.}~\cite{braun2009quantitative} consider the maximization of min-entropy leakage over all prior distributions on $X$ resulting in a leakage measure (later shown to be equal to {maximal leakage}~\cite{Issaetal}) known as \emph{min-capacity} that is dependent only on the conditional distribution (or channel) $P_{Y|X}$. Alvim~\emph{et al.}~\cite{Alvimetal12,Alvimetal14} defined \emph{$g$-leakage} and \emph{$g$-capacity} by introducing a gain function $g:\mathcal{X}\times\hat{\mathcal{X}}\rightarrow \mathbb{R}$ instead of looking at the probability of correctly guessing $X$. Alvim~\emph{et al.}~\cite{Alvimetal12} showed that maximizing $g$-capacity over all gain functions $g$ yields maximal leakage. {\blue{Note that the threat model in all these works focuses on guessing $X$ itself rather than a (potentially) randomized function $U$ of $X$. Issa \emph{et al.}~\cite{Issaetal} considered maximal gain leakage based on the gain functions and $g$-leakage introduced by Alvim~\emph{et al.}~\cite{Alvimetal12}, where the adversary is interested in a (potentially) randomized function of $X$ and maximum is taken over all gain functions $g:\mathcal{U}\times\hat{\mathcal{U}}\rightarrow [0,1]$. Similar to Alvim~\emph{et~al.}~\cite{Alvimetal12}, they showed that maximal gain leakage is equal to maximal leakage. As we show later, our definition of gain function where a real-valued function is applied to the probability of correctly guessing can be seen as a special case of the gain function definition of Alvim~\emph{et al.}~\cite{Alvimetal12,Alvimetal14}. However, the information leakage measures we study in this work differ substantially from that of the works~\cite{Alvimetal12,Alvimetal14,Issaetal}. In particular, Alvim~\emph{et al.}~\cite{Alvimetal12,Alvimetal14} and Issa \emph{et al.}~\cite{Issaetal} consider leakage measures with a maximum taken over all the gain functions while the problem of computing leakage measures with specific gain functions (other than the identity gain function that corresponds to maximal leakage) remains open and is conjectured to be challenging~\cite[Section~VI.A]{Alvimetal14}. In this work, we study leakage measures with specific gain functions and obtain their closed-form expressions. We present more details on the distinction from Alvim~\emph{et al.}~\cite{Alvimetal12,Alvimetal14} and Issa \emph{et~al.}~\cite{Issaetal} in Section~\ref{section:maximalg-leakage} (in particular, see Remark~\ref{distinction-alvimissa}}).}

Instead of multiplicative increase, the notion of additive increase has also been considered in quantifying privacy leakage. For example, an additive version of $g$-leakage is studied by Alvim~\emph{et al.}~\cite{Alvimetal14}. Also, the notion of \emph{semantic security} in cryptography defines `advantage' as the additive increase, upon observing the encrypted message, of the probability of correctly guessing the value of the function. Calmon \emph{et al.}~\cite{Calmonetal13} and Li and El Gamal~\cite{Lig18} use maximal correlation as a measure of information leakage. Mutual information has been used as a privacy measure in many works, see e.g., \cite{Shannon49,PrabhakaranR07,TyagiNG11,CalmonF12,SankarRP13,Asoodeh15,Wangetal16,Kameletal19}. Asoodeh~\emph{et~al.}~\cite{Asoodehetal17} use the probability of correctly guessing as a privacy measure. Li~\emph{et al.}~\cite{LiOG19} use hypothesis test performance of an adversary to measure leakage. Total variation distance is used as privacy measure by Rassouli and G\"{u}nd\"{u}z~\cite{Rassouli20}. 

Another line of work in privacy leakage is based on indistinguishability, i.e., whether an adversary can distinguish between two items of interest. Differential privacy~\cite{Dwork06}, proposed in the context of querying databases, ensures that all databases that differ only in one entry produce an output to a query with \emph{almost} equal probabilities. Similar to differential privacy and pointwise maximal leakage, Calmon and Fawaz~\cite{CalmionF12}, Issa~\emph{et al.}~\cite{Issaetal}, and Jiang~\emph{et~al.}\cite{JiangSTL21} study privacy leakage providing worst-case guarantees, note that Issa~\emph{et al.}~\cite{Issaetal} study an average case leakage measure (maximal leakage) also. An equivalent notion of differential privacy using (conditional) mutual information is studied in \cite{CuffY16}. The notion of differential privacy is known to be very strict and has limited applicability~\cite{McSherry10,ComasD13}. Approximate differential privacy~\cite{dwork2006our} and R\'{e}nyi differential privacy\cite{mironov2017renyi} are proposed as relaxations of differential privacy to allow data releases with higher utility. For an extensive list of leakage measures see the surveys by Wagner and Eckhoff~\cite{wagner2018technical}, Bloch~\emph{et~al.}~\cite{Blochsurvey}, and Hsu~\emph{et al.}~\cite{hsuetal21}.

\subsection{Main Contributions}
The main contributions of this paper are as follows:
\begin{itemize}
    \item {\blue{We show that maximal leakage is an upper bound on maximal $g$-leakage under $k\geq 1$ guesses for any arbitrary non-negative gain function $g$ (Proposition~\ref{prop:upperboundg}). We obtain closed-form expression for maximal $g$-leakage under $k\geq 1$ guesses for a class of gain functions $g$. Specifically, when $g$ is a non-negative concave function with a finite non-zero derivative at 0, we show that maximal $g$-leakage under $k\geq 1$ guesses is equal to the Sibson mutual information of order infinity~\cite{sibson1969information}, i.e., the maximal leakage (Theorem~\ref{thm:maximalg-leakage-concaveg}). {\nblue In addition, we obtain a closed-form expression for maximal $g$-leakage when $g(t)=1+t$, for binary $X$ and $Y$ (Theorem~\ref{thm:shiftedmaxL}). An interesting aspect of this expression is that it is fundamentally different in structure from that of maximal leakage and maximal $\alpha$-leakage.}}} 
    \item We then focus on maximal $g$-leakage under multiple guesses for a specific class of gain functions related to the $\alpha$-loss, $\alpha\in(0,\infty)$~\cite{arimoto1971information,LiaoKS20,SypherdDSK19,Liaoetal}, which interpolates log-loss ($\alpha=1$) and $0$-$1$ loss ($\alpha=\infty$). We define two leakage measures, the $\alpha$-leakage and the maximal $\alpha$-leakage under multiple guesses. We show that $\alpha$-leakage does not change with the number of guesses for a class of probability distributions (Theorem~\ref{thm:robustness-of-alphaleakage})~\cite{KurriKS21}. We prove that maximal $\alpha$-leakage under multiple guesses is at least that of with a single guess (Theorem~\ref{thm:lowerbound-maxalphaleakage-kguesses}).
    \item To prove these, we completely characterize the minimal expected $\alpha$-loss under $k$ guesses (Theorem~\ref{thm:minloss-alpha-kguesses})~\cite{KurriKS21}, thereby recovering the known results for $\alpha=\infty$~\cite{Issaetal}. We illustrate an optimal guessing strategy of the adversary through Examples~\ref{Example1} and \ref{Example2}. To the best of our knowledge, such a result even for log-loss under multiples guesses was not explored earlier. We derive a technique for transforming the optimization problem over probability simplex associated with multiple random variables to that of with a single random variable using tools drawn from duality in linear programming, which may be of independent interest. 
    \item We introduce and study two variants of maximal $g$-leakage, namely, opportunistic maximal $g$-leakage (when the adversary could choose the function of interest depending on the realization of $Y$) and maximal realizable $g$-leakage (when the adversary is interested in \emph{maximum} guessing performance over all the realizations of $Y$ instead of \emph{average} performance). We obtain closed-form expressions for these leakage measures in terms of Sibson mutual information of order $\infty$ and R\'{e}nyi divergence of order $\infty$, respectively (Corollary~\ref{thm:alphaleakage-variants}). 
    \item We do this by devising a new variational characterization for R\'{e}nyi divergence of order $\infty$ expressed in terms of the ratio of maximal expected gains in guessing a randomized function of $X$, which may be of independent interest (Theorem~\ref{theorem:main-varchar})~\cite{KurriKS22}. An important aspect of our characterization is that it remains agnostic to the particular gain function considered as long as it satisfies some mild regularity conditions. Even though the variational characterizations mentioned earlier are presented for any finite order $\alpha$, one can obtain such characterizations for $\infty$-R\'{e}nyi divergence by applying a limiting argument (see the discussion above Proposition~\ref{proposition}). Our characterization differs from those characterizations in view of its connection to guessing, and more importantly because of its robustness to the gain function. We also show that our variational characterization naturally extends the notion of pointwise maximal leakage~\cite{saeidian2022pointwise} to incorporate arbitrary gain functions under mild regularity assumptions retaining the same closed-form expression (Corollary~\ref{thm:ptwisegleakage}). 
\end{itemize}

\subsection{Organization of the Paper}
The remainder of this paper is organized as follows. We introduce the gain function viewpoint and review maximal leakage and maximal $\alpha$-leakage in Section~\ref{section:preliminaries}. {\blue{In Section~\ref{section:maximalg-leakage}, we present our results on maximal $g$-leakage under multiple guesses. In Section~\ref{section:multiple-guesses}, we present our results on maximal $g$-leakage under multiple guesses by focusing on a class of specific gain functions related to the $\alpha$-loss~\cite{arimoto1971information,LiaoKS20,SypherdDSK19}, parameterized by $\alpha\in(0,1)\cup(1,\infty)$.}} In Section~\ref{section:variational}, we develop a variational characterization for $\infty$-R\'{e}nyi divergence and show how it can be employed to obtain closed-form expressions for variants of maximal $g$-leakage depending on the type of adversary. The proofs of all the results are presented in Appendix.

\section{Preliminaries}\label{section:preliminaries}
\emph{Notation.} We use capital letters to denote random variables, e.g., $X$, and capital calligraphic letters to denote their corresponding alphabet, e.g., $\mathcal{X}$. We consider only finite alphabets in this work. We use $\mathbb{E}_X[\cdot]$ to denote expectation with respect to $P_X$ and $U-X-Y$ to denote that the random variables form a Markov chain. We use $\text{supp}(X):=\{x:P_X(x)>0\}$ to denote the support set of $X$. We use $H(X)$, $I(X;Y)$, and $D(P_X\|Q_X)$ to denote entropy, mutual information, and relative entropy, respectively. Given two probability distributions $P_X$ and $Q_X$ over an alphabet $\mathcal{X}$, we write $P_X\ll Q_X$ to denote that $P_X$ is absolutely continuous with respect to $Q_X$. Finally, we use $\log$ to denote the natural logarithm.

 We begin by defining the maximal expected gain of an adversary in guessing an unknown random variable, which is {\blue{an}} impetus for this work. {\blue{A \emph{gain function} is defined as a function $g:[0,1]\rightarrow [0,\infty)$, and can be interpreted as a function that is applied to the probability of correctly guessing an unknown random variable.}} 
\begin{definition}[Maximal expected gain]\label{defn:gain}
{\blue{Given a gain function $g:[0,1]\rightarrow \mathbb{R}$ and a probability distribution $P_X$ on a finite alphabet $\mathcal{X}$, the maximal expected gain is defined as }}
\begin{align}\label{eqn:maxi-expectedgain}
    \sup_{P_{\hat{X}}}\mathbb{E}_X\left[g(P_{\hat{X}}(X)\right],
\end{align}
{\blue{where $\hat{X}$ represents an estimator of $X$ with same support as $X$.}}
\end{definition}
{\blue{Though maximal expected gain in \eqref{defn:gain} is defined for real-valued gain functions, we restrict our attention to non-negative gain functions for most of the time in this paper. However, some of our results hold for negative gain functions also, e.g., $g(t)=\log{t}$, $t\in[0,1]$ (we discuss this in Remark~\ref{remark:log} in Section~\ref{section:variational}). The objective function in \eqref{eqn:maxi-expectedgain} is the expected value of gain function $g$ applied to the probability of correctly guessing an unknown random variable $X$. Alvim~\emph{et al.}~\cite{Alvimetal12,Alvimetal14} consider a similar gain function viewpoint of information leakage with gain functions $g:\mathcal{X}\times\hat{\mathcal{X}}\rightarrow \mathbb{R}$. It is worth mentioning that the gain function of Alvim~\emph{et al.}~\cite{Alvimetal12,Alvimetal14} subsumes the gain function used in \eqref{eqn:maxi-expectedgain}. In particular, if we consider the alphabet $\hat{\mathcal{X}}$ to be equal to the simplex of probability distributions $P_{\hat{X}}$ on $\mathcal{X}$}, we can define a gain function $g(x,P_{\hat{X}})=g^\prime(P_{\hat{X}}(x))$ for a function $g^\prime:[0,1]\rightarrow \mathbb{R}$. However, our gain function viewpoint is conceptually different from that of Alvim~\emph{et al.}~\cite{Alvimetal12,Alvimetal14} in that the maximal expected gain in \eqref{eqn:maxi-expectedgain} is expressed in terms of the probability of correctly guessing. Moreover, as mentioned earlier in Section~\ref{subsection:related-work}, our approach to study information leakage measures differs substantially from that of Alvim~\emph{et al.}~\cite{Alvimetal12,Alvimetal14} (see Section~\ref{section:maximalg-leakage} for more details).}  We note that the notion of maximal expected gain in \eqref{eqn:maxi-expectedgain} for specific gain functions, $g(t)=t$ and $g(t)=\frac{\alpha}{\alpha-1}t^{\frac{\alpha-1}{\alpha}}$, $\alpha\in(1,\infty]$, {\blue{plays a crucial role}} in the definitions of maximal leakage~\cite{Issaetal} and maximal $\alpha$-leakage~\cite{Liaoetal} (see the bulleted list near \eqref{eqn:logloss-relent} for optimal guessing strategies for these gain functions in addition to the logarithmic gain function $g(t)=\log{t}$). 
\begin{definition}[Maximal $\alpha$-leakage~\cite{LiaoKS20,Liaoetal}]
Given a joint distribution $P_{XY}$ on a finite alphabet $\mathcal{X}\times\mathcal{Y}$, the maximal $\alpha$-leakage from $X$ to $Y$ is defined as, for $\alpha\in {\nblue{(0,1)}}\cup(1,\infty)$, 
\begin{align}\label{eqn:maxalpha-leakagedef}
&\mathcal{L}_\alpha^\emph{max}(X\rightarrow Y)\nonumber\\
&=\sup_{U-X-Y}\frac{\alpha}{\alpha-1}\log\frac{\max_{P_{\hat{U}|Y}}\mathbb{E}_{UY}\left[{\nblue{\frac{\alpha}{\alpha-1}}}P_{\hat{U}|Y}(U|Y)^{\frac{\alpha-1}{\alpha}}\right]}{\max_{P_{\hat{U}}}\mathbb{E}_U\left[{\nblue{\frac{\alpha}{\alpha-1}}}P_{\hat{U}}(U)^\frac{\alpha-1}{\alpha}\right]},
\end{align}
where $U$ represents any randomized function of $X$ that the adversary is interested in guessing and takes values in an arbitrary finite alphabet. Moreover, $\hat{U}$ is an estimator of $U$ with the same support as $U$.
\end{definition}
Maximal $\alpha$-leakage captures the information leaked about \emph{any function} of the random variable $X$ to an adversary that observes a correlated random variable $Y$. Maximal $\alpha$-leakage is proposed as a generalization of maximal leakage~\cite{Issaetal}, where the former recovers the latter when $\alpha\rightarrow \infty$, to allow tuning the measure to specific applications. The ratio inside the logarithm in \eqref{eqn:maxalpha-leakagedef} is the multiplicative increase, upon observing $Y$, of the maximal expected gain of an adversary in guessing a randomized function of $X$, with gain function $g(t)=\frac{\alpha}{\alpha-1}t^{\frac{\alpha-1}{\alpha}}$, $\alpha\in(0,1)\cup(1,\infty)$. In this work, we study maximal expected gain \eqref{eqn:maxi-expectedgain} and the associated leakage measures for arbitrary gain functions $g:[0,1]\rightarrow (0,\infty)$. Maximal leakage, maximal $\alpha$-leakage, and the new leakage measures we study in this work can be expressed in terms of Sibson mutual information~\cite{sibson1969information} and R\'{e}nyi divergence~\cite{renyi1961measures}.
\begin{definition}[Sibson mutual information of order $\alpha$~\cite{sibson1969information}]
For a given joint distribution $P_{XY}$ on finite alphabet $\mathcal{X}\times\mathcal{Y}$, the Sibson mutual information of order $\alpha\in(0,1)\cup(1,\infty)$ is 
\begin{align*}\label{eqn:sibson-def}
    I_\alpha^{\emph{S}}(X;Y)=\frac{\alpha}{\alpha-1}\log\sum_{y\in\mathcal{Y}}\left(\sum_{x\in\mathcal{X}}P_X(x)P_{Y|X}(y|{\blue{x}})^\alpha\right)^\frac{1}{\alpha}.
\end{align*}
It is defined by its continuous extension for $\alpha=1$ and $\alpha=\infty$, respectively, and is given by
\begin{align}
    I_1^\emph{S}(X;Y)&=I(X;Y) \ \emph{(Shannon mutual information)},\\
    I_\infty^\emph{S}(X;Y)&=\log\sum_{y\in\mathcal{Y}}\max_{x:P_X(x)>0}P_{Y|X}(y|x).
\end{align}
\end{definition}
\begin{definition}(R\'{e}nyi divergence of order $\alpha$~\cite{renyi1961measures})
The R\'{e}nyi divergence of order $\alpha\in(0,1)\cup(1,\infty)$ between two probability distributions $P_X$ and $Q_X$ on a finite alphabet $\mathcal{X}$ is defined as
\begin{align}\label{eqn:Renyidiv-def}
    D_\alpha(P_X||Q_X)=\frac{1}{\alpha-1}\log\left(\sum_{x\in\mathcal{X}}P_X(x)^\alpha Q_X(x)^{1-\alpha}\right).
\end{align}
It is defined by its continuous extension for $\alpha=1$ and $\alpha=\infty$, respectively, and is given by
\begin{align}
    D_1(P_X||Q_X)&=\sum_{x\in\mathcal{X}}P_X(x)\log\frac{P_X(x)}{Q_X(x)},\\
    D_\infty(P_X||Q_X)&=\max_{x\in\mathcal{X}}\log\frac{P_X(x)}{Q_X(x)}.
\end{align}
\end{definition}
{\blue{Notice that for $\alpha\geq 1$, $D_\alpha(P_X\|Q_X)$ is finite if and only if $P_X$ is absolutely continuous with respect to $Q_X$.}} Issa~\emph{et al.}~\cite{Issaetal} showed that maximal leakage is equal to Sibson mutual information of order $\infty$, i.e., \begin{align}
    \mathcal{L}_\infty^{\text{max}}(X\rightarrow Y)&=I_\infty^\text{S}(X;Y)\\
    &=\log\sum_{y\in\mathcal{Y}}\max_{x:P_X(x)>0}P_{Y|X}(y|x).
\end{align}
Generalizing this, Liao~\emph{et al.}~\cite{Liaoetal} proved that maximal $\alpha$-leakage is given by the Sibson mutual information of order $\alpha$, i.e.,
\begin{align}\label{eqn:maxalphaleak}
    \mathcal{L}_\alpha^{\text{max}}(X\rightarrow Y)=\sup_{P_{\tilde{X}}}I^{\text{S}}_{\blue{\alpha}}(\tilde{X};Y),
\end{align}
where the supremum is over all probability distributions $P_{\blue{\tilde{X}}}$ on the support of $P_X$.
{\blue\section{Maximal $g$-leakage Under Multiple Guesses}\label{section:maximalg-leakage}}
The definitions of maximal leakage and maximal $\alpha$-leakage consider the adversaries interested in maximizing the expected values of specific gain functions, in particular, maximal leakage uses the gain function $g(t)=t$ and maximal $\alpha$-leakage uses the gain function $g(t)=\frac{\alpha}{\alpha-1}t^{\frac{\alpha-1}{\alpha}}$, $\alpha\in(1,\infty)$. These leakage measures can be seen as special cases of a more general leakage measure incorporating an adversary interested in maximizing an arbitrary gain function $g$.

\begin{definition}[Maximal $g$-leakage]\label{def:maximal-g-leakage}
{\blue{Given a gain function $g:[0,1]\rightarrow [0,\infty)$ and a joint probability distribution $P_{XY}$ on a finite alphabet $\mathcal{X}\times\mathcal{Y}$, the maximal $g$-leakage is defined as}}
\begin{align}\label{eqn:max-g-leakage}
    &{\mathcal{L}}_g^{\emph{max}}(X\rightarrow Y)\nonumber\\
    =
&\sup_{U:U-X-Y}\log \frac{\sup_{P_{\hat{U}|Y}}\mathbb{E}_{UY}\left[g(P_{\hat{U}|Y}(U|Y))\right]}{\sup_{P_{\hat{U}}}\mathbb{E}_U\left[g(P_{\hat{U}}(U))\right]}.
\end{align}
\end{definition}

{\blue {Maximal $g$-leakage captures how much information an adversary can learn about any randomized function of a random variable $X$ from a correlated random variable $Y$ when a single guess is allowed. We now define \emph{maximal $g$-leakage under $k$ guesses} which captures the information an adversary can learn when $k$ guesses are allowed. This definition is also related to maximal leakage under $k$ guesses \cite[Definition~4]{Issaetal}.}}
{\blue\begin{definition}[Maximal $g$-leakage under $k$ guesses]
Given a gain function $g:[0,1]\rightarrow [0,\infty)$ and a joint probability distribution $P_{XY}$, the maximal $g$-leakage from $X$ to $Y$ under $k$ guesses is defined as
\begin{align}\label{defn:kgleakage}
    &\mathcal{L}^{(k)-\emph{max}}_g(X\rightarrow Y)\nonumber\\
    &= \sup_{U:U-X-Y}\log{\frac{\max\limits_{P_{\hat{U}_{[1:k]}|Y}}\mathbb{E}\left[g\left(\mathrm{P}\left(\bigcup\limits_{i=1}^k(\hat{U}_i=U)|Y\right)\right)\right]}{\max\limits_{P_{\hat{U}_{[1:k]}}}\mathbb{E}\left[g\left(\mathrm{P}\left(\bigcup\limits_{i=1}^k(\hat{U}_i=U)\right)\right)\right]}},
\end{align}
where $\hat{U}_1,\hat{U}_2,\dots,\hat{U}_k$ represent $k$ estimators of $U$ with the same support as $U$.
\end{definition}
}

{\blue {We interpret $\mathrm{P}(\bigcup\limits_{i=1}^k(\hat{U}_i=u)|Y=y)$ as the probability of correctly estimating $U=u$ given $Y=y$ in $k$ guesses. Due to the fact that gain function used in maximal $g$-leakage is a special case of the gain function of Alvim \emph{et al.}~\cite{Alvimetal12} (as discussed in the paragraph after Definition~\ref{defn:gain}), an upper bound on maximal $g$-leakage directly follows from Issa \emph{et al.}~\cite[Theorem~5]{Issaetal}. In particular, Issa~\emph{et al.}~\cite[Theorem~5]{Issaetal} showed that, for a given joint distribution $P_{XY}$,
    \begin{align}\label{remark1-issa}
    \sup_{\substack{U:U-X-Y\\ \hat{\mathcal{U}},g:\mathcal{U}\times\hat{\mathcal{U}}\rightarrow [0,\infty):\\ \sup_{\hat{u}\in\hat{\mathcal{U}}}\mathbb{E}_U[g(U,\hat{u})]>0}}\log{\frac{\sup_{\hat{u}(\cdot)}\mathbb{E}_{UY}[g(U,\hat{u}(Y))]}{\sup_{\hat{u}\in\hat{\mathcal{U}}}\mathbb{E}_U[g(U,\hat{u})]}}=I_\infty^\text{S}(X;Y).
    \end{align}
    Thus, it follows from from \eqref{remark1-issa} and Issa~\emph{et~al.}~\cite[Theorem~1]{Issaetal} that
\begin{align}\label{eqn:supachvg}
    \mathcal{L}^{\text{max}}_g(X\rightarrow Y)\leq I_\infty^\text{S}(X;Y)
\end{align}
with equality if $g$ is the identity gain function, $g(t)=t$, for $t\in[0,1]$.
}}  

{\blue{In the following proposition, we show that the upper bound in \eqref{eqn:supachvg} holds for maximal $g$-leakage under $k$ guesses also.

\begin{proposition}[Upper bound on maximal $g$-leakage under $k$ guesses]\label{prop:upperboundg}
   Given a gain function $g:[0,1]\rightarrow [0,\infty)$ and a joint probability distribution $P_{XY}$, we have that maximal leakage is an upper bound on maximal $g$-leakage under $k$ guesses, i.e.,  
   \begin{align}\label{eqn:newmultipleguess}
        \mathcal{L}^{(k)-{\nblue{\emph{max}}}}_g(X\rightarrow Y)\leq I_{\infty}^{\emph{S}}(X;Y),
   \end{align}
   with equality if $g$ is the identity gain function, $g(t)=t$, for $t\in[0,1]$
\end{proposition}
\begin{remark}
    Proposition~\ref{prop:upperboundg} follows from \eqref{remark1-issa} together with an observation that the gain function definition used by Alvim~\emph{et~al.}~\cite{Alvimetal12} (and Issa~\emph{et~al.}~\cite[Theorem~5]{Issaetal}) captures not only our gain function definition but also the multiple guesses scenario ~\cite[Section~III-C]{Alvimetal14} simultaneously. In particular, let $\mathcal{\hat{U}}$ be the simplex of all probability distributions $P_{\hat{U}{_{[1:k]}}}$ with each $U_i$ taking values in $\mathcal{U}$ and define 
    \begin{align}\label{eqn:alvimmultguesse}
    g^\prime(u,P_{\hat{U}{_{[1:k]}}}):=g\left(\mathrm{P}\left(\bigcup\limits_{i=1}^k(\hat{U}_i=u)\right)\right),
    \end{align}
   for $u\in\mathcal{U}$, $P_{\hat{U}{_{[1:k]}}}\in\mathcal{\hat{U}}$. Then, Proposition~\ref{prop:upperboundg} follows from \eqref{remark1-issa} together with \eqref{eqn:alvimmultguesse}. We remark that this observation also provides a simpler proof to the upper bound part in \cite[Proof~of~Theorem~4]{Issaetal}, i.e., maximal leakage upper bounds maximal leakage under $k$ guesses.
\end{remark}}}
{\blue{Proposition~\ref{prop:upperboundg} shows that maximal $g$-leakage under multiple guesses is maximized when $g(t)=t$, for $t\in[0,1]$, which corresponds to maximal leakage under multiple guesses~\cite[Theorem~4]{Issaetal}.}} Also, we note that closed-form expressions for even maximal $g$-leakage are known only for few gain functions, in particular, when $g(t)=\frac{\alpha}{\alpha-1}t^{\frac{\alpha-1}{\alpha}}$, $\alpha\in(0,1)(1,\infty)$~\cite{Issaetal,LiaoKS20,Liaoetal}, that corresponds to maximal $\alpha$-leakage (and maximal leakage when $\alpha\rightarrow \infty$). In the following theorem, we obtain closed-form expression for maximal $g$-leakage {\blue{under multiple guesses}} for a class of concave gain functions.
\begin{theorem}[{\blue{Maximal $g$-leakage under multiple guesses}}]\label{thm:maximalg-leakage-concaveg}
{\blue{Let $P_{XY}$ be a  joint probability distribution on a finite alphabet $\mathcal{X}\times\mathcal{Y}$ and $g:[0,1]\rightarrow [0,\infty)$ be a gain function satisfying the following assumptions:}}
\begin{itemize}
    \item $g$ is a concave function,
    \item $g(0)=0$ and $0<g^\prime(0)<\infty$.
    \end{itemize}
    Then we have that maximal $g$-leakage under $k\geq 1$ guesses is exactly equal to maximal leakage, i.e.,
    \begin{align*}
     \mathcal{L}_g^{(k)-\emph{max}}(X\rightarrow Y)&=I_\infty^\emph{S}(X;Y).\\
    \end{align*}
\end{theorem}
{\nblue\begin{remark}
    An important consequence of Theorem~\ref{thm:maximalg-leakage-concaveg} in the design of privacy mechanisms is the following. Suppose the inferential capability of an adversary in guessing about $X$ from $Y$ is measured via a function $g$ of the probability of correctly guessing and we use maximal $g$-leakage as the privacy measure. The system designer needs to find an optimal privacy mechanism $P_{Y|X}$ minimizing the leakage $\mathcal{L}_g(X\rightarrow Y)$ while maintaining a minimum level of utility measured by, say, $\mathcal{U}(X,Y)$. In many practical situations, we may have only a limited knowledge about the inferential capability of adversary. Specifically, assume that all we know about the function $g$ is that it belongs to a class of non-negative concave gain functions $g$ that satisfy $g(0)=0$ with a finite positive derivative at $0$ (see the paragraph below Remark~\ref{remark:4} for an interpretation of this class). Now, as a result of Theorem~\ref{thm:maximalg-leakage-concaveg}, it suffices for the system designer to find the optimal mechanism that minimizes the Sibson mutual information of order infinity subject to utility constraints without worrying about further details of $g$, i.e.,
    \begin{align}\label{eqn:mechanism}
        \inf_{P_{Y|X}:\mathcal{U}(X,Y)\geq k}I_{\infty}^\emph{S}(X;Y).
    \end{align}
    Thus, notice that the effective situation is same as that of knowing the exact form of the gain function $g$ as the optimal privacy mechanism with a particular $g$ remains the same as that of \eqref{eqn:mechanism} as long as $g$ belongs to the aforementioned class.   
\end{remark}
}
\begin{remark}\label{distinction-alvimissa}
    {\blue{In this remark, using Theorem~\ref{thm:maximalg-leakage-concaveg} we outline the distinction between maximal $g$-leakage and the leakage measures of Alvim~\emph{et al.}~\cite{Alvimetal12} and Issa \emph{et al.}~\cite{Issaetal} based on gain functions. Alvim~\emph{et~al.}~\cite[Definition~3.4]{Alvimetal12} define $g$-capacity of a channel $P_{Y|X}$ as
    \begin{align}
        \mathcal{ML}_g(X\rightarrow Y):=\sup_{P_X}\log{\frac{\sup_{\hat{x}(\cdot)}\mathbb{E}_{XY}[g(X,\hat{x}(Y))]}{\sup_{\hat{x}\in\hat{\mathcal{X}}}\mathbb{E}_X[g(X,\hat{x})]}},
    \end{align}
    for any $g:\mathcal{X}\times\hat{\mathcal{X}}\rightarrow [0,\infty)$,
    and showed that \cite[Theorem~5.1]{Alvimetal12}
    \begin{align}\label{remark1-alvim}
        \sup_{\substack{\hat{\mathcal{X}},g:\mathcal{X}\times\hat{\mathcal{X}}\rightarrow [0,\infty):\\ \sup_{\hat{x}\in\hat{\mathcal{X}}}\mathbb{E}_X[g(X,\hat{x})]>0}}\!\!\!\mathcal{ML}_g(X\rightarrow Y)=\log\sum_{y\in\mathcal{Y}}\max_{x\in\mathcal{X}}P_{Y|X}(y|x),
    \end{align} 
    which considers the worst-case scenario over all $g$ (note that the expression in RHS of \eqref{remark1-alvim} is equal to that of \eqref{remark1-issa} if $X$ has a full support). Thus, Alvim~\emph{et~al.}~\cite{Alvimetal12} and Issa~\emph{et~al.}~\cite{Issaetal} (see \eqref{remark1-issa}) obtained conclusive results for leakage measures with a maximum taken over all the gain functions rather than focusing on leakage measures with specific gain functions.
    Theorem~\ref{thm:maximalg-leakage-concaveg} provides closed-form expression for maximal $g$-leakage for a specific class of concave gain functions.}}
\end{remark}
\begin{remark}\label{remark:4}
Note that the gain function that corresponds to maximal $\alpha$-leakage (for $\alpha\in(0,1)\cup(1,\infty)$), i.e., $g(t)=\frac{\alpha}{\alpha-1}t^{\frac{\alpha-1}{\alpha}}$, does not belong to the class of gain functions for which Theorem~\ref{thm:maximalg-leakage-concaveg} holds even though the gain function is concave; this is because $g^\prime(0)=\infty$. However, as $\alpha\rightarrow \infty$, we have $g(t)=t$ for which $g^\prime(0)=1<\infty$. So, Theorem~\ref{thm:maximalg-leakage-concaveg} recovers \cite[Theorem~1]{Issaetal} when $g(t)=t$. Some other examples of gain functions for which Theorem~\ref{thm:maximalg-leakage-concaveg} holds are $g(t)=1-(1-t)^2, \sin{t}, \min\{ct,\frac{1}{2}\}$, for some constant $c>0$.
\end{remark}
{\blue{To interpret the conditions on $g$ in Theorem~\ref{thm:maximalg-leakage-concaveg}, we first note that the concavity of $g$ is natural in that we are examining the optimization problems involving maximization of gain functions in the definition of maximal $g$-leakage. The condition $g(0)=0$ suggests that the minimum possible value of the gain function be assigned to adversary when the probability of correctly guessing is zero. When the minimum possible value has been assigned to $g(0)$, it has to be the case that $g$ is increasing at $0$, in fact, we need $g$ to be strictly increasing at $0$ with a finite derivative in Theorem~\ref{thm:maximalg-leakage-concaveg}. Extending this intuition, we may impose an additional constraint that the maximum possible value of the gain function be assigned when the probability of correctly guessing is $1$, i.e., $g(1)=\sup_{t\in[0,1]}$g(t) (though it is not required for Theorem~\ref{thm:maximalg-leakage-concaveg}). Then because of concavity, the function $g$ needs to be non-decreasing in the probability of correctly guessing\footnote{\blue{To see this, suppose that there exist $x,y$ such that $x<y$ and $g(x)>g(y)$. Since $y\in[x,1]$, there exists $\lambda\in[0,1]$ with $y=\lambda x+(1-\lambda)$. Then, concavity of $g$ implies that $g(y)=g(\lambda x+(1-\lambda)\geq \lambda g(x)+(1-\lambda)g(1)>\lambda g(y)+(1-\lambda)g(1)\geq \lambda g(y)+(1-\lambda)g(y)=g(y)$, which is a contradiction.}}, which may be more relevant in practical scenarios.}}
A detailed proof of Theorem~\ref{thm:maximalg-leakage-concaveg} is in Appendix~\ref{thm:maxgleakage2}. {\blue{A consequence of Theorem~\ref{thm:maximalg-leakage-concaveg}} is that the maximum in \eqref{eqn:newmultipleguess} is achieved by a class of concave gain functions with identity gain function (that corresponds to maximal leakage) being one of them.}

{\nblue Notice that maximal $g$-leakage in Theorem~\ref{thm:maximalg-leakage-concaveg} and maximal $\alpha$-leakage depends on $P_X$ only through its support. The following theorem presents a closed-form expression for maximal $g$-leakage that corresponds to the gain function $g(t)=1+t$, $t\in[0,1]$, for binary $X$ and $Y$. Interestingly, it turns out that this expression is fundamentally different in its structure from that of maximal leakage and maximal $\alpha$-leakage. In particular, this leakage explicitly depends on $P_X$ (i.e., not just through its support).

\begin{theorem}\label{thm:shiftedmaxL}
    Let $P_{XY}$ be a joint probability distribution on $\mathcal{X}\times\mathcal{Y}$ with $|\mathcal{X}|=|\mathcal{Y}|=2$. Then, for $g(t)=1+t$, $t\in[0,1]$, we have
    \begin{align}
       \mathcal{L}_g^{\emph{max}}(X\rightarrow Y)=\log{\frac{1+p^*\sum_{y\in\mathcal{Y}}\max_{x\in\mathcal{X}}P_{Y|X}(y|x)}{1+p^*}},
    \end{align}
    where $p^*=\min_{x:P_X(x)>0}P_X(x)$.
\end{theorem}
A detailed proof of Theorem~\ref{thm:shiftedmaxL} is in Appendix~\ref{proof of thm2}.
}

\section{Evaluating Multiple Guesses via a Tunable Loss Function}\label{section:multiple-guesses}
{\blue{In this section, we focus on a specific class of gain functions, related to the $\alpha$-loss~\cite{arimoto1971information,LiaoKS20,SypherdDSK19}, parameterized by $\alpha\in(0,1)\cup(1,\infty)$ and the corresponding maximal $g$-leakage measures under multiple guesses. We first characterize the minimal expected $\alpha$-loss (or equivalently, the associated maximal expected gain) under multiples guesses, and then study maximal $g$-leakage for the corresponding class of gain functions parameterized by $\alpha\in(0,1)\cup(1,\infty)$.}}

Consider a setup where an adversary is interested in guessing the unknown value of a random variable $X$ on observing another correlated random variable $Y$, where $X$ and $Y$ are jointly distributed according to $P_{XY}$ over the finite support $\mathcal{X}\times\mathcal{Y}$. The adversary can make a fixed number of guesses, say $k$, to estimate $X$. We focus on evaluating the adversary's success using loss functions that in turn can measure the information leaked by $Y$ about $X$. 
To this end, we model the adversary's strategy using $\alpha$-loss, a class of tunable loss functions parameterized by $\alpha \in (0,\infty]$ ~\cite{LiaoKS20,SypherdDSK19}. This class captures the well-known exponential loss ($\alpha=1/2$)~\cite{FREUND1997119}, log-loss ($\alpha=1$)~\cite{MerhavF1998,NguyenWJ09,CourtadeW11}, and the 0-1 loss ($\alpha=\infty$)~\cite{NguyenWJ09,BartlettJM06}. The adversary then seeks to find the optimal (possibly randomized) guessing strategy that minimizes the expected $\alpha$-loss over $k$ guesses. We first review $\alpha$-loss and then define maximal expected $\alpha$-loss under $k$ guesses.
\begin{definition}[$\alpha$-loss~\cite{arimoto1971information,LiaoKS20,SypherdDSK19}]\label{defn:alphaloss}
For $\alpha\in(0,1)\cup(1,\infty)$, the $\alpha$-loss is a function defined from $[0,1]$ to $\mathbb{R}_+$ as
\begin{align}
    \ell_\alpha(p):=\frac{\alpha}{\alpha-1}\left(1-p^{\frac{\alpha-1}{\alpha}}\right).
\end{align}
It is defined by continuous extension for $\alpha=1$ and $\alpha=\infty$, respectively, and is given by
\begin{align}
    \ell_1(p)=\log{\frac{1}{p}},\  \ell_\infty(p)=1-p.
\end{align}
\end{definition}
Notice that $\ell_\alpha(p)$ is decreasing in $p$.
\begin{definition}[Minimal expected $\alpha$-loss under $k$ guesses]
Consider random variables $(X,Y)\sim P_{XY}$ and an adversary that makes $k$ guesses $\hat{X}_{[1:k]}=\hat{X}_1,\hat{X}_2,\dots,\hat{X_k}$ 
on observing $Y$ such that $X-Y-\hat{X}_{[1:k]}$ is a Markov chain. Let $P_{\hat{X}_{[1:k]}|Y}$ be a strategy for estimating $X$ from $Y$ in $k$ guesses. For $\alpha\in(0,\infty],$ the minimal expected $\alpha$-loss under $k$ guesses is defined as
\begin{align}\label{minimallossdef} 
& \mathcal{ME}^{(k)}_\alpha(P_{XY})\nonumber\\
& :=\min_{P_{\hat{X}_{[1:k]}|Y}}
    \sum_{x,y} P_{XY}(x,y) \ell_\alpha\left(\mathrm{P}\left(\bigcup_{i=1}^k (\hat{X}_i=x)|Y=y\right)\right).  
\end{align}
\end{definition}
$\mathrm{P}(\bigcup\limits_{i=1}^k(\hat{X}_i=x)|Y=y)$ is the probability of correctly estimating $X=x$ given $Y=y$ in $k$ guesses. An adversary seeks to find the optimal guessing strategy in \eqref{minimallossdef}.
Note that the optimization problem in \eqref{minimallossdef} was solved for a special case of $k=1$ by Liao \textit{et al.}~\cite[Lemma~1]{LiaoKS20}.
Notice that
\begin{align}\label{eqn:simpli}
    \mathcal{ME}^{(k)}_\alpha(P_{XY})=\sum_yP_Y(y)\mathcal{ME}_\alpha^{(k)}(P_{X|Y=y}),
\end{align}
where we have slightly abused the notation in the R.H.S. of \eqref{eqn:simpli}. Hence, in view of \eqref{eqn:simpli}, in order to solve the optimization problem in \eqref{minimallossdef}, it suffices to solve for a case where $Y=\emptyset$, i.e., 
\begin{equation}\label{eqn:optimizationproblem}
   \mathcal{ME}_\alpha^{(k)}(P_X):= \min_{P_{\hat{X}_{[1:k]}}}
    \sum_x P_X(x) \ell_\alpha\left(P\left(\bigcup_{i=1}^k (\hat{X}_i=x)\right)\right).
\end{equation}
 Also, in the sequel, it suffices to consider the optimization problem in \eqref{eqn:optimizationproblem} only for the case where $k<n$, where $P_X$ is supported on $\mathcal{X}=\{x_1,x_2,\dots,x_n\}$ because if $k\geq n$, we have
 $\mathcal{ME}_\alpha^{(k)}(P_X)=0$,
since a strategy $P^*_{\hat{X}_{[1:k]}}$ such that $P^*_{\hat{X}_{[1:n]}}(x_1,x_2,\dots,x_n)=1$ is optimal.
We review some simple special cases of \eqref{eqn:optimizationproblem} that are well known in the literature. 
\begin{itemize}
\item Consider \eqref{eqn:optimizationproblem} for the case of log-loss ($\alpha=1$) and $k=1$~\cite{CourtadeW11}. It is well known that the expected log-loss can be expanded as
\begin{align}\label{eqn:logloss-relent}
    \mathbb{E}\left[\log{\frac{1}{P_{\hat{X}}(X)}}\right]=H(X)+D(P_X\|P_{\hat{X}}),
\end{align}
which implies that the optimal guessing strategy is equal to the original distribution, $P_X$, and the minimal expected log-loss is given by its entropy, $H(X)$. Moreover, the relative entropy $D(P_{X}\|P_{\hat{X}})$ quantifies the relative performance of an arbitrary adversarial guessing strategy $P_{\hat{X}}$ with respect to the optimal guessing strategy $P_X$. To the best of our knowledge, minimal expected loss under multiple guesses was not explored even for log-loss ($\alpha=1$) earlier. 
\item At the other extreme, Issa \emph{et al.}~\cite{Issaetal} studied maximal expected probability of correctness in a fixed number of guesses which corresponds to minimal expected $\alpha$-loss for the special case of $\alpha=\infty$. The optimal guessing strategy here is to guess the $k$ most likely outcomes according to the distribution $P_X$. 
\item Liao \emph{et al.}~\cite[Lemma~1]{Liaoetal} solved the optimization problem in \eqref{eqn:optimizationproblem} for a special case of $k=1$ and for arbitrary $\alpha\in(0,\infty]$, where they showed that the optimal guessing strategy is a tilted probability distribution given by $P_X^{(\alpha)}(x):=\frac{P_X(x)^\alpha}{\sum\limits_xP_X(x)^\alpha}$. 
\end{itemize}
We completely characterize the minimal expected $\alpha$-loss under $k$ guesses, for $\alpha\in(0,\infty]$. We first express the objective function in \eqref{eqn:optimizationproblem}, i.e., the expected $\alpha$-loss, in a way similar to \eqref{eqn:logloss-relent} where it turns out that Bregman divergence~\cite{BREGMAN1967200}, a generalization of relative entropy, naturally arises. {\blue{For a continuously-differentiable convex function $F:\Omega\rightarrow \mathbb{R}$, the associated Bregman divergence between $p$ and $q$ in $\Omega$ is defined as $B_{F}(p,q) = F(p) - F(q) - \langle \nabla F(q), p-q \rangle$.}}

\begin{lemma}[Expected $\alpha$-loss under $k$ guesses]\label{thm:expected-alpha-loss-singleguess}
For a fixed probability distribution $P_X$ and an arbitrary guessing strategy $P_{\hat{X}_{[1:k]}}$, we have
\begin{align}\label{eqn:singleguess-altform}
 &\mathbb{E}_{X}\left[\ell_\alpha(\mathrm{P}(\cup_{i=1}^k (\hat{X}_i=X)))\right]\nonumber\\
 &=\frac{\alpha}{\alpha-1} \left(1-k^{\frac{\alpha-1}{\alpha}}\mathrm{e}^{\frac{1-\alpha}{\alpha}H_\alpha(X)}\right)+k^{\frac{\alpha-1}{\alpha}}B_F(P_X,P^{(\frac{1}{\alpha})}_{\blue\tilde{X}}),
\end{align}
where $P_{\tilde{X}}(x):=\frac{\mathrm{P}(\cup_{i=1}^k (\hat{X}_i=x))}{k}$, $P_{\tilde{X}}^{(\frac{1}{\alpha})}(x):=\frac{P_{\tilde{X}}(x)^{\frac{1}{\alpha}}}{\sum\limits_xP_{\tilde{X}}(x)^{\frac{1}{\alpha}}}$, and $B_F(\cdot,\cdot)$ is the Bregman divergence associated with the function $F(P_X)=\frac{\alpha}{\alpha-1}\left((\sum_xP_X(x)^\alpha)^{\frac{1}{\alpha}}-1\right)$ {\blue given by $B_{F}(P_X,P_{\tilde{X}}^{(\frac{1}{\alpha})}) = k^{\frac{\alpha-1}{\alpha}}\left(\sum_xP_X(x)^\alpha\right)^{\frac{1}{\alpha}}\left(1-\mathrm{e}^{\frac{1-\alpha}{\alpha}D_\frac{1}{\alpha}(P^{(\alpha)}_X\|P_{\tilde{X}})}\right)$.} Moreover, the minimal expected $\alpha$-loss is given by $ \mathcal{ME}_\alpha^{(k)}(P_X)=\frac{\alpha}{\alpha-1}(1-k^{\frac{\alpha-1}{\alpha}}\mathrm{e}^{\frac{1-\alpha}{\alpha}H_\alpha(X)})$ if and only if $P_X^{(\alpha)}(x)\leq \frac{1}{k}$, for all $x\in\mathcal{X}$.
\end{lemma}
A detailed proof of Lemma~\ref{thm:expected-alpha-loss-singleguess} is in Appendix~\ref{proofofthm1}. Note that $P_{\tilde{X}}(x)$ defined in Lemma~\ref{thm:expected-alpha-loss-singleguess} may not be even a probability distribution from the way it is defined since $\sum_x\mathrm{P}(\cup_{i=1}^k (\hat{X}_i=x))\leq k$, in general. However, as we prove later (see Lemma~\ref{fact:alphaloss2noneq} and Remark~\ref{remark}), it suffices to consider the guessing strategies $P_{\hat{X}_{[1:k]}}$ that satisfy the condition $\sum_x\mathrm{P}(\cup_{i=1}^k (\hat{X}_i=x))= k$ in order to solve the optimization problem in \eqref{eqn:optimizationproblem}, thereby making it sufficient to limit $P_{\tilde{X}}$ to be a probability distribution. Note that Lemma~\ref{thm:expected-alpha-loss-singleguess} characterizes the minimal expected $\alpha$-loss under $k$ guesses only for the case when $P_{X}^{(\alpha)}(x)\leq \frac{1}{k}$, for all $x\in\mathcal{X}$. This condition is trivially satisfied by any $P_X$ when $k=1$. {\blue Thus, when $P_{X}^{(\alpha)}(x)\leq \frac{1}{k}$, for all $x\in\mathcal{X}$, Bregman divergence $B_F$ (with the function $F$ in Lemma~\ref{thm:expected-alpha-loss-singleguess}) captures the relative performance of an arbitrary guessing strategy with respect to an optimal strategy in minimizing the expected $\alpha$-loss. This generalizes the observation below \eqref{eqn:logloss-relent} that relative entropy quantifies the relative performance of an arbitrary adversarial guessing strategy with respect to the optimal guessing strategy in minimizing the expected log-loss.} We completely characterize the minimal expected $\alpha$-loss under $k$ guesses in the next theorem, i.e., comprising the case when $P_{X}^{(\alpha)}(x)>\frac{1}{k}$, for some $x\in\mathcal{X}$ also. 
\begin{theorem}[Minimal expected $\alpha$-loss under $k$ guesses]\label{thm:minloss-alpha-kguesses}
Let $P_X$ be a probability distribution supported on $\mathcal{X}=\{x_1,x_2,\dots,x_n\}$ such that $p_1\geq p_2\geq \dots \geq p_n$, where $p_i:=P_X(x_i)$, for $i\in[1:n]$. Then the minimal expected $\alpha$-loss under $k$ guesses is given by
\begin{align}\label{eqn:thm11}
    \!\!\mathcal{ME}^{(k)}_\alpha(P_X)=\! \frac{\alpha}{\alpha-1}\sum\limits_{i=s^*}^np_i\bigg(1-\bigg(\frac{(k-s^*+1)p_i^\alpha}{\sum_{j=s^*}^np_j^\alpha}\bigg)^\frac{\alpha-1}{\alpha}\bigg),
\end{align}
where 
\begin{align}\label{eqn:thm1sstar}
    s^*=\min\left\{r\in\{1,2,\ldots,k\}: \frac{(k-r+1)p_r^\alpha}{\sum_{i=r}^np_i^\alpha}\leq 1\right\}.
\end{align}
\end{theorem}
\begin{remark}
Theorem~\ref{thm:minloss-alpha-kguesses} implies that the minimal expected $\alpha$-loss under $k$ guesses induces a polychotomy on the simplex of probability distributions $P_X$ on $\mathcal{X}$ depending on the value of the parameter $s^*$. Also, in an optimal guessing strategy, the adversary always guesses each of the $s^*-1$ most likely outcomes in one of the $k$ guesses with probability $1$, i.e., $\mathrm{P}(\cup_{j=1}^k (\hat{X}_j=x_i))=1$, for $i\in[1:s^*-1]$, and guesses the remaining outcomes with a probability proportional their tilted probability values, i.e., $\mathrm{P}(\cup_{j=1}^k (\hat{X}_j=x_i))=\frac{(k-s^*+1)p_i^\alpha}{\sum_{j=s^*}^np_j^\alpha}$, for $i\in[s^*:n]$. It may not be immediately clear if there indeed exists an optimal guessing strategy $P_{\hat{X}_{[1:k]}}$ consistent with these value assignments to $\mathrm{P}(\cup_{j=1}^k (\hat{X}_j=x))$, $x\in\mathcal{X}$.  However, as we show in the proof of Theorem~\ref{thm:minloss-alpha-kguesses} (see Lemma~\ref{theorem:probclass1} and the discussion above it), it follows from Farkas' lemma~\cite[Proposition~6.4.3]{Matousek} that an arbitrary value assignment to $\mathrm{P}(\cup_{j=1}^k (\hat{X}_j=x))$, for each $x\in\mathcal{X}$, will in fact guarantee the existence of a consistent probability distribution $P_{\hat{X}_{[1:k]}}$ as long as the assignment is such that $\sum_{x\in\mathcal{X}}\mathrm{P}(\cup_{j=1}^k (\hat{X}_j=x))=k$, which is true for the case above. For the special case when $k=s^*=2$, this optimal strategy is exactly the same as that of a seemingly different guessing problem considered in \cite[Section~II-B]{HuleihelSM17}.
\end{remark}
\begin{remark}
Notice that whenever $s^*=1$ in \eqref{eqn:thm1sstar}, the expression in \eqref{eqn:thm11} simplifies to
\begin{align}\label{eqn:thm1recover}
     \frac{\alpha}{\alpha-1}\left(1-k^{\frac{\alpha-1}{\alpha}}\mathrm{e}^{\frac{1-\alpha}{\alpha}H_\alpha(X)}\right),
\end{align}
where $H_{\alpha}(X)=\frac{1}{1-\alpha}\log{\left(\sum_{i=1}^np_i^\alpha\right)}$ is the R\'{e}nyi entropy of order $\alpha$~\cite{renyi1961measures}. Also, note that for the special case of $k=1$, we always have $s^*=1$, thereby recovering \cite[Lemma~1]{LiaoKS20}.
\end{remark}
A detailed proof of Theorem~\ref{thm:minloss-alpha-kguesses} is in Appendix~\ref{proofofthm2}. The proof builds up on  a careful decomposition of the simplex of probability distributions $P_X$ and then solving the desired optimization problem in each case using tools from convex optimization. We present a simpler and a more descriptive proof for the special case of log-loss ($\alpha=1$) and $k=2$ in Appendix~\ref{logloss2}. We remark that the proof of this special case does not directly generalize to arbitrary $\alpha\in(0,\infty]$ (even for the case of $k=2$). However, we remark that it essentially captures the intuition and the main tools involved in the relatively complex analysis in the proof for the more general case with $\alpha\in(0,\infty]$ and $k\in\mathbb{N}$.

We illustrate the optimal guessing strategies to achieve the minimal expected $\alpha$-loss in \eqref{eqn:thm11} through Examples~\ref{Example1} and \ref{Example2}.
\begin{example}[Optimal guessing strategy with $k=2$ guesses]\label{Example1}
Consider a probability distribution $P_X$ supported on $\mathcal{X}=\{x_1,x_2,x_3\}$. Let $p_i=P_X(x_i)$ and $p_i^{(\alpha)}=\frac{p_i^\alpha}{\sum_{j=1}^3p_j^\alpha}$, for $i\in[1:3]$ and fix $\alpha=2$. Optimal guessing strategy depends on the value of the parameter $s^*$ in \eqref{eqn:thm1sstar}. We present an optimal guessing strategy for each case specified by $s^*\in[1:2]$. Notice that $\frac{p_r^\alpha}{\sum_{i=r}^np_i^\alpha}=\frac{p_r^{(\alpha)}}{\sum_{i=r}^np_i^{(\alpha)}}$.

\noindent\underline{With $s^*=1$}: Let $p_1=\frac{3}{8}$, $p_2=\frac{3}{8}$, and $p_3=\frac{1}{4}$. This gives us $p_1^{(\alpha)}=\frac{9}{22}$, $p_2^{(\alpha)}=\frac{9}{22}$, and $p_3^{(\alpha)}=\frac{2}{11}$. It can be verified that $s^*=1$ for this $P_X$. In particular, we have $p_i^{(\alpha)}\leq \frac{1}{2}$, for all $i\in[1:3]$. From \eqref{eqn:thm11}, an optimal guessing strategy $P^*_{\hat{X}_1\hat{X}_2}$ is such that the adversary always guesses $x_i$ in one of the two guesses with a probability proportional to $p_i^{(\alpha)}$, i.e., $\mathrm{P}^*(\hat{X}_1=x_1\ \text{or}\ \hat{X}_2=x_1)=2p_1^{(\alpha)}=\frac{9}{11}$, $\mathrm{P}^*(\hat{X}_1=x_2\ \text{or}\ \hat{X}_2=x_2)=2p_2^{(\alpha)}=\frac{9}{11}$, and $\mathrm{P}^*(\hat{X}_1=x_3\ \text{or}\ \hat{X}_2=x_3)=2p_3^{(\alpha)}=\frac{4}{11}$.\\
\noindent\underline{With $s^*=2$}: Let $p_1=\frac{2}{3}$, $p_2=\frac{1}{4}$, and $p_3=\frac{1}{12}$. This gives us $p_1^{(\alpha)}=\frac{32}{37}$, $p_2^{(\alpha)}=\frac{9}{74}$, and $p_3^{(\alpha)}=\frac{1}{74}$. It can be verified that $s^*=2$ for this $P_X$. In particular, we have $p_1^{(\alpha)}>\frac{1}{2}$. From \eqref{eqn:thm11}, an optimal guessing strategy $P^*_{\hat{X}_1\hat{X}_2}$ is such that the adversary always guesses $x_1$ in one of the two guesses, i.e., $\mathrm{P}^*(\hat{X}_1=x_1\ \text{or}\ \hat{X}_2=x_1)=1$, and guesses $x_i$ with a probability proportional to $p_i^{(\alpha)}$, for $i\in[2:3]$, i.e., $\mathrm{P}^*(\hat{X}_1=x_2\ \text{or}\ \hat{X}_2=x_2)=\frac{p_2^{(\alpha)}}{\sum_{j=2}^3p_j^{(\alpha)}}=\frac{9}{10}$ and $\mathrm{P}^*(\hat{X}_1=x_3\ \text{or}\ \hat{X}_2=x_3)=\frac{p_3^{(\alpha)}}{\sum_{j=2}^3p_j^{(\alpha)}}=\frac{1}{10}$.
\end{example}
\begin{example}[Optimal guessing strategy with $k=3$ guesses]\label{Example2}
Let $P_X$ be supported on $\mathcal{X}=\{x_1,x_2,x_3,x_4\}$. Let $p_i=P_X(x_i)$ and $p_i^{(\alpha)}=\frac{p_i^\alpha}{\sum_{j=1}^4p_j^\alpha}$, for $i\in[1:4]$ and fix $\alpha=2$. We present an optimal guessing strategy for each case specified by $s^*\in[1:3]$.

\noindent\underline{With $s^*=1$}: Notice that specifying a tilted distribution $P_X^{(\alpha)}$ uniquely determines the original distribution $P_X$. Let $P_X$ be such that $p_1^{(\alpha)}=p_2^{(\alpha)}=\frac{1}{4}$, $p_3^{(\alpha)}=\frac{1}{5}$, and $p_4^{(\alpha)}=\frac{3}{10}$. It can be verified that $s^*=1$ for this $P_X$. In particular, we have $p_1^{(\alpha)}\leq\frac{1}{3}$, for all $i\in[1:4]$. From \eqref{eqn:thm11}, an optimal guessing strategy $P^*_{\hat{X}_1\hat{X}_2\hat{X}_3}$ is such that the adversary guesses $x_i$ in one of the three guesses with a probability proportional to $p_i^{(\alpha)}$, for $i\in[1:3]$, i.e., $\mathrm{P}^*(\cup_{j=1}^{{\blue 3}}(\hat{X}_j=x_1))={\blue 3}p_1^{(\alpha)}=\frac{3}{4}$, $\mathrm{P}^*(\cup_{j=1}^{{\blue 3}}(\hat{X}_j=x_2))={\blue {3}}p_2^{(\alpha)}=\frac{3}{5}$, and {\blue{$\mathrm{P}^*(\cup_{j=1}(\hat{X}_j=x_3))=3p_3^{(\alpha)}=\frac{9}{10}$}}.\\
\noindent\underline{With $s^*=2$}: Let $P_X$ be such that $p_1^{(\alpha)}=\frac{3}{8}$, $p_2^{(\alpha)}=\frac{1}{4}$, $p_3^{(\alpha)}=\frac{3}{16}$, and $p_4^{(\alpha)}=\frac{3}{16}$. It can be verified that $s^*=2$ for this $P_X$. In particular, we have $p_1^{(\alpha)}>\frac{1}{3}$ and $\frac{2p_i^{(\alpha)}}{\sum_{j=2}^4p_j^{(\alpha)}}\leq 1$, for $i\in[2:4]$. From \eqref{eqn:thm11}, an optimal guessing strategy $P^*_{\hat{X}_1\hat{X}_2\hat{X}_3}$ is such that the adversary always guesses $x_1$ in one of the three guesses with a probability $1$, i.e., $\mathrm{P}^*(\cup_{j=1}^{{\blue 3}}(\hat{X}_j=x_1))=1$, and guesses $x_i$ with a probability proportional to $p_i^{(\alpha)}$, for $i\in[2:4]$, i.e., $\mathrm{P}^*(\cup_{j=1}^{{\blue 3}}(\hat{X}_j=x_2))=\frac{2p_2^{(\alpha)}}{\sum_{j=2}^4p_j^{(\alpha)}}=\frac{4}{5}$, and $\mathrm{P}^*(\cup_{j=1}(\hat{X}_j=x_i))=\frac{2p_i^{(\alpha)}}{\sum_{j=2}^4p_j^{(\alpha)}}=\frac{6}{10}$, for $i\in[3:4]$.  \\
\noindent\underline{With $s^*=3$}: Let $P_X$ be such that $p_1^{(\alpha)}=\frac{2}{3}$, $p_2^{(\alpha)}=\frac{1}{4}$, and $p_3^{(\alpha)}=p_4^{(\alpha)}=\frac{1}{24}$. It can be verified that $s^*=3$ for this $P_X$. In particular, we have $p_1^{(\alpha)}>\frac{1}{3}$, $\frac{2p_2^{(\alpha)}}{\sum_{j=2}^4p_j^{(\alpha)}}> 1$, and $\frac{p_3^{(\alpha)}}{\sum_{j=3}^4p_j^{(\alpha)}}\leq 1$. From \eqref{eqn:thm11}, an optimal guessing strategy $P^*_{\hat{X}_1\hat{X}_2\hat{X}_3}$ is such that the adversary always guesses $x_i$ in one of the three guesses with a probability $1$, i.e., $\mathrm{P}^*(\cup_{j=1}^{{\blue 3}}(\hat{X}_j=x_i))=1$, for $i\in[1:2]$, and guesses $x_i$ with a probability proportional to $p_i^{(\alpha)}$, i.e., $\mathrm{P}^*(\cup_{j=1}^{{\blue 3}}(\hat{X}_j=x_i)){\blue =\frac{p_i^{\alpha}}{\sum_{i=3}^4p_i^{(\alpha)}}}=\frac{1}{2}$, for $i\in[3:4]$. 
\end{example}
The following two corollaries of Theorem~\ref{thm:minloss-alpha-kguesses} give the expressions for the minimal expected log-loss ($\alpha=1$) for $k$ guesses and the minimal expected $0$-$1$ loss ($\alpha=\infty$) for $k$ guesses, respectively.
\begin{corollary}[Minimal expected log-loss \{$\alpha=1$\} for $k$ guesses]\label{corollary1}
Under the notations of Theorem~\ref{thm:minloss-alpha-kguesses}, the minimal expected log-loss for $k$ guesses is given by
\begin{align}
    \mathcal{ME}^{(k)}_1(P_X)&= H(X)-H_{s^*}\left(p_1,p_2,\dots,p_{{s^*}-1},\sum_{i={s^*}}^np_i\right)\nonumber\\
    &\hspace{12pt}-\left(\sum\limits_{i={s^*}}^np_i\right)\log{(k-{s^*}+1)},
\end{align}
where $s^*=\min\left\{r\in\{1,2,\ldots,k\}: \frac{(k-r+1)p_r}{\sum_{i=r}^np_i}\leq 1\right\}$
and $H_{s^*}(q_1,q_2,\dots,q_{s^*}):=\sum_{i=1}^{s^*}q_i\log{\frac{1}{q_i}}$ is the entropy function.
\end{corollary}
{\blue{The proof of Corollary~\ref{corollary1} follows by taking limit $\alpha\rightarrow 1$ using L'H\^{o}pital's rule in the result of Theorem~\ref{thm:minloss-alpha-kguesses} and rearranging the terms.}} 
\begin{corollary}[Minimal expected $0$-$1$ loss \{$\alpha=\infty$\} for $k$ guesses]\label{corollary2}
Under the notations of Theorem~\ref{thm:minloss-alpha-kguesses}, the minimal expected $0$-$1$ loss for $k$ guesses is given by
\begin{align}
   \mathcal{ME}^{(k)}_\infty(P_X)
    &=1-\sum_{i=1}^kp_i\nonumber\\
    &=1-\max_{\substack{a_1,a_2,\dots,a_k:\\ a_l\neq a_m,l\neq m}}\sum_{i=1}^kP_X(a_i).
\end{align}
\end{corollary}
{\blue{The proof of Corollary~\ref{corollary2} follows by taking limit $\alpha\rightarrow \infty$ in Theorem~\ref{thm:minloss-alpha-kguesses}}}. Interestingly, the polychotomy induced by minimal expected $\alpha$-loss for $\alpha\in(0,\infty)$ collapses for the case of $\alpha=\infty$~\cite{Issaetal} as clear from Corollary~\ref{corollary2}.\\

\subsection{$\alpha$-Leakage and Maximal $\alpha$-Leakage Under Multiple Guesses}\label{subsection:alphaleaka-multguesses}
{\blue{Notice that minimizing the expected $\alpha$-loss in \eqref{eqn:optimizationproblem} amounts to maximizing the expected gain for the gain function $g(t)=t^{\frac{\alpha-1}{\alpha}}$, $t\in(0,1)\cup(1,\infty)$.}} Motivated by $\alpha$-leakage~\cite[Definition~5]{LiaoKS20} (that is defined based on this gain function) which captures how much information an adversary can learn about a random variable $X$ from a correlated random variable $Y$ when a single guess is allowed, {\blue{we define $\alpha$-leakage under multiple guesses that captures the information an adversary can learn when $k$ guesses are allowed.}} 
\begin{definition}[$\alpha$-leakage under $k$ guesses]
{\blue{Given a joint distribution $P_{XY}$, the $\alpha$-leakage from $X$ to $Y$ under $k$ guesses is defined as}}
\begin{align}\label{defn:kalphaleakage}
    &\mathcal{L}^{(k)}_\alpha(X\rightarrow Y)\nonumber\\
    &\triangleq \frac{\alpha}{\alpha-1}\log{\frac{\max\limits_{P_{\hat{X}_{[1:k]}|Y}}\mathbb{E}\left[{\nblue{\frac{\alpha}{\alpha-1}}}\mathrm{P}\left(\bigcup\limits_{i=1}^k(\hat{X}_i=X)|Y\right)^{\frac{\alpha-1}{\alpha}}\right]}{\max\limits_{P_{\hat{X}_{[1:k]}}}\mathbb{E}\left[{\nblue{\frac{\alpha}{\alpha-1}}}\mathrm{P}\left(\bigcup\limits_{i=1}^k(\hat{X}_i=X)\right)^{\frac{\alpha-1}{\alpha}}\right]}},
\end{align}
{\blue where $\hat{X}_1,\hat{X}_2,\dots,\hat{X}_k$ represent $k$ estimators of $X$ with the same support as $X$, for $\alpha\in(0,1)\cup (1,\infty)$, and by the continuous extension of \eqref{defn:kalphaleakage} for $\alpha=1$ and $\alpha=\infty$}.
\end{definition}
{\blue{We call maximal $g$-leakage under $k$ guesses when specialized to the gain function $g(t)=t^{\frac{\alpha-1}{\alpha}}$, $t\in(0,1)\cup(1,\infty)$, as maximal $\alpha$-leakage under $k$ guesses.}}
\begin{definition}[Maximal $\alpha$-leakage under $k$ guesses]
Given a joint distribution $P_{XY}$, the maximal $\alpha$-leakage from $X$ to $Y$ under $k$ guesses is defined as 
\begin{align}\label{eqn:maxalphaleakage-multiguesses}
     \mathcal{L}^{(k)-\emph{max}}_\alpha(X\rightarrow Y)\triangleq\sup_{U-X-Y}\mathcal{L}^{(k)}_\alpha(U\rightarrow Y),
\end{align}
for {\blue{$\alpha\in(0,1)\cup(1,\infty)$}}.
\end{definition}
 
Let $P_{X|Y=y}^{(\alpha)}$ denote the tilted distribution of $P_{X|Y=y}$, i.e., $P_{X|Y}^{(\alpha)}(x|y)=\frac{P_{X|Y}(x|y)^\alpha}{\sum_xP_{X|Y}(x|y)^\alpha}$.
\begin{theorem}[Robustness of $\alpha$-leakage to number of guesses]\label{thm:robustness-of-alphaleakage}
Consider a $P_{XY}$ such that $P_{X|Y}^{(\alpha)}(x|y)\leq \frac{1}{k}$ and $P_{X}^{(\alpha)}(x)\leq \frac{1}{k}$, for all $x,y$, and {\blue{for $\alpha\in(0,1)\cup(1,\infty)$}}. Then
\begin{align}
    \mathcal{L}^{(k)}_\alpha(X\rightarrow Y)=\mathcal{L}^{(1)}_\alpha(X\rightarrow Y).
\end{align}
\end{theorem}
{\blue{For $k=1$, recall from Theorem~\ref{thm:minloss-alpha-kguesses} that in the optimal guessing strategy, the adversary guesses $x$ with a probability that is equal to the tilted probability value $P^{(\alpha)}_X(x)$, $x\in\mathcal{X}$. For $k>1$, when the condition $P_{X}^{(\alpha)}(x)\leq \frac{1}{k}$, for all $x$, holds, it follows from Theorem~\ref{thm:minloss-alpha-kguesses} that the adversary can apply essentially the same guessing strategy as with $k=1$ in the following sense. In particular, when this condition holds, in the optimal guessing strategy, the adversary can still guess $x$ in one of the $k$ guesses with a probability that is proportional to tilted probability $P^{(\alpha)}_X(x)$, i.e., $\mathrm{P}^*(\cup_{j=1}^{k}(\hat{X}_j=x))=kP_X^{(\alpha)}(x)$, $x\in\mathcal{X}$. Moreover, for the existence of a strategy to satisfy this equality, the condition $P_{X}^{(\alpha)}(x)\leq \frac{1}{k}$, for all $x,$ is a necessary and sufficient condition. This is because if $P^{(\alpha)}_X(x)>\frac{1}{k}$, for some $x\in\mathcal{X}$, then $\mathrm{P}^*(\cup_{j=1}^{k}(\hat{X}_j=x))$ cannot be equal to $kP^{(\alpha)}_X(x)>1$.}}
A detailed proof of Theorem~\ref{thm:robustness-of-alphaleakage} is in Appendix~\ref{proofofthm3}. It is shown in~\cite[Theorem~4]{Issaetal} that maximal leakage, i.e., maximal $\alpha$-leakage with $\alpha=\infty$, does not change with the number of guesses, i.e., $\mathcal{L}^{(k)-\text{max}}_\infty(X\rightarrow Y)=  \mathcal{L}^{(1)-\text{max}}_\infty(X\rightarrow Y)$. The proof of this result does not directly generalize to any arbitrary $\alpha\in(0,\infty]$. This is mainly due to the fact that unlike the case for $\alpha=\infty$, the minimal expected $\alpha$-loss for arbitrary $\alpha$ ($\neq\infty$) induces a polychotomy on the simplex of probability distributions. However, we are able to show that maximal $\alpha$-leakage under $k$ guesses is at least that of under a single guess where the proof relies on a non-trivial observation about the optimizer in \eqref{eqn:maxalphaleakage-multiguesses} and it also uses Theorem~\ref{thm:robustness-of-alphaleakage}. 

\begin{theorem}[Lower Bound on Maximal $\alpha$-leakage Under $k$ guesses]\label{thm:lowerbound-maxalphaleakage-kguesses}
Given a joint probability distribution $P_{XY}$ on finite alphabet $\mathcal{X}\times\mathcal{Y}$, we have
\begin{align}\label{eqn:maxalphaleakmult-atleast-single}
    \mathcal{L}^{(k)-\emph{max}}_\alpha(X\rightarrow Y)\geq  \mathcal{L}^{(1)-\emph{max}}_\alpha(X\rightarrow Y),
\end{align}
{\blue{for $\alpha\in(0,1)\cup(1,\infty)$}}.
\end{theorem}
A detailed proof of Theorem~\ref{thm:lowerbound-maxalphaleakage-kguesses} is in Appendix~\ref{proofofthm4}. We do not know if the reverse direction of \eqref{eqn:maxalphaleakmult-atleast-single} is true, i.e., if $\mathcal{L}^{(k)-\text{max}}_\alpha(X\rightarrow Y)\leq  \mathcal{L}^{(1)-\text{max}}_\alpha(X\rightarrow Y)$, in general.

\section{A Variational Formula for $\infty$-R\'{e}nyi divergence with Applications to Information Leakage}\label{section:variational}
{\blue{In the previous two sections, we examined information leakage measures, in particular, maximal $g$-leakage, where the adversary is allowed to make multiple guesses. In this section, we study some variants of maximal $g$-leakage depending on the type of adversary by obtaining a variational characterization of R\'{e}nyi divergence of order $\infty$. Such a characterization adds to a growing set of information measures written in a variational form, which may be of independent interest. We focus on the scenario where the adversary can make only a single guess in this section.}} In particular, we consider maximal expected gains of an adversary in separately guessing randomized functions of $X$ and our variational characterization is in terms of the ratio of these maximal expected gains. 
\begin{theorem}[A variational characterization for $D_\infty(\cdot||\cdot)$]\label{theorem:main-varchar}
{\blue{Let $P_X$ and $Q_X$ be two probability distributions on a finite alphabet $\mathcal{X}$ and $g:[0,1]\rightarrow [0,\infty)$ be a gain function  satisfying the following assumptions:}}
\begin{itemize}
    \item $g(0)=0$ and $g$ is continuous at 0,
    \item $0<\sup_{p\in[0,1]}g(p)<\infty$.
\end{itemize}
Then, we have
\begin{align}\label{eqn:variationalmain}
    D_\infty(P_X||Q_X)=\sup_{P_{U|X}}\log{\frac{\sup_{P_{\hat{U}}}\mathbb{E}_{U\sim P_U}\left[g(P_{\hat{U}}(U))\right]}{\sup_{P_{\hat{U}}}\mathbb{E}_{U\sim Q_U}\left[g(P_{\hat{U}}(U))\right]}},
\end{align}
where $P_U(u)=\sum_xP_X(x)P_{U|X}(u|x)$ and $Q_U(u)=\sum_xQ_X(x)P_{U|X}(u|x)$.
\end{theorem}
\begin{remark}\label{remark:log}
{\blue{Interestingly, there are non-positive gain functions {\blue{also}} for which while the conditions in Theorem~\ref{theorem:main-varchar} are not satisfied, \eqref{eqn:variationalmain} still holds.}} For example, $g(t)=\log{t}$ is one such function (see Appendix~\ref{appendix-loggain} for details). 
\end{remark}
A detailed proof of Theorem~\ref{theorem:main-varchar} is in Appendix~\ref{proofofthm5}. The numerator and the denominator in the ratio in \eqref{eqn:variationalmain} capture the maximal expected gains in guessing an unknown random variable $U$ distributed according to $P_U$ or $Q_U$, respectively. In a way, this ratio compares the distributions $P_X$ and $Q_X$ and is certainly dependent on the gain function $g$. However, our variational characterization in \eqref{eqn:variationalmain} shows that this ratio when optimized over all the channels $P_{U|X}$ remains constant irrespective of the gain function used, i.e., \eqref{eqn:variationalmain} holds for a broad class of gain functions. Some examples of gain function $g$ that satisfy the conditions in Theorem~\ref{theorem:main-varchar} are
\begin{align*}
    g(p)=p^2,\ \mathbbm{1}\{p=1/2\}, \ \frac{\alpha}{\alpha-1}p^{\frac{\alpha-1}{\alpha}} \ \text{where}\ \alpha\in(1,\infty).
\end{align*}
We obtain the following corollary from Theorem~\ref{theorem:main-varchar} by substituting the latter gain function 
$g_\alpha(t)=\frac{\alpha}{\alpha-1}t^{\frac{\alpha-1}{\alpha}}, \ \text{where}\ \alpha\in(1,\infty)$
(related to a class of adversarial loss functions, namely, $\alpha$-loss~\cite{LiaoKS20}) and using \cite[Lemma~1]{LiaoKS20} which gives closed-form expressions for the corresponding optimization problems in the numerator and the denominator in \eqref{eqn:variationalmain}.
\begin{corollary}\label{corollary:var3}
Given two probability distributions $P_X$ and $Q_X$ on a finite alphabet $\mathcal{X}$, we have, for $\alpha\in(1,\infty)$,
\begin{align}\label{prop1eq}
    D_\infty(P_X||Q_X)=\sup_{P_{U|X}}\log\frac{\left(\sum_uP_U(u)^\alpha\right)^\frac{1}{\alpha}}{\left(\sum_uQ_U(u)^\alpha\right)^\frac{1}{\alpha}},
\end{align}
where $P_U(u)=\sum_xP_X(x)P_{U|X}(u|x)$ and $Q_U(u)=\sum_xQ_X(x)P_{U|X}(u|x)$.
\end{corollary}

As mentioned earlier, we note that the existing variational characterizations for $D_\alpha(\cdot\|\cdot)$ (with finite $\alpha$)  also give rise to variational characterizations for $D_\infty(\cdot\|\cdot)$ by taking limit $\alpha\rightarrow \infty$. Shayevitz~\cite{Shayevitz11} and Birrell~\emph{et al.}~\cite{BirrellDKRW21} proved that 
\begin{align}
    D_\alpha(P_X\|Q_X)=\sup_{R_X:R_X\ll P_X}(D(R_X\|Q_X)-\frac{\alpha D(R_X\|P_X)}{\alpha-1}),\nonumber\\
   \label{eqn:varshayvitz}
    \end{align}
    for $\alpha>1$, and
\begin{align}
   &D_\alpha(P_X\|Q_X)\nonumber\\
   & =\sup_{g:\mathcal{X}\rightarrow \mathbb{R}}\left(\frac{\alpha\log{\mathbb{E}_{X\sim Q_X}\mathrm{e}^{(\alpha-1)g(X)}}}{\alpha-1}-\log{\mathbb{E}_{X\sim P_X}\mathrm{e}^{\alpha g(X)}}\right),\label{eqn:varBirrell}
\end{align}
{for} $\alpha\in\mathbb{R}\setminus\{0,1\}$, respectively. {\blue{More general forms of \eqref{eqn:varshayvitz} and an equivalent form of \eqref{eqn:varBirrell} appear in \cite{VanH14,Sason16,Anantharam} and \cite{atar2015robust}, respectively.}} One can obtain the variational characterizations for $D_\infty(\cdot\|\cdot)$ by taking limit $\alpha\rightarrow \infty$ in \eqref{eqn:varshayvitz} and assuming interchangeability of the limit and the supremum; 
one can similarly do so, in \eqref{eqn:varBirrell}, using a change of variable $f=\mathrm{e}^{\alpha g}$ and assuming interchangeability of the limit and the supremum. For the sake of completeness and rigor, we summarize the resulting variational forms for $D_\infty(\cdot\|\cdot)$ in the following proposition and present a proof in Appendix~\ref{appendix1}.
\begin{proposition}\label{proposition}
Given two probability distributions $P_X$ and $Q_X$ on a finite alphabet $\mathcal{X}$, we have
\begin{align}
     D_{\infty}(P_X\|Q_X)&=\sup_{R_X:R_X\ll P_X} \left(D(R_X\|Q_X)-D(R_X\|P_X)\right),\label{eqn:varinf-Anantharam}\\
      D_\infty(P_X\|Q_X)&= \sup_{f:\mathcal{X}\rightarrow[0,\infty)}\log\frac{\mathbb{E}_{X\sim P_X}[f(X)]}{\mathbb{E}_{X\sim Q_X}[f(X)]}\label{eqn:varinf-Birrell}.
\end{align}
\end{proposition}

Motivated by Issa~\emph{et al.}~\cite[Definitions~2 and 8]{Issaetal}, we define the following variants of maximal $g$-leakage (see Definition~\ref{def:maximal-g-leakage}) depending on the type of adversary. In particular, note that the definition of maximal $g$-leakage corresponds to an adversary interested in a \emph{single} randomized function of $X$. However, in some scenarios, the adversary could choose the guessing function depending on the realization of $Y$, leading to the following definition.

\begin{definition}[Opportunistic maximal $g$-leakage]\label{def:opp-g-leakage}
{\blue{Given a gain function $g:[0,1]\rightarrow [0,\infty)$ and a joint probability distribution $P_{XY}$ on a finite alphabet $\mathcal{X}\times\mathcal{Y}$, the opportunistic maximal $g$-leakage is defined as}} 
\begin{align}\label{opp-g-leakage}
    \tilde{\mathcal{L}}_g^{\emph{max}}&(X\rightarrow Y)=
\log \sum_{y\in\emph{supp}(Y)}P_Y(y)\nonumber\\
&\sup_{U:U-X-Y}\frac{\sup_{P_{\hat{U}|Y=y}}\mathbb{E}_{U|Y=y}\left[g(P_{\hat{U}|Y}(U|y))\right]}{\sup_{P_{\hat{U}}}\mathbb{E}\left[g(P_{\hat{U}}(U))\right]}.
\end{align}
\end{definition}

Maximal $g$-leakage captures the \emph{average} guessing performance of the adversary over all the realizations of $Y$. In some contexts, it might be more relevant to consider \emph{maximum} instead of the average.  
\begin{definition}[Maximal realizable $g$-leakage]\label{def:max-rel-leakage}
{\blue{Given a gain function $g:[0,1]\rightarrow [0,\infty)$ and a joint probability distribution $P_{XY}$ on a finite alphabet $\mathcal{X}\times\mathcal{Y}$, the opportunistic maximal $g$-leakage is defined as}}  
\begin{align}\label{max-rel-gleakage}
    &{\mathcal{L}}_g^{\emph{r-max}}(X\rightarrow Y)=
\sup_{U:U-X-Y}\\
&\log\frac{\max_{y\in\emph{supp}(Y)}\sup_{P_{\hat{U}|Y=y}}\mathbb{E}_{U|Y=y}\left[g(P_{\hat{U}|Y}(U|y))\right]}{\sup_{P_{\hat{U}}}\mathbb{E}\left[g(P_{\hat{U}}(U))\right]}.
\end{align}
\end{definition}

When $g(t)=t$, note that the Definitions~\ref{def:opp-g-leakage} and \ref{def:max-rel-leakage} simplify to those of opportunistic maximal leakage~\cite[Definition~2]{Issaetal} and maximal realizable leakage~\cite[Definition~8]{Issaetal}, respectively. Unlike the expressions for maximal $g$-leakage (e.g., for $g(t)=t$, $g(t)=\frac{\alpha}{\alpha-1}t^{\frac{\alpha-1}{\alpha}}$), interestingly, it turns out that the closed-form expressions for the opportunistic maximal $g$-leakage and maximal realizable $g$-leakage do not depend on the particular gain function $g$ as long as it satisfies some mild regularity conditions. This is a consequence of the robustness of our variational characterization to gain function (Theorem~\ref{theorem:main-varchar}).   
\begin{corollary}[Opportunistic maximal, and maximal realizable $g$-leakages]\label{thm:alphaleakage-variants}
Let $g:[0,1]\rightarrow [0,\infty)$ be a function satisfying the following assumptions:
\begin{itemize}
    \item $g(0)=0$ and $g$ is continuous at 0,
    \item $0<\sup_{p\in[0,1]}g(p)<\infty$.
\end{itemize}
Then the opportunistic maximal $g$-leakage and maximal realizable $g$-leakage defined in \eqref{opp-g-leakage} and \eqref{max-rel-gleakage}, respectively, simplify to
\begin{align}
    \tilde{\mathcal{L}}_g^{\emph{max}}(X\rightarrow Y)&=I_{\infty}^{\text{S}}(X;Y),\\
     {\mathcal{L}}_g^{\emph{r-max}}(X\rightarrow Y)&=D_\infty(P_{XY}\|P_X\times P_Y).
\end{align}
\end{corollary}

{\blue{Note that the intuitive interpretation given for $g(0)=0$ below Theorem~\ref{thm:maximalg-leakage-concaveg} holds for Corollary~\ref{thm:alphaleakage-variants} also.}} A detailed proof of Corollary~\ref{thm:alphaleakage-variants} is in {\nblue{Appendix}}~\ref{proofofthm6}. When $g(t)=t$, {\nblue{Corollary}}~\ref{thm:alphaleakage-variants} recovers the expressions for opportunistic maximal leakage and maximal realizable leakage~\cite[Theorems~2 and 13]{Issaetal}. As another concrete example, consider Theorem~\ref{thm:alphaleakage-variants} for the corresponding variants of maximal $\alpha$-leakage with the gain function $g(t)=\frac{\alpha}{\alpha-1}t^{\frac{\alpha-1}{\alpha}}$, wherein the optimization problems in the numerators and the denominators of \eqref{opp-g-leakage} and \eqref{max-rel-gleakage} have closed-form expressions~\cite[Lemma~1]{LiaoKS20}. In particular, {\blue{for $\alpha\in(1,\infty)$}}, the opportunistic maximal $\alpha$-leakage is given by
\begin{align}\label{eqn:opp-alpha-leakage}
    &\frac{\alpha}{\alpha-1}\log \sum_{y\in\text{supp}(Y)}P_Y(y)\sup_{U:U-X-Y}\frac{\left(\sum_uP_{U|Y}(u|y)^\alpha\right)^\frac{1}{\alpha}}{\left(\sum_uP_U(u)^\alpha\right)^{\frac{1}{\alpha}}}\nonumber\\
    &=\frac{\alpha}{\alpha-1}I_{\infty}^{\text{S}}(X;Y)
\end{align}
and the maximal realizable $\alpha$-leakage is given by
\begin{align}
     &\sup_{U:U-X-Y}\frac{\alpha}{\alpha-1}\log\frac{\max_{y\in\text{supp}(Y)}\left(\sum_uP_{U|Y}(u|y)^\alpha\right)^{\frac{1}{\alpha}}}{\left(\sum_uP_U(u)^\alpha\right)^{\frac{1}{\alpha}}}\nonumber\\
     &=\frac{\alpha}{\alpha-1}D_\infty(P_{XY}\|P_X\times P_Y).
\end{align}
{\blue{Note that when $X$ and $Y$ are independent, the opportunistic maximal $(\alpha=1)$-leakage and  the maximal realizable $(\alpha=1)$-leakage defined by taking limit $\alpha\rightarrow 1$ are both equal to zero.}} {\blue{When $X$ and $Y$ are not independent,}} it can be inferred from the above expressions that opportunistic maximal $(\alpha=1)$-leakage and  maximal realizable $(\alpha=1)$-leakage are both equal to $\infty$. However, note that maximal $\alpha$-leakage as $\alpha\rightarrow 1$ is equal to Shannon channel capacity or Shannon mutual information depending on whether we define it using the supremum first and the limit next, or the limit first and the supremum next~\cite[Theorem~2]{LiaoKS20}. We show in Appendix~\ref{appendix:1-leakages} that the latter way of defining the opportunistic maximal $1$-leakage and the maximal realizable $1$-leakage also yields $\infty$.

 Another interesting consequence of our variational characterization is in its natural connection to the definition of pointwise maximal leakage, another measure of information leakage studied by Saeidian~\emph{et al.}~\cite{saeidian2022pointwise}. Pointwise maximal leakage captures the maximum multiplicative increase in the probability of correctly guessing \emph{any} randomized function of $X$ upon observing a single outcome $Y=y$. 
\begin{definition}[Pointwise maximal leakage~\cite{saeidian2022pointwise}]
Given a joint distribution $P_{XY}$ on a finite alphabet $\mathcal{X}\times\mathcal{Y}$ and a $y\in\mathcal{Y}$, the pointwise maximal leakage is defined as 
\begin{align}
    &\mathcal{L}^{\emph{pw-max}}(X\rightarrow y)\nonumber\\
    &=\sup_{U:U-X-Y}\log\frac{\sup_{P_{\hat{U}|Y=y}}\mathbb{E}_{U|Y=y}\left[P_{\hat{U}|Y}(U|y)\right]}{\sup_{P_{\hat{U}}}\mathbb{E}\left[P_{\hat{U}}(U)\right]}.
\end{align}
\end{definition}
Notice that pointwise maximal leakage is also defined in terms of an adversary interested in maximizing the expected value of the gain function $g(t)=t$ but it differs from the maximal leakage because the latter captures the average performance of the guessing adversary over all the realizations of $Y$. Saeidian~\emph{et al.}~\cite{saeidian2022pointwise} showed that 
\begin{align}
    \mathcal{L}^{\text{pw-max}}(X\rightarrow y)=D_\infty(P_{X|Y=y}\|P_X)
\end{align}
following the techniques in the proof of the closed-form expression for maximal realizable leakage~\cite[Theorem~13]{Issaetal}. 

An immediate consequence of our variational characterization is that the definition of the pointwise maximal leakage can be generalized by incorporating an adversary interested in maximizing the expected values of an arbitrary gain function $g$ that satisfies the mild regularity conditions mentioned in Corollary~\ref{theorem:main-varchar}.
\begin{corollary}[Pointwise maximal $g$-leakage]\label{thm:ptwisegleakage}
{\blue{Let $P_{XY}$ be a probability distribution on a finite alphabet $\mathcal{X}\times\mathcal{Y}$ and $g:[0,1]\rightarrow [0,\infty)$ be a gain function satisfying the following conditions:}}
\begin{itemize}
    \item $g(0)=0$ and $g$ is continuous at 0,
    \item $0<\sup_{p\in[0,1]}g(p)<\infty$.
\end{itemize}
Then, for $y\in\mathcal{Y}$, we have
\begin{align}\label{eqn:ptwisegleakage}
    \sup_{U:U-X-Y}\log\frac{\sup_{P_{\hat{U}|Y=y}}\mathbb{E}_{U|Y=y}\left[g(P_{\hat{U}|Y}(U|y))\right]}{\sup_{P_{\hat{U}}}\mathbb{E}\left[g(P_{\hat{U}}(U))\right]}\nonumber\\
    =D_\infty(P_{X|Y=y}\|P_X).
    \end{align}
\end{corollary}
The proof of this {\blue{corollary}} follows directly from the proof of Corollary~\ref{thm:alphaleakage-variants}. However, this result may be of separate interest in view of the following observations. The expression on the LHS in \eqref{eqn:ptwisegleakage} can be thought of as a new leakage measure, namely, pointwise maximal $g$-leakage (denoted by $\mathcal{L}_g^{\text{pw-max}}(X\rightarrow y)$) and Theorem~\ref{thm:ptwisegleakage} implies that pointwise maximal $g$-leakage is exactly the same as pointwise maximal leakage, i.e., $\mathcal{L}_g^{\text{pw-max}}(X\rightarrow y)=\mathcal{L}^{\text{pw-max}}(X\rightarrow y)$, for all $g$ satisfying some mild conditions. Moreover, adopting the gain function approach of Alvim \emph{et al.}~\cite{Alvimetal14}, Saeidian~\emph{et al.}{\blue{\cite{saeidian2022pointwise}}} showed that the multiplicative increase in the maximal expected value of the gain function (a function of the true value and the guessed value) in guessing $X$ after observing an outcomne $Y=y$ when maximized over all gain functions $g(x,\hat{x})$ is equal to the pointwise maximal leakage (\cite[Theorem~2 and Corollary~1]{saeidian2022pointwise}). On the other hand, Theorem~\ref{thm:ptwisegleakage} implies that such a multiplicative increase with gain function applied to the probability of correctly guessing a (randomized) function of $X$ when maximized over all the functions of $X$ is equal to pointwise maximal leakage, for \emph{any} gain function $g$ satisfying some mild conditions. Finally, we note that all the properties and privacy guarantees of pointwise maximal leakage studied in \cite{saeidian2022pointwise} (e.g., \emph{composition} and \emph{data-processing}) will straightaway generalize to pointwise maximal $g$-leakage.

\section{Conclusion}

{\nblue{We have proposed a gain function viewpoint of information leakage using arbitrary non-negative functions applied to the probability of correctly guessing. The primary benefit of restricting the gain functions considered by \cite{Alvimetal14,Issaetal} to the ones using the probability of correctness (as in Definition~1) is that it allows us to obtain closed-form expressions for maximal $g$-leakage for a class of concave gain functions (Theorem~1). We have shown that maximal $g$-leakage under multiple guesses is equal to Sibson mutual information of order infinity for a class of concave gain functions. This is in contrast to the corresponding results of \cite{Alvimetal14,Issaetal} where the worst-case scenario over all the gain functions is considered. Moreover, obtaining a closed-form expression for such leakage measures with a fixed gain function was conjectured to be challenging in \cite[Section~VI-A]{Alvimetal14} even when the threat model focuses on guessing $X$ itself (rather than a possibly randomized function of $X$). 

Even though the closed-form expression for maximal leakage is equal to that of min-capacity (i.e., the maximum min-entropy leakage) \cite{braun2009quantitative}, maximal leakage is important mainly because of the relaxation of assumptions about the adversary. In particular, in the threat model considered for maximal leakage, the adversary is interested in the (possibly randomized) function of $X$. In contrast, in the setup of min-capacity, the adversary is interested in $X$ itself. In the same vein, we extended the framework of maximal leakage to incorporate arbitrary functions of probability of correctness in the performance measure of the adversary. Thus, our setup actually leads to considering a larger class of adversaries because the gain functions in Theorem~\ref{thm:maximalg-leakage-concaveg} recover the probability of correctness as a special case. We also presented a variational characterization of R\'{e}nyi divergence of order infinity, which is naturally related to the pointwise version of maximal $g$-leakage. 

We also studied the setting in which the adversary is allowed multiple attempts in guessing by focusing on a specific tunable gain function. Such a setting has not been explored earlier even for log-loss (a special case of the tunable loss function considered here) which is extensively used in machine learning. We have proved that a new measure of divergence that belongs to the class of Bregman divergences captures the relative performance of an arbitrary adversarial strategy with respect to an optimal strategy in minimizing the expected $\alpha$-loss. 

All these results strengthen the connection between privacy leakage and the information measures -- Sibson mutual information and R\'{e}nyi divergence of infinite orders. We believe that these results are beneficial in applications where the adversary's performance is measured via generalized gain functions. In terms of privacy-utility tradeoff, a consequence of our results is that we can obtain the same utility even when we consider various privacy leakage measures with potentially different operational interpretations as per the adversary's performance measure of interest. Motivated by the results in this work, we anticipate that closed-form expressions for maximal $g$-leakage measures for more gain functions depending on specific applications will be explored further. There are many questions to be further studied. One limitation of our study is that we have considered only the gain functions that are non-negative (though we studied a particular non-positive gain function, $g(t)=\log{t}$ in Appendix~\ref{appendix-loggain}). It would be interesting to characterize maximal $g$-leakage for any general class of gain functions. We have shown that maximal $\alpha$-leakage under multiple guesses is at least that of with a single guess. It would be worth studying if the reverse direction is also true as in $\alpha=\infty$ (maximal leakage) case. }}

{\blue\section{Acknowledgment} Gowtham R. Kurri would like to thank Tyler Sypherd for helpful discussions on Bregman divergence and its connections to R\'{e}nyi divergence in Lemma~\ref{thm:expected-alpha-loss-singleguess} motivated via $\alpha$-loss.}
\appendices
\section{Proofs for Section~\ref{section:maximalg-leakage}}
\subsection{Proof of Theorem~\ref{thm:maximalg-leakage-concaveg}}\label{thm:maxgleakage2}
{\blue{
{\blue{The upper bound $\mathcal{L}_g^{(k)-\text{max}}(X\rightarrow Y)\leq I_{\infty}^{\emph{S}}(X;Y)$ follows from Proposition~\ref{prop:upperboundg}.}} We prove the lower bound now. For this, we use the {\blue{``shattering'' conditional distribution $P_{U|X}$~\cite[Proof of Theorem~1]{Issaetal},\cite[Proof of Theorem~5]{Liaoetal}.}} Let $\mathcal{U}=\cup_{x\in\mathcal{X}}\mathcal{U}_x$ (a disjoint union) and $|\mathcal{U}_x|=m_x$. Then define
\begin{align}\label{eqn:shatter}
    P_{U|X}(u|x)=\begin{cases}
    \frac{1}{m_x}, &u\in\mathcal{U}_x,\\
    0, &\text{otherwise}.
    \end{cases}
\end{align}
This gives 
\begin{align}
    P_U(u)&=P_X(x)/m_x, u\in\mathcal{U}_x,\label{eqn:newthm11}\\
    P_{U|Y}(u|y)&=P_{X|Y}(x|y)/m_x, u\in\mathcal{U}_x.
\end{align}
{\blue{To simplify the notation, let
\begin{align}
    \hat{P}_k(u)&:=\mathrm{P}\left(\bigcup\limits_{i=1}^k(\hat{U}_i=u)\right),\\
    \hat{P}_k(u|y)&:=\mathrm{P}\left(\bigcup\limits_{i=1}^k(\hat{U}_i=u)|Y=y\right).
\end{align}
}}
We upper bound the denominator of
{\blue\begin{align}\label{eqn:objmaxgleakge}
    \frac{\sup_{P_{\hat{U}_{[1:k]}|Y}}\mathbb{E}_{UY}\left[g(\hat{P}_k(U|Y))\right]}{\sup_{P_{\hat{U}_{[1:k]}}}\mathbb{E}_U\left[g(\hat{P}_k(U))\right]}
\end{align}}
as
\begin{align}
&\sup_{P_{\hat{U}_{[1:k]}}}\mathbb{E}_U\left[g(\hat{P}_k(U))\right]\nonumber\\
&=\sup_{P_{\hat{U}_{[1:k]}}}\sum_ug(\hat{P}_k(u))P_U(u)\\
&\leq \sup_{P_{\hat{U}_{[1:k]}}}\sum_u(g(0)+g^\prime(0)\hat{P}_k(u))P_U(u)\label{eqn:thm61}\\
&=g^\prime(0)\sup_{P_{\hat{U}_{[1:k]}}}\sum_u\hat{P}_k(u)P_U(u)\label{eqn:thm613}\\
&=g^\prime(0){\blue{\max_{u_1,u_2,\dots,u_k:u_i\neq u_j, i\neq j}\sum_{i=1}^kP_U(u_i)}}\label{eqn:thm62}\\
&=kg^\prime(0)\max_{u}P_U(u)\label{eqn:newthm12},
\end{align}
where \eqref{eqn:thm61} follows because $g(s)\leq g(t)+g^\prime(t)(s-t)$, for all $s,t\in[0,1]$ since $g$ is a concave function, \eqref{eqn:thm613} follows because $g(0)=0$, and \eqref{eqn:newthm12} follows from \eqref{eqn:newthm11} if $k\leq m_x$, for all $x\in\mathcal{X}$.

We now lower bound the numerator. Since the function $g$ is differentiable at $0$, there exists a function $h(x)$ such that 
\begin{align}\label{eqn:Tyler}
    g(x)=g(0)+g^\prime(0)x+xh(x),\  \lim_{x\rightarrow 0}h(x)=0.
\end{align}
Notice that
\begin{align}\label{eqn:thm6new}
    &\sup_{P_{\hat{U}_{[1:k]}|Y}}\mathbb{E}_{UY}\left[g(\hat{P}_k(U|Y))\right]\nonumber\\
    &=\sum_yP_Y(y)\sup_{P_{\hat{U}_{[1:k]}|Y=y}}\mathbb{E}_{U|Y=y}\left[g(\hat{P}_k(U|y))\right].
\end{align}
Consider, for a fixed $y\in\mathcal{Y}$,
\begin{align}
    &\sup_{P_{\hat{U}_{[1:k]}|Y=y}}\mathbb{E}_{U|Y=y}\left[g(\hat{P}_k(U|y))\right]\label{eqn:neqopt1}\\
    &=\sup_{P_{\hat{U}_{[1:k]}|Y=y}}\sum_ug(\hat{P}_k(u|y))P_{U|Y}(u|y)\\
    &=\sup_{P_{\hat{U}_{[1:k]}|Y=y}}\sum_{x}P_{X|Y}(x|y)\sum_{u\in\mathcal{U}_x}\frac{1}{m_x}g(\hat{P}_k(u|y))\\
    &\geq\sup_{\substack{P_{\hat{U}_{[1:k]}|Y=y}:\forall u_1,u_2\in\mathcal{U}_x,x\in\mathcal{X}\\ \hat{P}_k(u_1|y)=\hat{P}_k(u_2|y)}}\bigg[\sum_{x}P_{X|Y}(x|y)\nonumber\\
    &\hspace{24pt}\sum_{u\in\mathcal{U}_x}\frac{1}{m_x}g(\hat{P}_k(u|y))\bigg]\label{eqn:thm6bound}\\
    &=\sup_{\substack{P_{\hat{U}_{[1:k]}|Y=y}:\forall u_1,u_2\in\mathcal{U}_x,x\in\mathcal{X}\\ \hat{P}_k(u_1|y)=\hat{P}_k(u_2|y)}}\bigg[\sum_{x}P_{X|Y}(x|y)\nonumber\\
    &\hspace{24pt}g\left(\frac{1}{m_x}\sum_{u\in\mathcal{U}_x}\hat{P}_k(u|y)\right)\bigg]\label{eqn:thm6concaveg}\\
    &=\sup_{P_{\hat{X}|Y=y}}\sum_xP_{X|Y}(x|y)g\left(\frac{k}{m_x}P_{\hat{X}|Y}(x|y)\right)\label{eqn:thm63}\\
    &\geq \sup_{P_{\hat{X}|Y=y}}\sum_xP_{X|Y}(x|y)(g(0)+\frac{k}{m_x}P_{\hat{X}|Y}(x|y)(g^\prime(0)-\epsilon))\label{eqn:thm64}\\
    &=(g^\prime(0)-\epsilon)\sup_{P_{\hat{X}|Y=y}}\sum_x\frac{k}{m_x}P_{X|Y}(x|y)P_{\hat{X}|Y}(x|y)\\
    &=k(g^\prime(0)-\epsilon)\max_x(\frac{1}{m_x}P_{X|Y}(x|y))\\
    &=k(g^\prime(0)-\epsilon)\max_x\max_{u\in\mathcal{U}_x}P_{U|Y}(u|y)\\
    &=k(g^\prime(0)-\epsilon)\max_uP_{U|Y}(u|y),
\end{align}
{\blue{where \eqref{eqn:thm6bound} follows because $\sup_{x\in A}f(x)\geq \sup_{x\in B}f(x)$ when $A\supseteq B$ and \eqref{eqn:thm6concaveg} follows because $\sup_{x\in B}f(x)=\sup_{x\in B}g(x)$ when $f(x)=g(x)$, for $x\in B$.}} We remark that it suffices to consider $P_{\hat{U}_{[1:k]}|Y=y}$ such that $\sum_{u}\hat{P}_k(u|y)=k$ in all the optimizations problems in \eqref{eqn:neqopt1}-\eqref{eqn:thm6concaveg} (see Lemma~\ref{fact:alphaloss2noneq} and Remark~\ref{remark} in Appendix~\ref{proofofthm2}). Then, it follows that $\eqref{eqn:thm6concaveg}\leq\eqref{eqn:thm63}$ by defining $P_{\hat{X}|Y}(x|y):=\frac{1}{k}\sum_{u\in\mathcal{U}_x}\hat{P}_k(u|y)$ which is a probability distribution noticing that $\mathcal{U}_x$, $x\in\mathcal{X}$, are disjoint and that $\sum_{u}\hat{P}_k(u|y)=k$. It follows that $\eqref{eqn:thm6concaveg}\geq\eqref{eqn:thm63}$ by defining $\hat{P}_k(u|y):=\frac{k}{m_x}P_{\hat{X}|Y}(x|y)$, for $u\in\mathcal{U}_x$, and using Lemma~\ref{theorem:probclass1} (see also the discussion above it) in Appendix~\ref{proofofthm2}. Thus, $\eqref{eqn:thm6concaveg}=\eqref{eqn:thm63}$. The inequality \eqref{eqn:thm64} follows from \eqref{eqn:Tyler} by choosing $m_x$ appropriately large enough, for $x\in\mathcal{X}$ and for $0<\epsilon\leq g^\prime(0)$. This gives
\begin{align}
   \sup_{P_{\hat{U}_{[1:k]}|Y}}&\mathbb{E}_{UY}\left[g(\hat{P}_k(U|Y))\right]\nonumber\\
    &\geq {\nblue{k}}(g'(0)-\epsilon)\sum_{y}P_Y(y)\max_{u}P_{U|Y}(u|y)\label{eqn:thm6numer}.
\end{align}

Putting together the bounds in \eqref{eqn:thm62} and \eqref{eqn:thm6numer} and using \eqref{eqn:thm6new}, we have
\begin{align}
     \sup_{U:U-X-Y}&\frac{\sup_{P_{\hat{U}_{[1:k]}|Y}}\mathbb{E}_{UY}\left[g(\hat{P}_k(U|Y))\right]}{\sup_{P_{\hat{U}_{[1:k]}}}\mathbb{E}_U\left[g(\hat{P}_k(U))\right]}\\
     &\geq\sup_{0<\epsilon\leq g^\prime(0)}\sup_{\substack{U:U-X-Y,\\ \ P_{U|X}\ \text{in}\  \eqref{eqn:shatter}}}\nonumber\\
     &\hspace{12pt}{\blue{\frac{k(g^\prime(0)-\epsilon)\sum_{y}P_Y(y)\max_{u}P_{U|Y}(u|y)}{kg^\prime(0)\max_{u}P_U(u)}}}\\
     &=\sup_{0<\epsilon\leq g^\prime(0)}\frac{(g^\prime(0)-\epsilon)}{g^\prime(0)}\sum_y\max_{x}P_{Y|X}(y|x)\label{eqn:maxleakageshater}\\
     &=\sum_y\max_{x}P_{Y|X}(y|x),\label{eqn:nonengarivestrictg}
\end{align}
where \eqref{eqn:maxleakageshater} follows because maximal leakage is achieved by the shattering $P_{U|X}$ {\blue{(defined in \eqref{eqn:shatter})}} with sufficiently large $m_x$~\cite{Liaoetal} and \eqref{eqn:nonengarivestrictg} follows because $0<g^\prime(0)<\infty$.

{\nblue\subsection{Proof of Theorem~\ref{thm:shiftedmaxL}}\label{proof of thm2}
Note that 
\begin{align}
&{\mathcal{L}}_g^{\text{max}}(X\rightarrow Y)\nonumber\\
&=\sup_{U:U-X-Y}\log \frac{1+\sup_{P_{\hat{U}|Y}}\mathbb{E}_{UY}\left[P_{\hat{U}|Y}(U|Y)\right]}{1+\sup_{P_{\hat{U}}}\mathbb{E}_U\left[P_{\hat{U}}(U)\right]}\\
&=\sup_{U:U-X-Y}\log \frac{1+\sum_y \max_u P_{UY}(u,y)}{1+\max_u P_U(u)}\label{thm2:objshiftdleak}.
\end{align}
The lower bound follows directly by using the ``Shattering'' conditional distribution $P_{U|X}$ used in \cite[Proof of Theorem~1]{Issaetal}. In particular, let $p^*=\min_{x:P_X(x)>0}$, $k(x)=\frac{P_X(x)}{p^*}$, for each $x\in\text{supp}(X)$, and $\mathcal{U}=\cup_{x\in\text{supp}(X)}\{(x,1),\dots,(x,\lceil k(x)\rceil)\}$. For each $u=(i_u,j_u)\in\mathcal{U}$ and $x\in\text{supp}(X)$, define $P_{U|X}$ as
\begin{align}
    P_{U|X}(i_u,j_u|X)=\begin{cases}
        \frac{p^*}{P_X(x)}, & \!\!\!i_u=x, \!\ j_u\in[1:\lfloor k(x)\rfloor],\\
        1-\frac{\lfloor k(x)\rfloor p^*}{P_X(x)}, &\!\!\! i_u=x, \!\ j_u=\lceil k(x)\rceil,\\
        0,&\!\!\!\text{otherwise}.
    \end{cases}
\end{align}
By substituting this into the objective function of \eqref{thm2:objshiftdleak}, we get the lower bound
\begin{align}\label{eqn:thm2-lowerbound}
    \mathcal{L}_g(X\rightarrow Y)\geq \log\frac{1+p^*\sum_{y\in\mathcal{Y}}\max_{x\in\text{supp}(X)}P_{Y|X}(y|x)}{1+p^*}.
\end{align}

We note that this lower bound holds for $\mathcal{X}$ and $\mathcal{Y}$ of arbitrary cardinalities. The binary assumption on $\mathcal{X}$ and $\mathcal{Y}$ is required for the upper bound. We prove the upper bound now. We have
\begin{align}
&{\mathcal{L}}_g^{\text{max}}(X\rightarrow Y)\nonumber\\
&=
\sup_{U:U-X-Y}\log \frac{\sup_{P_{\hat{U}|Y}}\mathbb{E}_{UY}\left[g(P_{\hat{U}|Y}(U|Y))\right]}{\sup_{P_{\hat{U}}}\mathbb{E}_U\left[g(P_{\hat{U}}(U))\right]}
\\&=\sup_{U:U-X-Y}\log \frac{1+\sup_{P_{\hat{U}|Y}}\mathbb{E}_{UY}\left[P_{\hat{U}|Y}(U|Y)\right]}{1+\sup_{P_{\hat{U}}}\mathbb{E}_U\left[P_{\hat{U}}(U)\right]}
\\&=\sup_{U:U-X-Y}\log \frac{1+\sum_y \max_u P_{UY}(u,y)}{1+\max_u P_U(u)}
\\&=\sup_{q\in[0,1]}\ \sup_{\substack{U:U-X-Y,\\\max_u P_U(u)\le q}} \log \frac{1+\sum_y \max_u P_{UY}(u,y)}{1+q}
\\&=\sup_{q\in[0,1]} \log \frac{\displaystyle 1+\sup_{\substack{U:U-X-Y,\\\max_u P_U(u)\le q}} \sum_y \max_u P_{UY}(u,y)}{1+q}\label{eqn:thm2-eq1}.
\end{align}
Let us consider just the optimization in the numerator in \eqref{eqn:thm2-eq1}:
\begin{equation}\label{eqn:thm2-fq-def}
f(q)=\sup_{\substack{U:U-X-Y,\\\max_u P_U(u)\le q}} \sum_y \max_u P_{U,Y}(u,y).
\end{equation}
We may upper bound it by
\begin{align}
&f(q)\nonumber\\
&=\sup_{\substack{U:U-X-Y,\\\max_u P_U(u)\le q}} \sum_y \max_u \sum_x P_U(u) P_{X|U}(x|u) P_{Y|X}(y|x) 
\\&\le \sup_{\substack{U:U-X-Y,\\\max_u P_U(u)\le q}} q\sum_y \max_u \sum_x P_{X|U}(x|u) P_{Y|X}(y|x) 
\\&\le q\sum_y \max_x P_{Y|X}(y|x).\label{fq_bound1}
\end{align}
We will also need another upper bound on $f(q)$ which uses the assumption that $\mathcal{X}$ and $\mathcal{Y}$ are binary. Let $\mathcal{X}=\{x_1,x_2\}$ and $\mathcal{Y}=\{y_1,y_2\}$. Assume without loss of generality that $P_X(x_1)\le P_X(x_2)$. Also, assume without loss of generality that $x_1=\argmax_x P_{Y|X}(y_1|x)$. That is,
\begin{equation}
P_{Y|X}(y_1|x_1)\ge P_{Y|X}(y_1|x_2).\label{y1x1_y1x2_bound}
\end{equation}
Since $\mathcal{Y}$ is binary, a consequence is that
\begin{equation}
P_{Y|X}(y_2|x_2)\ge P_{Y|X}(y_2|x_1).\label{y2x2_y2x1_bound}
\end{equation}
That is, $x_2=\argmax_x P_{Y|X}(y_2|x)$. We may write \eqref{eqn:thm2-fq-def} as
\begin{equation}
f(q)=\sup_{\substack{U:U-X-Y,\\\max_u P_U(u)\le q}} \max_u P_{U,Y}(u,y_1)+\max_u P_{U,Y}(u,y_2).
\end{equation}
Let $u_i=\argmax_u P_{U,Y}(u,y_i)$ for $i=1,2$. If $u_1=u_2$, then 
\begin{align}
f(q)&\le \sup_{\substack{U:U-X-Y,\\\max_u P_U(u)\le q}} P_{U,Y}(u_1,y_1)+P_{U,Y}(u_1,y_2)\nonumber\\
&=P_U(u_1)\le q.
\end{align}
Now consider the case where $u_1\ne u_2$. We construct an upper bound using weak duality as follows. Consider the constraints
\begin{align}
\sum_{j=1}^2 P_X(x_j)P_{U|X}(u_i|x_j)\le q,\ i=1,2,\\
\sum_{i=1}^2 P_{U|X}(u_i|x_j)\le 1,\ j=1,2,\\
P_{U|X}(u_i|x_j)\ge 0,\ i=1,2,j=1,2.
\end{align}
We may upper bound $f(q)$ by
\begin{align}
f(q)&\le \min_{\substack{\lambda_i\ge 0,\ i=1,2,\\
\alpha_j\ge 0,\ j=1,2,\\
\nu_{ij}\ge 0,\ i=1,2,j=1,2}}
\sup_{P_{U|X}} \sum_{i,j=1}^2 P_{X,Y}(x_j,y_i) P_{U|X}(u_i|x_j)
\nonumber\\&\qquad-\sum_{i=1}^2 \lambda_i\left(\sum_{j=1}^2 P_X(x_j)P_{U|X}(u_i|x_j)-q\right)\nonumber\\
&\qquad-\sum_{j=1}^2 \alpha_j \left(\sum_{i=1}^2 P_{U|X}(u_i|x_j)-1\right)
\nonumber\\&\qquad+\sum_{i,j=1}^2 \nu_{ij}P_{U|X}(u_i|x_j)
\\&=\min_{\substack{\lambda_i\ge 0,\ i=1,2,\\
\alpha_j\ge 0,\ j=1,2,\\
\nu_{ij}\ge 0,\ i=1,2,j=1,2}}
\sup_{P_{U|X}} \sum_{i,j=1}^2 \bigg[P_{U|X}(u_i|x_j)\nonumber\\
&\qquad\times\left(P_{X,Y}(x_j,y_i)-\lambda_i P_X(x_j)-\alpha_j+\nu_{ij}\right)\nonumber\\
&\qquad+\sum_{i=1}^2 \lambda_i q+\sum_{j=1}^2 \alpha_j\bigg].
\end{align}
This leads to the upper bounding dual program
\begin{equation}
\begin{array}{ll}
\text{minimize} & \displaystyle \sum_{i=1}^2 \lambda_i q+\sum_{j=1}^2 \alpha_j\\
\text{subject to} & \lambda_i\ge 0,\ i=1,2,\\
&\alpha_j\ge 0,\ j=1,2,\\
&\nu_{ij}\ge 0,\ i=1,2,j=1,2,\\
&P_{X,Y}(x_j,y_i)-\lambda_i P_X(x_j)-\alpha_j+\nu_{ij}=0,\nonumber\\
& i=1,2,j=1,2.
\end{array}
\end{equation}
Using the equality constraint to eliminate $\nu_{ij}$, this can be further simplified to
\begin{equation}
\begin{array}{ll}
\text{minimize} & \displaystyle \sum_{i=1}^2 \lambda_i q+\sum_{j=1}^2 \alpha_j\\
\text{subject to} & \lambda_i\ge 0,\ i=1,2,\\
&\alpha_j\ge 0,\ j=1,2,\\
&-P_{X,Y}(x_j,y_i)+\lambda_i P_X(x_j)+\alpha_j\ge 0,\nonumber\\
&i=1,2,j=1,2.
\end{array}
\end{equation}
Noting that, for fixed $\alpha_j$, $j\in\{1,2\}$, the optimal value of $\lambda_i$ is $\max_j\max\{0,P_{Y|X}(y_i|x_j)-\frac{\alpha_j}{P_X(x_j)}\}$, the dual program can again be simplified to
\begin{equation}
\begin{array}{ll}
\text{minimize} & \displaystyle q\sum_{i=1}^2 \max_j\max\left\{0,P_{Y|X}(y_i|x_j)-\frac{\alpha_j}{P_X(x_j)}\right\}\nonumber\\
&+\sum_{j=1}^2 \alpha_j\\
\text{subject to} &\alpha_j\ge 0,\ j=1,2.\\
\end{array}
\end{equation}
We can now form a further upper bound on $f(q)$ by choosing $\alpha_j$ as we like. With some hindsight, we set
\begin{equation}
\alpha_1=P_X(x_1)(P_{Y|X}(y_1|x_1)-P_{Y|X}(y_1|x_2)),\quad \alpha_2=0.
\end{equation}
By the assumption in \eqref{y1x1_y1x2_bound}, $\alpha_1$ is non-negative. Now we have the upper bound
\begin{align}
f(q)&\le 
q \sum_{i=1}^2 \max\{0,P_{Y|X}(y_i|x_1)\nonumber\\
&\qquad-(P_{Y|X}(y_1|x_1)-P_{Y|X}(y_1|x_2)),
P_{Y|X}(y_i|x_2)\}
\nonumber\\&\qquad+P_X(x_1)(P_{Y|X}(y_1|x_1)-P_{Y|X}(y_1|x_2))
\\&=q\left(P_{Y|X}(y_1|x_2)+P_{Y|X}(y_2|x_2)\right)\nonumber\\
&\qquad+P_X(x_1)(P_{Y|X}(y_1|x_1)-P_{Y|X}(y_1|x_2))\label{maximizing_j}
\\&=q+P_X(x_1)(P_{Y|X}(y_1|x_1)-P_{Y|X}(y_1|x_2))\label{fq_bound2}
\end{align}
where in \eqref{maximizing_j} we have used the assumption in \eqref{y2x2_y2x1_bound}. Note that the bound in \eqref{fq_bound2} is larger than the bound $f(q)\le q$ for the case where $u_1=u_2$. Thus \eqref{fq_bound2} is an upper bound on $f(q)$ in all cases.
Let us now consider the quantity $\frac{1+f(q)}{1+q}$, separately for $q\le P_X(x_1)$ and for $q\ge P_X(x_1)$. For $q\le P_X(x_1)$, by the bound on $f(q)$ in \eqref{fq_bound1},
\begin{align}
\frac{1+f(q)}{1+q}&\leq\frac{1+q\sum_y \max_x P_{Y|X}(y|x)}{1+q}\label{eqn:thm2boundwithincreasinginq}\\
&\le \frac{1+P_X(x_1)\sum_y \max_x P_{Y|X}(y|x)}{1+P_X(x_1)}\label{thm2:qtopx}
\end{align}
where the inequality \eqref{thm2:qtopx} holds because the right-hand side of \eqref{eqn:thm2boundwithincreasinginq} is non-decreasing in $q$ as $\sum_y \max_x P_{Y|X}(y|x)\ge 1$. For $q\ge P_X(x_1)$, by the bound on $f(q)$ in \eqref{fq_bound2},
\begin{align}
&\frac{1+f(q)}{1+q}\nonumber\\
&\le \frac{1+q+P_X(x_1)(P_{Y|X}(y_1|x_1)-P_{Y|X}(y_1|x_2))}{1+q}\label{rational_function_q0}
\\&\le \frac{1+P_X(x_1)+P_X(x_1)(P_{Y|X}(y_1|x_1)-P_{Y|X}(y_1|x_2))}{1+P_X(x_1)}\label{rational_function_q}
\\&=\frac{1+P_X(x_1)(P_{Y|X}(y_1|x_1)+P_{Y|X}(y_2|x_2))}{1+P_X(x_1)}
\\&=\frac{1+P_X(x_1)\sum_y \max_x P_{Y|X}(y|x)}{1+P_X(x_1)}
\end{align}
where \eqref{rational_function_q} holds because the right-hand side of \eqref{rational_function_q0} is non-increasing in $q$ since $P_X(x_1)(P_{Y|X}(y_1|x_1)-P_{Y|X}(y_1|x_2))\ge 0$.
This proves that
\begin{align}
\mathcal{L}_g^{\text{max}}(X\to Y)&=\sup_{q\in[0,1]}\log \frac{1+f(q)}{1+q}\\
&\le\log\frac{1+P_X(x_1)\sum_y \max_x P_{Y|X}(y|x)}{1+P_X(x_1)}.
\end{align}
By the assumption that $P_X(x_1)=\min_x P_X(x)=p^*$, we have proven an upper bound on $\mathcal{L}_g^{\text{max}}(X\rightarrow Y)$ matching the lower bound in \eqref{eqn:thm2-lowerbound}.

}
}}

\section{Proofs for Section~\ref{section:multiple-guesses}}
\subsection{Proof of Proposition~\ref{thm:expected-alpha-loss-singleguess}}\label{proofofthm1}
Let $P_{\tilde{X}}(x):=\frac{\mathrm{P}(\cup_{i=1}^k (\hat{X}_i=x))}{k}$. Consider the following simplification to the expected $\alpha$-loss under $k$ guesses.
\begin{align}
&\mathbb{E}_{X}\left[\ell_\alpha(\mathrm{P}(\cup_{i=1}^k (\hat{X}_i=X)))\right]\nonumber\\
    &=\mathbb{E}\left[\ell_\alpha(k{P}_{\tilde{X}}(X)) \right]\nonumber
    \\
    &=\frac{\alpha}{\alpha-1}\sum_{x}P_{X}(x)\left(1-(kP_{\tilde{X}}(x))^{\frac{(\alpha-1)}{\alpha}}\right)\label{eqn:lemma2}\\
    &=\frac{\alpha}{\alpha-1}(1-\sum_{x}P_{X}(x)(kP_{\tilde{X}}(x))^{\frac{(\alpha-1)}{\alpha}})\\
    &=\frac{\alpha}{\alpha-1}\bigg(1-k^{\frac{\alpha-1}{\alpha}}(\sum_xP_{X}(x)^\alpha)^\frac{1}{\alpha}+k^{\frac{\alpha-1}{\alpha}}(\sum_xP_{X}(x)^\alpha)^\frac{1}{\alpha}\nonumber\\
    &\hspace{12pt}-\sum_{x}P_{X}(x)(kP_{\tilde{X}}(x))^{\frac{(\alpha-1)}{\alpha}}\bigg)\\
    &=\frac{\alpha}{\alpha-1}\bigg(1-k^{\frac{\alpha-1}{\alpha}}\mathrm{e}^{(\frac{1-\alpha}{\alpha}H_\alpha(X))}+{\blue{k^{\frac{\alpha-1}{\alpha}}}}(\sum_xP_{X}(x)^\alpha)^\frac{1}{\alpha}\nonumber\\
    &\hspace{12pt}-\sum_{x}P_{X}(x)(kP_{\tilde{X}}(x))^{\frac{(\alpha-1)}{\alpha}}\bigg)\\
    &=\frac{\alpha}{\alpha-1}\bigg(1-k^{\frac{\alpha-1}{\alpha}}\mathrm{e}^{(\frac{1-\alpha}{\alpha}H_\alpha(X))}+k^{\frac{\alpha-1}{\alpha}}(\sum_xP_{X}(x)^\alpha)^\frac{1}{\alpha}\nonumber\\
    &\hspace{12pt}-(\sum_xP_{X}(x)^\alpha)^\frac{1}{\alpha}\sum_x(\frac{P_{X}(x)^\alpha}{\sum\limits_{\blue{x'}}P_{X}({\blue{x'}})^\alpha})^\frac{1}{\alpha}(k{P}_{\tilde{X}}(x))^{(1-\frac{1}{\alpha})}\bigg)\\
     &=\frac{\alpha}{\alpha-1}\bigg(1-k^{\frac{\alpha-1}{\alpha}}\mathrm{e}^{\left(\frac{1-\alpha}{\alpha}H_\alpha(X)\right)}+k^{\frac{\alpha-1}{\alpha}}(\sum_xP_{X}(x)^\alpha)^\frac{1}{\alpha}\nonumber\\
     &\hspace{12pt}\times\big(1-\sum_x\big(\frac{P_{X}(x)^\alpha}{\sum\limits_{\blue{x'}}P_{X}({\blue{x'}})^\alpha}\big)^\frac{1}{\alpha}{P}_{\tilde{X}}(x)^{(1-\frac{1}{\alpha})}\big)\bigg)\\
     &=\frac{\alpha}{\alpha-1}\bigg(1-k^{\frac{\alpha-1}{\alpha}}\mathrm{e}^{\left(\frac{1-\alpha}{\alpha}H_\alpha(X)\right)}+k^{\frac{\alpha-1}{\alpha}}(\sum_xP_{X}(x)^\alpha)^\frac{1}{\alpha}\nonumber\\
     &\hspace{12pt}\times(1-\sum_xP_X^{(\alpha)}(x)^\frac{1}{\alpha}{P}_{\tilde{X}}(x)^{(1-\frac{1}{\alpha})})\bigg)\\
     %
    &=\frac{\alpha}{\alpha-1} \bigg(1-k^{\frac{\alpha-1}{\alpha}}\mathrm{e}^{\frac{1-\alpha}{\alpha}H_\alpha(X)}+k^{\frac{\alpha-1}{\alpha}}(\sum_xP_X(x)^\alpha)^{\frac{1}{\alpha}}\nonumber\\
    &\hspace{12pt}\times\big(1-\mathrm{e}^{\frac{1-\alpha}{\alpha}D_\frac{1}{\alpha}(P^{(\alpha)}_X\|P_{\tilde{X}})}\big)\bigg)\label{eqn:thm1eqn1}\\
    &=\frac{\alpha}{\alpha-1} \left(1-k^{\frac{\alpha-1}{\alpha}}\mathrm{e}^{\frac{1-\alpha}{\alpha}H_\alpha(X)}\right)+k^{\frac{\alpha-1}{\alpha}}B_F(P_X,P^{(\frac{1}{\alpha})}_{\hat{X}})\label{eqn:thm1proof1new},
\end{align}
where \eqref{eqn:thm1proof1new} follows from Appendix~\ref{appx:Bregman} with Bregman divergence $B_F$ associated with $F(P_X)=\frac{\alpha}{\alpha-1}\left((\sum_xP_X(x)^\alpha)^{\frac{1}{\alpha}}-1\right)$. In view of Lemma~\ref{fact:alphaloss2noneq} and Remark~\ref{remark} in {\nblue{Appendix}}~\ref{proofofthm2}, we can consider $P_{\tilde{X}}$ to be a probability distribution, i.e., $\sum_xP_{\tilde{X}}(x)=1$, for optimizing the expected $\alpha$-loss. Now from the non-negativity of the R\'{e}nyi divergence, it can be seen that the last term inside the brackets in \eqref{eqn:thm1eqn1} is non-negative and is equal to zero if and only if there exists a guessing strategy $P_{\hat{X}_{[1:k]}}$ such that $P_{\tilde{X}}=P_X^{(\alpha)}$. There always exists such a guessing strategy for the special case of $k=1$, in particular, $P_{\hat{X}_1}=P_{X}^{(\alpha)}$, thereby giving an alternative proof of \cite[Lemma~1]{LiaoKS20} in addition to quantifying the relative performance of an arbitrary guessing strategy with respect to the optimal guessing strategy. However, this is not always possible when $k>1$. In particular, it follows from Lemma~\ref{theorem:probclass1} and the discussion above it (in {\nblue{Appendix}}~\ref{proofofthm2}) that there exists a guessing strategy $P_{\hat{X}_{[1:k]}}$ such that $P_{\tilde{X}}=P_X^{(\alpha)}$ if and only if $P_{X}^{(\alpha)}(x)\leq \frac{1}{k}$, for all $x\in\mathcal{X}$. Then the minimal expected $\alpha$-loss is given by $ \mathcal{ME}_\alpha^{(k)}(P_X)=\frac{\alpha}{\alpha-1}(1-k^{\frac{\alpha-1}{\alpha}}\mathrm{e}^{\frac{1-\alpha}{\alpha}H_\alpha(X)})$ when $P_{X}^{(\alpha)}(x)\leq \frac{1}{k}$, for all $x\in\mathcal{X}$.

\subsection{Minimal expected log-loss under $2$ guesses}\label{logloss2}
{\blue{Here we present a proof of Theorem~\ref{thm:minloss-alpha-kguesses} for the special case of $\alpha=1$ and $k=2$ that essentially captures the intuition and the main tools involved in the relatively complex analysis in the proof of Theorem~\ref{thm:minloss-alpha-kguesses}.}}
\begin{theorem}[Minimal Expected log-loss \{$\alpha=1$\} under $2$ guesses]\label{thm:minimal-logloss2}
Let $P_X$ be a probability distribution supported on $\mathcal{X}=\{x_1,x_2,\dots,x_n\}$, and let $p_{\emph{max}}=\max_{i\in[1:n]}P_X(x_i)=P_X(x_\emph{max})$. Then the minimal expected log-loss under 2 guesses is given by
\begin{align}\label{eqn:logloss2thm}
    \min_{P_{\hat{X}_1\hat{X}_2}}\mathbb{E}&\left[\log{\frac{1}{\mathrm{P}\left(\bigcup_{i=1}^2(\hat{X}_i=X)\right)}}\right] \nonumber\\
    &=H(X)-h_\emph{b}\left(\max{\{\frac{1}{2},p_\emph{max}\}}\right),
\end{align}
where $h_\emph{b}(\cdot)$ is the binary entropy function defined by $h_\emph{b}(t):=-t\log{t}-(1-t)\log{(1-t)}$ and the optimal guessing strategy is given by $P_{\hat{X}_1,\hat{X}_2}$ such that
\begin{align}\label{eqn:optimalstrategylog}
\begin{cases}
\mathrm{P}(\hat{X}_1=x\ \emph{or}\ \hat{X}_2=x)=2P_X(x), \\
\hspace{12pt}\emph{if} \ P_{X}(x_{\emph{max}})\leq \frac{1}{2},\\
P_{\hat{X}_1,\hat{X}_2}(x_\emph{max},x)+P_{\hat{X}_1,\hat{X}_2}(x,x_\emph{max})=\frac{P_X(x)}{\sum\limits_{x^\prime:x^\prime\neq x_\emph{max}}P_X(x^\prime)}, \\
\hspace{12pt}x\neq x_{\emph{max}}, \emph{if}\  P_X(x_{\emph{max}})>\frac{1}{2}.
\end{cases}
\end{align}
\end{theorem}
\begin{remark}
In words, the optimal guessing strategy in \eqref{eqn:optimalstrategylog} is to guess $x$ in either of the guesses with a probability proportional to $P_{X}(x)$, for all $x$, if $p_{\text{max}}\leq \frac{1}{2}$. If $p_{\text{max}}>\frac{1}{2}$, the optimal strategy is to guess $x_\text{max}$ with probability $1$ (i.e., deterministically) and randomly guess the other $x$ with a probability proportional to $P_X(x)$, for $x\neq x_{\text{max}}$. 
\end{remark}
It can be inferred from \eqref{eqn:logloss2thm} that the minimal expected log-loss under 2 guesses induces a dichotomy on the simplex of probability distributions $P_X$ on $\mathcal{X}$. 
\begin{proof}[Proof of Theorem~\ref{thm:minimal-logloss2}]
We first state some useful lemmas which will be needed in the proof.
\begin{lemma}[Non-negativity]\label{fact:positivityrelent}
If $\tilde{P}_X$ is a probability distribution and $\tilde{Q}_X$ is such that $\tilde{Q}_X(x)\geq 0$ for all $x$ and $\sum\limits_x\tilde{Q}_X(x)\leq 1$ with the same support set as $\tilde{P}_X$. Then $D(\tilde{P}_X||\tilde{Q}_X)\geq 0$.
\end{lemma}
Lemma~\ref{fact:positivityrelent} is proved in Appendix~\ref{appndx:proofsof-facts}.
\begin{lemma}[Non-equality]\label{fact:non-equal}
If $P^*_{\hat{X}_1,\hat{X}_2}$ is an optimal strategy for the optimization problem in \eqref{eqn:logloss2thm}, then 
\begin{align}
    P^*_{\hat{X}_1,\hat{X}_2}(x,x)=0,\ \text{for all}\ x.
\end{align}
\end{lemma}
Lemma~\ref{fact:non-equal} follows from Lemma~\ref{fact:alphaloss2noneq} for the special case of $k=2$.
\begin{lemma}\label{fact:factfromFarkas}
There exists $P_{\hat{X}_1\hat{X}_2}$ such that $\mathrm{P}\left(\hat{X}_1=x\ \emph{or} \ \hat{X}_2=x\right)=2P_X(x)$, for all $x$ if and only if $P_X(x)\leq \frac{1}{2}$, for all $x$.
\end{lemma}
Lemma~\ref{fact:factfromFarkas} follows from Lemma~\ref{theorem:probclass1} for the special case of $k=2$. A consequence of  Lemma~\ref{fact:non-equal} is that if $P^*_{\hat{X}_1\hat{X}_2}$ is an optimal strategy, then we have 
\begin{align}\label{eqn:nonequalimpl}
    \sum_x\mathrm{P}^*\left(\hat{X}_1=x\ \text{or} \ \hat{X}_2=x\right)=2,
\end{align}
where the probability $\mathrm{P}^*$ is taken with respect to the optimal strategy $P^*_{\hat{X}_1\hat{X}_2}$. So, it suffices to consider the optimization in \eqref{eqn:logloss2thm} over the strategies $P_{\hat{X}_1\hat{X}_2}$ satisfying \eqref{eqn:nonequalimpl}. Notice that for such strategies, $P_{\tilde{X}}(x):=\frac{\Pr\left(\hat{X}_1=x\  \text{or}\  \hat{X}_2=x \right)}{2}$ is a probability distribution. Now we are ready to prove Theorem~\ref{thm:minimal-logloss2}. Let $p_{\max}\leq \frac{1}{2}$. Then we have
\begin{align}
&\mathbb{E}\left[\ell_1(\textrm{P}(\hat{X}_1=x\ \text{or} \ \hat{X}_2=x))\right]\nonumber\\
&=\sum_{x}P_{X}(x)\log{\frac{1}{\mathrm{P}(\hat{X}_1=x\ \text{or}\ \hat{X}_2=x)}}\\
&=\sum_{x}P_{X}(x)\log{\frac{1}{P_{X}(x)}}\nonumber\\
&\hspace{12pt}+\sum_{x}P_{X}(x)\log{\frac{P_{X}(x)}{\mathrm{P}(\hat{X}_1=x\ \text{or}\ \hat{X}_2=x)}}\\
&=H(X)-1+\sum_{x}P_{X}(x)\log{\frac{2P_{X}(x)}{\mathrm{P}(\hat{X}_1=x\ \text{or}\ \hat{X}_2=x)}}\\
&=H(X)-1+D(P_{X}\|P_{\tilde{X}})\label{eqn:defnest}\\
&\geq H(X)-1\label{eqn:relentnonneg}
\end{align}
where we have defined $P_{\tilde{X}}(x)=\frac{\Pr\left(\hat{X}_1=x\  \text{or}\  \hat{X}_2=x \right)}{2}$ in \eqref{eqn:defnest}, and \eqref{eqn:relentnonneg} follows from the non-negativity of the relative entropy. We remark that \eqref{eqn:defnest} can be obtained from \eqref{eqn:singleguess-altform} also by substituting $k=2$ and taking limit $\alpha\rightarrow 1$. Notice that the inequality in \eqref{eqn:relentnonneg} can be tight if and only if there exists $P_{\hat{X}_1\hat{X}_2}$ such that $P_{\tilde{X}}(x)=2P_X(x)$, for all $x$, which is possible if and only if $p_{\text{max}}\leq \frac{1}{2}$ using Lemma~\ref{fact:factfromFarkas}.

Now, consider the case when $p_\text{max}>\frac{1}{2}$. Without loss of generality, suppose that $p_\text{max}=P_X(x_1)$. Let $t_x=\mathrm{P}(\hat{X}_1=x\ \text{or}\ \hat{X}_2=x)$. For example, $t_{x_1}=\sum\limits_{j=2}\limits^np_{1j}+\sum\limits_{j=2}\limits^np_{j1}=2\sum\limits_{j=2}\limits^np_{1j}$, where $P_{\hat{X}_1,\hat{X}_2}(x_i,x_j)=p_{ij}$. Also, note that $\sum\limits_{i=1}\limits^nt_x=2$ in view of Lemma~\ref{fact:non-equal}. Then we have
\begin{align}
&\mathbb{E}\left[\ell_{\text{log}}(X,P_{\hat{X}_1,\hat{X}_2})\right]\nonumber\\
&=P_X(x_1)\log{\frac{1}{t_{x_1}}}+\sum_{i=2}^nP_X(x_i)\log{\frac{1}{t_{x_i}}}\\
&=(P_X(x_1)-\sum_{i=2}^nP_X(x_i))\log{\frac{1}{t_{x_1}}}+(\sum_{i=2}^nP_X(x_i))\log{\frac{1}{t_{x_1}}}\nonumber\\
&\hspace{12pt}+\sum_{i=2}^nP_X(x_i)\log{\frac{1}{(\frac{P_X(x_i)}{\sum\limits_{j=2}\limits^nP_X(x_j)})}}\nonumber\\
&\hspace{12pt}+\sum_{i=2}^nP_X(x_i)\log{\frac{{(\frac{P_X(x_i)}{\sum\limits_{j=2}\limits^nP_X(x_j)})}}{t_{x_i}}}\\
&=\sum_{i=2}^nP_X(x_i)\log{\frac{1}{(\frac{P_X(x_i)}{\sum\limits_{j=2}\limits^nP_X(x_j)})}}\nonumber\\
&\hspace{12pt}+(P_X(x_1)-\sum_{i=2}^nP_X(x_i))\log{\frac{1}{t_{x_1}}}\nonumber\\
&\hspace{12pt}+\sum_{i=2}^nP_X(x_i)\log{\frac{{(\frac{P_X(x_i)}{\sum\limits_{j=2}\limits^nP_X(x_j)})}}{t_{x_1}t_{x_i}}}\label{eqn:thm21}
\end{align}
The second term in \eqref{eqn:thm21} is non-negative since $P_X(x_1)>0.5$ is equivalent to $P_X(x_1)>\sum\limits_{i=2}\limits^nP_X(x_i)$. Consider the third term in \eqref{eqn:thm21}.
\begin{align}
    &\sum_{i=2}^nP_X(x_i)\log{\frac{{(\frac{P_X(x_i)}{\sum\limits_{j=2}\limits^nP_X(x_j)})}}{t_{x_1}t_{x_i}}}\nonumber\\
    &= (\sum\limits_{j=2}\limits^nP_X(x_j))\sum_{i=2}^n(\frac{P_X(x_i)}{\sum\limits_{j=2}\limits^nP_X(x_j)})\log{\frac{{(\frac{P_X(x_i)}{\sum\limits_{j=2}\limits^nP_X(x_j)})}}{t_{x_1}t_{x_i}}}\\
    &=(\sum\limits_{j=2}\limits^nP_X(x_j))D(\tilde{P}_X\|\tilde{Q}_X)\label{eqn:relax1}\\
    &\geq 0\label{eqn:relax2},
\end{align}
where \eqref{eqn:relax1} follows by defining $\tilde{P}_X(x_i)=\frac{P_X(x_i)}{\sum\limits_{j=2}\limits^nP_X(x_j)}$ and $\tilde{Q}_X(x_i)=t_{x_1}t_{x_i}$, $i=[2:n]$, \eqref{eqn:relax2} follows from Lemma~\ref{fact:positivityrelent} noticing that $\sum\limits_{i=2}\limits^n\tilde{Q}_X(x_i)=t_{x_1}(\sum\limits_{i={\blue{2}}}\limits^nt_{x_i})=t_{x_1}(2-t_{x_1})\leq 1$. Equality in \eqref{eqn:relax2} is attained if and only if $t_{x_1}=1$ and $p_{1i}+p_{i1}=\frac{P_X(x_i)}{\sum\limits_{j=2}\limits^nP_X(x_j)},i\in[2:n]$. Under this condition, note that the second term in \eqref{eqn:thm21} is also zero and the first term simplifies to $H(X)-h_{\text{b}}(p_{\text{max}})$.
\end{proof}

\subsection{Proof of Theorem~\ref{thm:minloss-alpha-kguesses}}\label{proofofthm2}
We begin with the following lemmas which will be useful in the proof of Theorem~\ref{thm:minloss-alpha-kguesses}. It is intuitive to expect that an optimal strategy, $P^*_{\hat{X}_{[1:k]}}$, puts zero weight on ordered tuples $(a_1,a_2,\dots,a_k)$ (denoted as $a_{[1:k]}$ in the sequel) whenever $a_i=a_j$ for some $i\neq j$, since there is no advantage in guessing the same estimate more than once. The following lemma based on the monotonicity of the $\alpha$-loss formalizes this.
\begin{lemma}\label{fact:alphaloss2noneq}
If $P^*_{\hat{X}_{[1:k]}}$ is an optimal strategy for the optimization problem in \eqref{eqn:optimizationproblem}, then 
\begin{align*}
    P^*_{\hat{X}_{[1:k]}}(a_{[1:k]})=0, \ \text{for all} \ a_{[1:k]} \ \text{s.t.} \ a_i=a_j, \ \text{for some} \ i\neq j.
\end{align*}
\end{lemma}
The proof of Lemma~\ref{fact:alphaloss2noneq} is deferred to Appendix~\ref{app1}.

\begin{remark}\label{remark}
An important consequence of Lemma~\ref{fact:alphaloss2noneq} is that, if $P^*_{\hat{X}_{[1:k]}}$ is an optimal strategy for the optimization problem in \eqref{eqn:optimizationproblem}, then we have
\begin{align}\label{eqn:nonequality}
    \sum_x\mathrm{P}^*\left(\bigcup_{i=1}^k(\hat{X}_i=x)\right)=k,
\end{align}
where the probability $\mathrm{P}^*$ is taken with respect to an optimal strategy $P^*_{\hat{X}_{[1:k]}}$. Hence, it suffices to consider the optimization in \eqref{eqn:optimizationproblem} over all the strategies $P_{\hat{X}_{[1:k]}}$ satisfying \eqref{eqn:nonequality}.
\end{remark}
 A vector $(t_1,t_2,\dots,t_n)$ such that $\sum_{i=1}^nt_i=k$ is said to be \emph{admissible} if there exists a strategy $P_{\hat{X}_{[1:k]}}$ satisfying

\begin{align}\label{eqn:systemprob}
    t_i=\mathrm{P}\left(\bigcup_{j=1}^k(\hat{X}_j=x_i)\right ),\ \text{for all}\ i\in[1:n].
\end{align}
Equivalently, \eqref{eqn:systemprob} can be written as the following system of linear equations.
\begin{align}\label{eqn:linsystemk}
    t_i=\sum_{a_{[1:k]}:\bigcup\limits_{j=1}^k(a_j=x_i)} P_{\hat{X}_{[1:k]}}(a_{[1:k]}), \ \text{for all}\ i\in[1:n].
\end{align}
 In general, in order to determine whether a vector $(t_1,t_2,\dots,t_n)$ is admissible or not, we need to solve a linear programming problem (LPP) with number of variables and constraints that are polynomial in the support size of $P_X$, i.e, $n$. Nonetheless, the following lemma based on Farkas' lemma~\cite[Proposition 6.4.3]{Matousek} completely characterizes the necessary and sufficient conditions for the admissibility of a vector $(t_1,t_2,\dots,t_n)$.
\begin{lemma}\label{theorem:probclass1}
A vector $(t_1,t_2,\dots,t_n)$ such that $\sum\limits_{i=1}^nt_i=k$ is admissible if and only if $0\leq t_i\leq 1$, for all $i\in[1:n]$.
\end{lemma}
The proof of Lemma~\ref{theorem:probclass1} is deferred to Appendix~\ref{app2}. We now prove Theorem~\ref{thm:minimal-logloss2}.
From the definition of the minimal expected $\alpha$-loss for $k$ guesses in \eqref{eqn:optimizationproblem}, we have 
\begin{align}
    &\mathcal{ME}^{(k)}_\alpha(P_X)\nonumber\\
    &=\min_{P_{\hat{X}_{[1:k]}}}\frac{\alpha}{\alpha-1}\text{\small$\left[\sum_{i=1}^np_i\left(1-\mathrm{P}\left(\bigcup_{j=1}^k(\hat{X}_j=x_i)\right)^\frac{\alpha-1}{\alpha}\right)\right]$}\\
    &=\min_{P_{\hat{X}_{[1:k]}}}\frac{\alpha}{\alpha-1}\text{\small$\left[\sum_{i=1}^np_i\left(1-\mathrm{P}\left(\bigcup_{j=1}^k(\hat{X}_j=x_i)\right)^\frac{\alpha-1}{\alpha}\right)\right]$}\nonumber\\
  &\hspace{24pt}\text{s.t.}\ \sum\limits_{i=1}^n\mathrm{P}\left(\bigcup_{j=1}^k(\hat{X}_j=x_i)\right)=k\label{eqn:thmnoneq}\\
  &=\min_{t_1,\dots,t_n}\frac{\alpha}{\alpha-1}\left[\sum_{i=1}^np_i(1-t_i^{\frac{\alpha-1}{\alpha}})\right]\nonumber\\
  &\hspace{24pt}\text{s.t.}\ \sum_{i=1}^nt_i=k,\nonumber\\
  &\hspace{43pt}0\leq t_i\leq 1,\ i\in[1:n]\label{eqn:thmduality1},
\end{align}
where \eqref{eqn:thmnoneq} follows from Lemma~\ref{fact:alphaloss2noneq} and Remark~\ref{remark}, and \eqref{eqn:thmduality1} follows from  the change of variable $t_i=\mathrm{P}\left(\bigcup\limits_{j=1}^k(\hat{X}_j=x_i)\right)$ and Lemma~\ref{theorem:probclass1}.
Consider the Lagrangian 
\begin{align}
   \mathcal{L}&=\frac{\alpha}{\alpha-1}\left[\sum_{i=1}^np_i(1-t_i^{\frac{\alpha-1}{\alpha}})\right]+\lambda\left(\sum_{i=1}^nt_i-k\right)\nonumber\\
   &\hspace{12pt}+\sum_{i=1}^n\mu_i(t_i-1)
\end{align}
The Karush-Kuhn-Tucker (KKT) conditions~\cite[Chapter~5.5.3]{boyd_vandenberghe_2004} are given by
\begin{align}
    &\text{(Stationarity):}\ \frac{\partial\mathcal{L}}{\partial t_i}=0,i\in[1:n],\nonumber\\ &\text{i.e.,}\ t_i=\left(\frac{p_i}{\lambda+\mu_i}\right)^\alpha,i\in[1:n],\label{stationarity}\\
   &\text{(Primal feasibility):}\ \sum_{i=1}^nt_i=k,
    0\leq t_i\leq 1, i\in[1:n],\label{eqn:primalfeasibility}\\
   &\text{(Dual feasibility):}\ \mu_i\geq 0,i\in[1:n],\label{dualfeasibility}\\
&\text{(Complementary slackness):}\ \mu_i(t_i-1)=0, i\in[1:n]\label{slackness}.
\end{align}
Notice that for $\alpha>1$, $t^{\frac{\alpha-1}{\alpha}}$ is a concave function of $t$, meaning the overall objective function in \eqref{eqn:thmduality1} is convex. For $\alpha<1$, $t^{\frac{\alpha-1}{\alpha}}$ is a convex function of $t$, but since $\frac{\alpha}{\alpha-1}$ is negative, the overall function is again convex. Thus \eqref{eqn:thmduality1} amounts to a convex optimization problem. Now since KKT conditions are necessary and sufficient conditions for optimality in a convex optimization problem, it suffices to find values of $t_i$, $i\in[1:n]$, $\lambda$, $\mu_i$, $i\in[1:n]$ satisfying \eqref{stationarity}--\eqref{slackness} in order to solve the optimization problem~\eqref{eqn:thmduality1}. 

First we simplify the KKT conditions \eqref{stationarity}--\eqref{slackness} in the following manner. 
\begin{itemize}
\item For $i$ such that $\left(\frac{p_i}{\lambda}\right)^\alpha\leq 1$, we take $\mu_i=0$ and $t_i=\left(\frac{p_i}{\lambda}\right)^\alpha$.
\item For $i$ such that $\left(\frac{p_i}{\lambda}\right)^\alpha> 1$, we take $\mu_i=p_i-\lambda$ and $t_i=1$. Notice that for such $i$, we have $\mu_i>0$, since $p_i>\lambda$.
\end{itemize}
This is equivalent to choosing $t_i=\min\left\{\left(\frac{p_i}{\lambda}\right)^\alpha,1\right\}$ and $\mu_i=0$ or $\mu_i=p_i-\lambda$ depending on whether $t_i=\left(\frac{p_i}{\lambda}\right)^\alpha$ or $t_i=1$, respectively, for each $i\in[1:n]$. Notice that this choice is consistent with the KKT conditions \eqref{stationarity}--\eqref{slackness} except for that $\lambda$ has to be chosen appropriately satisfying $\sum_{i=1}^nt_i=k$ also. In effect, we have essentially reduced the KKT conditions \eqref{stationarity}--\eqref{slackness} to the following equations by eliminating $\mu_i$'s:
\begin{align}
    &t_i=\min\left\{\left(\frac{p_i}{\lambda}\right)^\alpha,1\right\},i\in[1:n],\label{eqn:new1}\\
    &\sum_{i=1}^nt_i=k\label{eqn:new2}.
\end{align}

We solve the equations \eqref{eqn:new1} and \eqref{eqn:new2} by considering the following $k$ mutually exclusive and exhaustive cases (clarified later) based on $P_X$. 

\medskip
\noindent \underline{Case $1$}  $\left(\frac{p_1^\alpha}{\sum_{i=1}^np_i^\alpha}\leq \frac{1}{k}\right)$:\\
Consider the choice
\begin{align}
\lambda=\left(\frac{\sum_{i=1}^np_i^\alpha}{k}\right)^\frac{1}{\alpha},\ 
    t_i=\frac{kp_i^\alpha}{\sum_{j=1}^np_j^\alpha},i\in[1:n].
\end{align}
This choice satisfies \eqref{eqn:new1} and \eqref{eqn:new2} since $\frac{kp_1^\alpha}{\sum_{i=1}^np_i^\alpha}\leq 1$ and $p_1\geq p_2\dots\geq p_n$. 

\bigskip

\noindent\underline{Case `$s$'} ($2\leq s\leq k$) $\left(\frac{(k-s+2)p_{s-1}^\alpha}{\sum_{i=s-1}^np_i^\alpha}>1, \frac{(k-s+1)p_s^\alpha}{\sum_{i=s}^np_i^\alpha}\leq 1\right)$:\\
Consider the choice
\begin{align}
    \lambda&=\left(\frac{\sum_{i=s}^np_i^\alpha}{k-s+1}\right)^\frac{1}{\alpha},\\
    t_i&=1, i\in[1:s-1], t_i=\frac{(k-s+1)p_i^\alpha}{\sum_{j=s}^np_j^\alpha}, i\in[s:n]\label{eqn:newcam}.
\end{align}
This choice satisfies \eqref{eqn:new1}
\begin{itemize}
    \item for $i\in[1:s-1]$ because $\frac{(k-s+2)p_{s-1}^\alpha}{\sum_{i=s-1}^np_i^\alpha}>1$ and $p_1\geq p_2\geq\dots\geq p_{s-1}$, and
    \item for $i\in[s:n]$ because $\frac{(k-s+1)p_s^\alpha}{\sum_{i=s}^np_i^\alpha}\leq 1$ and $p_s\geq p_{s+1}\geq\dots\geq p_n$.
\end{itemize} 
Also, this choice clearly satisfies \eqref{eqn:new2}. 
Finally, notice that the condition for Case `$s$', $2\leq s\leq n$, can be written as 
\begin{align}
    \frac{(k-i+1)p_{i}^\alpha}{\sum\limits_{j=i}^np_j^\alpha}>1, \ \text{for}\ i\in[1:s-1], \frac{(k-s+1)p_s^\alpha}{\sum\limits_{i=s}^np_i^\alpha}\leq 1
\end{align}
since $\frac{(k-s+2)p_{s-1}^\alpha}{\sum\limits_{i=s-1}^np_i^\alpha}>1$ and $p_1\geq p_2\geq\dots\geq p_{s-1}$. This proves that the cases considered above are mutually exclusive and exhaustive, and together with the case-wise analysis gives the expression for the minimal expected $\alpha$-loss for $k$ guesses as presented in Theorem~\ref{thm:minloss-alpha-kguesses}.

\subsection{Proof of Theorem~\ref{thm:robustness-of-alphaleakage}}\label{proofofthm3}
From the definition of $\alpha$-leakage with $k$ guesses in \eqref{defn:kalphaleakage}, we have
\begin{align}
    &\mathcal{L}^{(k)}_\alpha(X\rightarrow Y)\nonumber\\ &=\frac{\alpha}{\alpha-1}\log{\frac{\max\limits_{P_{\hat{X}_{[1:k]}|Y}}\mathbb{E}\left[{\nblue{\frac{\alpha}{\alpha-1}}}\mathrm{P}\left(\bigcup\limits_{i=1}^k(\hat{X}_i=X)|Y\right)^{\frac{\alpha-1}{\alpha}}\right]}{\max\limits_{P_{\hat{X}_{[1:k]}}}\mathbb{E}\left[{\nblue{\frac{\alpha}{\alpha-1}}}\mathrm{P}\left(\bigcup\limits_{i=1}^k(\hat{X}_i=X)\right)^{\frac{\alpha-1}{\alpha}}\right]}}\\
    &=\frac{\alpha}{\alpha-1}\log{\frac{{\nblue{\frac{\alpha}{\alpha-1}}}k^{\frac{\alpha-1}{\alpha}}\mathrm{e}^{\frac{1-\alpha}{\alpha}H_\alpha^A(X|Y)}}{{\nblue{\frac{\alpha}{\alpha-1}}}k^{\frac{\alpha-1}{\alpha}}\mathrm{e}^{\frac{1-\alpha}{\alpha}H_\alpha(X)}}}\label{eqn:leakage1}\\
     &=\frac{\alpha}{\alpha-1}\log{\frac{\mathrm{e}^{\frac{1-\alpha}{\alpha}H_\alpha^A(X|Y)}}{\mathrm{e}^{\frac{1-\alpha}{\alpha}H_\alpha(X)}}}\\
    &=\mathcal{L}_\alpha^{(1)}(X\rightarrow Y),
\end{align}
where \eqref{eqn:leakage1} follows from Theorem~\ref{thm:minloss-alpha-kguesses}, in particular from the case when $s^*=1$ since $P_{X|Y}^{(\alpha)}(x|y)\leq \frac{1}{k}$, for all $x,y$ and $P_{X}^{(\alpha)}(x)\leq \frac{1}{k}$, for all $x$, and $H_\alpha^A(X|Y)$ in \eqref{eqn:leakage1} is the Arimoto conditional entropy~\cite{arimoto1977information} defined as $H_\alpha^A(X|Y)=\frac{\alpha}{1-\alpha}\log{\sum\limits_y\left(\sum\limits_xP_{XY}(x,y)^\alpha\right)^\frac{1}{\alpha}}$.

\subsection{Proof of Theorem~\ref{thm:lowerbound-maxalphaleakage-kguesses}}\label{proofofthm4}
Consider 
\begin{align}\label{eqn:maxleakagek}
   \mathcal{L}_\alpha^{(1)-{\text{max}}}(X\rightarrow Y)=\sup_{U-X-Y}I_\alpha^A(U;Y). 
\end{align}
Then we have the following lemma proved later.
\begin{lemma}\label{claim}
There exists an optimizer $P_{U^*|X}$ for the optimization problem in the RHS of \eqref{eqn:maxleakagek} such that $P_{U^*|Y}^{(\alpha)}(u|y)$, $P_{U^*}^{(\alpha)}(u)\leq\frac{1}{k}$, for all $u,y$, where $P_{{U}^*|Y}^{(\alpha)}({u}|y)=\frac{P_{{U}^*|Y}^{\alpha}({u}|y)}{\sum_{{u}}P_{{U^*}|Y}^{\alpha}({u}|y)}$ and $P_{{U^*}}^{(\alpha)}({u})=\frac{P_{{U}^*}^{\alpha}({u)}}{\sum_{{u}}P_{{U}^*}^{\alpha}({u})}$.
\end{lemma}
Then for $P_{U^*|X}$ in the Lemma~\ref{claim}, by definition we have
\begin{align}
  &\mathcal{L}_\alpha^{(k)-{\text{max}}}(X\rightarrow Y)\nonumber\\
  &\geq \frac{\alpha}{\alpha-1}\log{\frac{\max_{P_{\hat{U}_{[1:k]}|Y}}\mathbb{E}[{\nblue{\frac{\alpha}{\alpha-1}}}\mathrm{P}(\bigcup_{i=1}^k(\hat{U}_i=U^*)|Y)^{\frac{\alpha-1}{\alpha}}]}{\max_{P_{\hat{U}_{[1:k]}}}\mathbb{E}[{\nblue{\frac{\alpha}{\alpha-1}}}\mathrm{P}(\bigcup_{i=1}^k\hat{U}_i=U^*)^{\frac{\alpha-1}{\alpha}}]}}\\
  &=\frac{\alpha}{\alpha-1}\log{\frac{{\nblue{\frac{\alpha}{\alpha-1}}}k^{\frac{\alpha-1}{\alpha}}{\mathrm{e}^{\frac{1-\alpha}{\alpha}H_\alpha^A(U^*|Y)}}}{{\nblue{\frac{\alpha}{\alpha-1}}}k^{\frac{\alpha-1}{\alpha}}\mathrm{e}^{\frac{1-\alpha}{\alpha}H_\alpha(U^*)}}}\label{eqn:isit}\\
  &=I_\alpha^A(U^*;Y)\\
  &=\mathcal{L}_\alpha^{(1)-{\text{max}}}(X\rightarrow Y)\label{eqn:claim1},
  \end{align}
  where \eqref{eqn:isit} follows from Theorem~\ref{thm:minloss-alpha-kguesses}, \eqref{eqn:claim1} follows because $P_{U^*|X}$ is an optimizer in \eqref{eqn:maxleakagek}.
  It remains to prove Lemma~\ref{claim}.
  
  \noindent\emph{Proof of Lemma~\ref{claim}:} Note that it suffices to prove that for every $P_{U|X}$, there exists $P_{\tilde{U}|U}$ such that $I_\alpha^A(U;Y)=I_\alpha^A(\tilde{U};Y)$ and $P_{\tilde{U}|Y}^{(\alpha)}(\tilde{u}|y),P^{(\alpha)}_{\tilde{U}}(\tilde{u})\leq \frac{1}{k}$, for all $u,y$.
  To that end, {\blue{we use the ``shattering'' conditional distribution $P_{U|X}$~\cite[Proof of Theorem~1]{Issaetal},\cite[Proof of Theorem~5]{Liaoetal}.}} Let us first define $\tilde{\mathcal{U}}=\bigcup_{u\in\mathcal{U}}\tilde{\mathcal{U}}_u$, with
 $\tilde{\mathcal{U}}_u=\{(u,1),(u,2),\dots,(u,m)\}$ for some $m$ to be fixed later. Let $P_{\tilde{U}|U}(\tilde{u}|u)=\frac{1}{m}$, for $\tilde{u}\in\tilde{\mathcal{U}}_u$. This gives
  \begin{align}
      &P_{\tilde{U}|Y}(\tilde{u}|y)=\sum_u P_{U|Y}(u|y)P_{\tilde{U}|U}(\tilde{u}|u)=\frac{P_{U|Y}(u|y)}{m},\nonumber\\
      &\ \text{for}\ \tilde{u}\in\tilde{\mathcal{U}}_u,\label{eqn:claim1proof1}\\
      &P_{\tilde{U}}(\tilde{u})=\sum_u P_U(u)P_{\tilde{U}|U}(\tilde{u}|u)=\frac{P_U(u)}{m}\ \text{for}\ \tilde{u}\in\tilde{\mathcal{U}}_u,\label{eqn:claim1proof2}\\
      &P_{U|\tilde{U}}(u|\tilde{u})=1\ \text{for}\ \tilde{u}\in\tilde{\mathcal{U}}_u,\label{eqn:claim1proof3}\\
      &P_{Y|\tilde{U}}(y|\tilde{u})=P_{Y|U}(y|u)\ \text{for}\ \tilde{u}\in\tilde{\mathcal{U}}_u\label{eqn:claim1proof4}.
  \end{align}
  Now we have
  \begin{align}
      &I_\alpha^A(\tilde{U};Y)\nonumber\\
      &=\frac{\alpha}{\alpha-1}\log{\frac{\sum_y\left(\sum_u\sum_{\tilde{u}\in\tilde{\mathcal{U}}_u}P_{Y|\tilde{U}}(y|\tilde{u})^\alpha P_{\tilde{U}}(\tilde{u})^\alpha\right)^{\frac{1}{\alpha}}}{\left(\sum_u\sum_{\tilde{u}\in\tilde{\mathcal{U}}_u}P_{\tilde{U}}(\tilde{u})^\alpha\right)^{\frac{1}{\alpha}}}}\\
      &=\frac{\alpha}{\alpha-1}\log{\frac{\sum_y\left(\sum_uP_{Y|{U}}(y|{u})^\alpha\sum_{\tilde{u}\in\tilde{\mathcal{U}}_u} P_{\tilde{U}}(\tilde{u})^\alpha\right)^{\frac{1}{\alpha}}}{\left(\sum_u\sum_{\tilde{u}\in\tilde{\mathcal{U}}_u}P_{\tilde{U}}(\tilde{u})^\alpha\right)^{\frac{1}{\alpha}}}}\label{eqn:claim1proof5}\\
      &=\frac{\alpha}{\alpha-1}\log{\frac{\sum_y\left(\sum_uP_{Y|{U}}(y|{u})^\alpha \frac{P_{{U}}({u})^\alpha}{m^{\alpha-1}}\right)^{\frac{1}{\alpha}}}{\left(\sum_u\frac{P_{{U}}({u})^\alpha}{m^{\alpha-1}}\right)^{\frac{1}{\alpha}}}}\label{eqn:claim1proof6}\\
      &=\frac{\alpha}{\alpha-1}\log{\frac{\sum_y\left(\sum_uP_{Y|{U}}(y|{u})^\alpha P_{{U}}({u})^\alpha\right)^{\frac{1}{\alpha}}}{\left(\sum_uP_{{U}}({u})^\alpha\right)^{\frac{1}{\alpha}}}}\label{eqn:claim1proof7}\\
      &=I_\alpha^A(U;Y),
  \end{align}
  where \eqref{eqn:claim1proof5} follows from \eqref{eqn:claim1proof4}, and \eqref{eqn:claim1proof6} follows from \eqref{eqn:claim1proof2}. Now consider
  \begin{align}
      P^{(\alpha)}_{\tilde{U}|Y}(\tilde{u}|y)&=\frac{P_{\tilde{U}|Y}(\tilde{u}|y)^\alpha}{\sum_u\sum_{\tilde{u}\in\tilde{\mathcal{U}}_u}P_{\tilde{U}|Y}(\tilde{u}|y)^\alpha}\\
      &=\frac{\frac{P_{U|Y}(u|y)^\alpha}{m^\alpha}}{\sum_um\cdot \frac{P_U(u|y)^\alpha}{m^\alpha}}\label{eqn:claim1proof9}\\
      &=\frac{1}{m}\cdot\frac{P_{U|Y}(u|y)^\alpha}{\sum_uP_U(u|y)^\alpha}\\
      &\leq \frac{1}{m}\label{eqn:claim1proof10},
  \end{align}
   where \eqref{eqn:claim1proof9} follows from \eqref{eqn:claim1proof1}. Now we choosing any $m$ such that $m\geq k$ guarantees that $P^{(\alpha)}_{\tilde{U}|Y}(\tilde{u}|y)\leq \frac{1}{k}$. Similarly, $P^{(\alpha)}_{\tilde{U}}(\tilde{u})\leq \frac{1}{k}$.

\section{Proofs for Section~\ref{section:variational}}
\subsection{Proof of Theorem~\ref{theorem:main-varchar}}\label{proofofthm5}
We first prove the lower bound $\text{LHS}\geq\text{RHS}$. Consider
\begin{align}
&\sup_{P_{U|X}} \log \frac{\sup_{P_{\hat{U}}} \mathbb{E}_{U\sim P_U} g(P_{\hat{U}}(U))}{\sup_{P_{\hat{U}}} \mathbb{E}_{U\sim Q_U} g(P_{\hat{U}}(U))}\nonumber\\
&=\sup_{P_{U|X}} \log \frac{\sup_{P_{\hat{U}}} \sum_{x,u} P_X(x)P_{U|X}(u|x) g(P_{\hat{U}}(u))}{\sup_{P_{\hat{U}}} \sum_{x,u} Q_X(x)P_{U|X}(u|x) g(P_{\hat{U}}(u))}\\
&=\sup_{P_{U|X}} \sup_{P_{\hat{U}}} \inf_{Q_{\hat{U}}} \log \frac{\sum_{x,u} P_X(x)P_{U|X}(u|x) g(P_{\hat{U}}(u))}{\sum_{x,u} Q_X(x)P_{U|X}(u|x) g(Q_{\hat{U}}(u))}\\
&\le \sup_{P_{U|X}} \sup_{P_{\hat{U}}}\log \frac{\sum_{x,u} P_X(x)P_{U|X}(u|x) g(P_{\hat{U}}(u))}{\sum_{x,u} Q_X(x)P_{U|X}(u|x) g(P_{\hat{U}}(u))}\\
&\le \sup_{P_{U|X}} \sup_{P_{\hat{U}}}\max_{x:P_X(x)>0} \log \frac{P_X(x) \sum_{u} P_{U|X}(u|x) g(P_{\hat{U}}(u))}{Q_X(x)\sum_{u} P_{U|X}(u|x) g(P_{\hat{U}}(u))}\label{eqn:thm1proofratiofact}\\
&=\max_{x:P_X(x)>0} \log \frac{P_X(x)}{Q_X(x)}\\
&=D_\infty(P_X\|Q_X).
\end{align}
where \eqref{eqn:thm1proofratiofact} follows because $\frac{\sum_ia_i}{\sum_ib_i}\leq \max_i\frac{a_i}{b_i}$, for $b_i>0$, $\forall i$. 

Now we prove the upper bound $\text{LHS}\leq\text{RHS}$. We lower bound the RHS of \eqref{eqn:variationalmain} by using the {\blue{``shattering" $P_{U|X}$~\cite[Proof of Theorem~1]{Issaetal},\cite[Proof of Theorem~5]{Liaoetal}}}. We pick a letter $x^\star$, and let $\mathcal{U}=\{x^\star\}\uplus\biguplus_{x\ne x^\star}\mathcal{U}_{x}$, where $|\mathcal{U}_{x}|=m$ for each $x\ne x^\star$. Then define
\begin{align}
P_{U|X}(u|x)=\begin{cases} 1 & u=x=x^\star,\\ 1/m & u\in\mathcal{U}_x,x\ne x^\star, \\ 0 & \text{otherwise.}\end{cases}
\end{align}
Note that
\begin{align}
P_U(u)=\begin{cases} P_X(x^\star) & u=x^\star\\ P_X(x)/m, & u\in\mathcal{U}_x,x\ne x^\star\end{cases}, 
\end{align}
and
\begin{align}Q_U(u)=\begin{cases} Q_X(x^\star) & u=x^\star\\ Q_X(x)/m, & u\in\mathcal{U}_x,x\ne x^\star\end{cases}.
\end{align}
Consider the numerator of the objective function in the RHS of \eqref{eqn:variationalmain}. We have
\begin{align}
\sup_{P_{\hat{U}}} \mathbb{E}_{U\sim P_U} \left[g(P_{\hat{U}}(U))\right]
&=\sup_{P_{\hat{U}}} \sum_u P_U(u) g(P_{\hat{U}}(u))
\\&\geq \sup_{P_{\hat{U}}} P_X(x^\star) g(P_{\hat{U}}(x^\star))
\\&=P_X(x^\star) \sup_{q\in[0,1]} g(q)\label{eqn:thm1proof1}.
\end{align}
Note that the expression in \eqref{eqn:thm1proof1} is finite because of the assumption on $g$ that $\sup_{q\in[0,1]}g(q)<\infty$. 

To bound the denominator of the objective function in the RHS of \eqref{eqn:variationalmain}, we will need the upper concave envelope of $g$, denoted $g^{**}$. Since $g$ is a function of a scalar, its upper concave envelope can be written as
\begin{align}
g^{**}(q)=\sup_{a,b,\lambda\in[0,1]: a\lambda+b(1-\lambda)=q} \lambda g(a)+(1-\lambda) g(b).
\end{align}
We claim that $g^{**}(0)=0$ and $g^{**}$ is continuous at 0. Fix some $\epsilon>0$. It suffices to show that there exists $\delta>0$ where $g^{**}(q)\leq \epsilon$ for all $q\in[0,\delta]$. By the assumption that $g(0)=0$ and $g$ is continuous at 0, there exists a $\delta$ small enough so that $g(q)\leq \epsilon/2$ for all $q\in[0,\sqrt{\delta}]$. Now, for any $q\in[0,\delta]$, consider any $a,b,\lambda$ where $a\lambda+b(1-\lambda)=q$. We assume without loss of generality that $a\leq q\leq b$. If $b\leq\sqrt{\delta}$, then we have $\lambda g(a)+(1-\lambda) g(b)\leq \epsilon/2$.
If $b>\sqrt{\delta}$, then we have
\begin{align}
q=a\lambda+b(1-\lambda)\geq b(1-\lambda)>\sqrt{\delta}(1-\lambda).
\end{align}
So we get $1-\lambda<\frac{q}{\sqrt{\delta}}\leq \frac{\delta}{\sqrt{\delta}}\leq\sqrt{\delta}$.
Thus
\begin{align}
\lambda g(a)+(1-\lambda) g(b)
&\leq \epsilon/2+\sqrt{\delta} \sup_{q\in[0,1]} g(q)\leq\epsilon\label{eqn:thm1proof5},
\end{align}
where \eqref{eqn:thm1proof5} holds for sufficiently small $\delta$, and again we have used the assumption that $\sup_{q\in[0,1]} g(q)<\infty$. This proves that $g^{**}(q)\leq\epsilon$ whenever $q\in[0,\delta]$. In particular, for sufficiently large $m$, 
\begin{align}\label{eqn:proofthm13}
\sup_{q\in[0,1/m]} g^{**}(q)\leq \epsilon.
\end{align}
Now the denominator in \eqref{eqn:variationalmain} can be upper bounded as
\begin{align}
&\sup_{P_{\hat{U}}} \mathbb{E}_{U\sim Q_U} \left[g(P_{\hat{U}}(U)))\right]\nonumber\\
&=\sup_{P_{\hat{U}}} \sum_u Q_U(u) g(P_{\hat{U}}(u))\\
&=\sup_{P_{\hat{U}}} (Q_X(x^\star) g(P_{\hat{U}}(x^\star))\nonumber\\
&\hspace{12pt}+\sum_{x\ne x^\star} Q_X(x) \sum_{u\in\mathcal{U}_x} \frac{1}{m} g(P_{\hat{U}}(u)))\\
&\leq \sup_{P_{\hat{U}}} Q_X(x^\star) g(P_{\hat{U}}(x^\star))\nonumber\\
&\hspace{12pt}+\sum_{x\neq x^\star} Q_X(x)  g^{**}(\sum_{u\in\mathcal{U}_x} \frac{1}{m} P_{\hat{U}}(u))\label{eqn:thm1proof2}\\
&\leq \sup_{q\in[0,1]} Q_X(x^\star) g(q)+\sum_{x\neq x^\star} Q_X(x) \sup_{q\in[0,1/m]} g^{**}(q)\\
&\leq \sup_{q\in[0,1]} Q_X(x^\star) g(q)+(1-Q_X(x^\star)) \epsilon,\label{eqn:thm1proof3}
\end{align}
where \eqref{eqn:thm1proof2} follows from the definition of the upper concave envelope and \eqref{eqn:thm1proof3} follows from \eqref{eqn:proofthm13} for sufficiently large $m$.

Putting together the bounds in \eqref{eqn:thm1proof1} and \eqref{eqn:thm1proof3}, we have
\begin{align}
   & \sup_{P_{U|X}} \log \frac{\sup_{P_{\hat{U}}} \mathbb{E}_{U\sim P_U} g(P_{\hat{U}}(U))}{\sup_{P_{\hat{U}}} \mathbb{E}_{U\sim Q_U} g(P_{\hat{U}}(U))}\nonumber\\
    &\geq \log\max_{x^\star}\ \sup_{\epsilon>0}\ \frac{\sup_{q\in[0,1]} P_X(x^\star) g(q)}{\sup_{q\in[0,1]} Q_X(x^\star) g(q)+(1-Q_X(x^\star)) \epsilon}\\
        &= \log\max_{x^\star} \frac{\sup_{q\in[0,1]} P_X(x^\star) g(q)}{\sup_{q\in[0,1]} Q_X(x^\star) g(q)}\\
        &=\log\max_{x^\star} \frac{P_X(x^\star)}{Q_X(x^\star)}\label{eqn:thm1proof4}\\
        &=D_{\infty}(P_X\|Q_X),
\end{align}
where \eqref{eqn:thm1proof4} uses the assumption that $\sup_{q\in[0,1]}g(q)<\infty$. Finally, we note that the assumption $\sup_{q\in[0,1]}g(q)>0$ is to ensure that the objective function in \eqref{eqn:variationalmain} is well-defined. In particular, for any $P_{U|X}$, fix a $u^\prime$ such that $P_U(u^\prime)>0$. Then we have
\begin{align}
    \sup_{P_{\hat{U}}} \mathbb{E}_{U\sim P_U} \left[g(P_{\hat{U}}(U))\right]&\geq \sup_{P_{\hat{U}}} g(P_{\hat{U}}(u^\prime))P_U(u^\prime)\\
    &= P_U(u^\prime)\sup_{q\in[0,1]}g(q)>0.
\end{align}
Similarly,  $\sup_{P_{\hat{U}}} \mathbb{E}_{U\sim Q_U} \left[g(P_{\hat{U}}(U))\right]>0$.

\subsection{Proof of Proposition~\ref{proposition}}\label{appendix1}
 For any $R_X\ll P_X$, consider
\begin{align}
   & D(R_X\|Q_X)-D(R_X\|P_X)\nonumber\\
    &=\sum_{x}R_X(x)\log{\frac{R_X(x)}{Q_X(x)}}-\sum_{x}R_X(x)\log{\frac{R_X(x)}{P_X(x)}}\\
    &=\sum_xR_X(x)\log{\frac{P_X(x)}{Q_X(x)}}\\
    &\leq \sum_x R_X(x)\left(\max_{x^\prime}\log\frac{P_X(x^\prime)}{Q_X(x^\prime)}\right)\label{eqn:appendix1}\\
    &=\max_{x^\prime}\log\frac{P_X(x^\prime)}{Q_X(x^\prime)}\\
    &=D_\infty(P_X\|Q_X).
\end{align}
Moreover, for $R_X$ such that $R_X(x^*)=1$ for a fixed $x^*\in\argmax \frac{P_X(x)}{Q_X(x)}$, \eqref{eqn:appendix1} is tight. This proves \eqref{eqn:varinf-Anantharam}.

To prove \eqref{eqn:varinf-Birrell}, for the upper bound, we give a choice of the function $f$ for which the objective function in the RHS of \eqref{eqn:varinf-Birrell} is equal to $D_\infty(P_X||Q_X)$. In particular, fix an $x^*\in\argmax\frac{P_X(x)}{Q_X(x)}$ and consider a function $\tilde{f}$ defined by
\begin{align}
    \tilde{f}(x)=\begin{cases}
    1, &\ \text{if}\ x=x^*,\\
    0, &\ \text{otherwise}
    \end{cases}.
\end{align}
Clearly, we have 
\begin{align}\log\frac{\mathbb{E}_{X\sim P_X}[\tilde{f}(X)]}{\mathbb{E}_{X\sim Q_X}[\tilde{f}(X)]}=\log\frac{P_X(x^*)}{Q_X(x^*)}=D_\infty(P_X||Q_X).
\end{align}
For the lower bound, consider
\begin{align}
    \log\frac{ \mathbb{E}_{X\sim P_X}[f(X)]}{ \mathbb{E}_{X\sim Q_X}[f(X)]}&=\log\frac{\sum_xP_X(x)f(x)}{\sum_xQ_X(x)f(x)}\\
    &\leq \log\max_x\frac{P_X(x)f(x)}{Q_X(x)f(x)}\label{eqn:appendixvar2proof1}\\
    &=\max_x\log\frac{P_X(x)}{Q_X(x)}=D_\infty(P_X\|Q_X),
\end{align}
where \eqref{eqn:appendixvar2proof1} follows from the fact that $\frac{\sum_ia_i}{\sum_ib_i}\leq \max_i\frac{a_i}{b_i}$, for $b_i>0$, $\forall i$. 
Taking supremum over all $f$, we get
\begin{align}
    \sup_{f:\mathcal{X}\rightarrow[0,\infty)}\log\frac{\mathbb{E}_{X\sim P_X}[f(X)]}{\mathbb{E}_{X\sim Q_X}[f(X)]}&\leq \log\max_{x\in\mathcal{X}}\frac{P_X(x)}{Q_X(x)}\\
    &=D_\infty(P_X||Q_X).
\end{align}
This proves \eqref{eqn:varinf-Birrell}.

\subsection{Proof of Corollary~\ref{thm:alphaleakage-variants}}\label{proofofthm6}
The expression for opportunistic maximal $g$-leakage in \eqref{opp-g-leakage} can be simplified as 
\begin{align}
&\tilde{\mathcal{L}}_g^{\text{max}}(X\rightarrow Y)\nonumber\\
     &=\log \sum_{y\in\text{supp}(Y)}P_Y(y)\nonumber\\
     &\hspace{12pt}\sup_{U:U-X-Y}\frac{\sup_{P_{\hat{U}|Y=y}}\mathbb{E}_{U|Y=y}[g(P_{\hat{U}|Y}(U|y))]}{\sup_{P_{\hat{U}}}\mathbb{E}_U[g(P_{\hat{U}}(U))]}\\
     &=\log\sum_{y:P_Y(y)>0}P_Y(y)\max_{x:P_{X|Y}(x|y)>0}\frac{P_{X|Y}(x|y)}{P_{X}(x)}\label{eqn:opportunisticproof2}\\
     &=\log\sum_{y:P_Y(y)>0}P_Y(y)\max_{x:P_{X|Y}(x|y)>0}\frac{P_{Y|X}(y|x)}{P_{Y}(y)}\label{eqn:opportunisticproof1}\\
    &=\log\sum_{y:P_Y(y)>0}\max_{x:P_{X|Y}(x|y)>0}P_{Y|X}(y|x)\\
     &=\log\sum_{y:P_Y(y)>0}\max_{x:P_{X}(x)>0}P_{Y|X}(y|x)\\
     &=I_{\infty}^{\text{S}}(X;Y),
\end{align}
where \eqref{eqn:opportunisticproof2} follows from {\blue{Theorem~\ref{theorem:main-varchar}}} and \eqref{eqn:opportunisticproof1} follows from the Bayes' Rule. 

The expression for maximal realizable $g$-leakage in \eqref{max-rel-gleakage} can be simplified as
\begin{align}
&\mathcal{L}_\alpha^{\text{r}-\max}(X\rightarrow Y)\nonumber\\
    &=\sup_{U:U-X-Y}\log\max_{y\in\text{supp}(Y)}\nonumber\\
    &\hspace{12pt}\frac{\sup_{P_{\hat{U}|Y=y}}\mathbb{E}\left[g(P_{\hat{U}|Y}(U|y))\right]}{\sup_{P_{\hat{U}}}\mathbb{E}\left[g(P_{\hat{U}}(U))\right]}\nonumber\\
    &=\log\max_{y\in\text{supp}(Y)}\nonumber\\
    &\hspace{12pt}\sup_{U:U-X-Y}\frac{\sup_{P_{\hat{U}|Y=y}}\mathbb{E}\left[g(P_{\hat{U}|Y}(U|y))\right]}{\sup_{P_{\hat{U}}}\mathbb{E}\left[g(P_{\hat{U}}(U))\right]}\\
    &=\log\max_{y:P_{Y}(y)>0}\max_{x:P_{X|Y}(x|y)>0}\frac{P_{X|Y}(x|y)}{P_{X}(x)}\label{eqn:realizableproof1}\\
    &=\log\max_{(x,y):P_{XY}(x,y)>0}\frac{P_{XY}(x,y)}{P_X(x)P_Y(y)}\\
    &=D_{\infty}(P_{XY}\|P_X\times P_Y),
\end{align}
where \eqref{eqn:realizableproof1} follows from {\blue{Theorem~\ref{theorem:main-varchar}}}. 

\section{Variational Characterization for $D_\infty(P_X\|Q_X)$ with Gain Function $g(t)=\log{t}$}\label{appendix-loggain}
Here we show that \eqref{eqn:variationalmain} holds for the non-positive gain function $g(t)=\log{t}$ that does not satisfy the conditions in Theorem~\ref{theorem:main-varchar}. The proof of the lower bound follows exactly along the same lines as that of Theorem~\ref{theorem:main-varchar} with the only difference that \eqref{eqn:thm1proofratiofact} holds for negative gain functions too noticing that $\frac{\sum_ia_i}{\sum_ib_i}\leq \max_i\frac{a_i}{b_i}$, for $b_i<0$, $\forall i$. 

For the upper bound, we first note that 
\begin{align}
    \sup_{P_{\hat{U}}}\mathbb{E}_{U\sim P_U}\left[\log{P_{\hat{U}}(U)}\right]&=-\inf_{P_{\hat{U}}} \left(H_P(U)+D(P_U\|P_{\hat{U}})\right)\nonumber\\
    &=-H_P(U).
\end{align}
We lower bound the RHS in \eqref{eqn:variationalmain} with gain function $g(t)=\log{t}$ by using the {\blue{``shattering" $P_{U|X}$~\cite[Proof of Theorem~1]{Issaetal},\cite[Proof of Theorem~5]{Liaoetal}}}. Let $\mathcal{U}=\uplus_{x\in\mathcal{X}}\mathcal{U}_x, \ |\mathcal{U}_x|=m_x$.
Define $P_{U|X}(u|x)=\frac{1}{m_x}$, $u\in\mathcal{U}_x$. So, we have
\begin{align}
       &\sup_{P_{U|X}}\frac{\sup_{P_{\hat{U}}}\mathbb{E}_{U\sim P_U}\left[\log{P_{\hat{U}}(U)}\right]}{\sup_{P_{\hat{U}}}\mathbb{E}_{U\sim Q_U}\left[\log{Q_{\hat{U}}(U)}\right]}
       \geq\frac{-H_P(U)}{-H_Q(U)}\\
       &=\frac{\sum_u\left(\sum_xP_X(x)P_{U|X}(u|x)\right)\log{\left(\sum_xP_X(x)P_{U|X}(u|x)\right)}}{\sum_u\left(\sum_xQ_X(x)P_{U|X}(u|x)\right)\log{\left(\sum_xQ_X(x)P_{U|X}(u|x)\right)}}\\
       &=\frac{\sum_xP_X(x)\log{\frac{P_X(x)}{m_x}}}{\sum_xQ_X(x)\log{\frac{Q_X(x)}{m_x}}}\\
       &=\frac{-H_P(X)/\log{m_{x^*}}-P_X(x^*)}{-H_Q(X)/\log{m_{x^*}}-Q_X(x^*)}\label{eqn:log-var-1}\\
       &=\frac{P_X(x^*)}{Q_X(x^*)}\label{eqn:log-var2}\\
       &=2^{D_\infty(P_X\|Q_X)},
\end{align}
where \eqref{eqn:log-var-1} follows by fixing an $x^*\in\argmax_x\frac{P_X(x)}{Q_X(x)}$ and choosing $m_x=1$, for $x\neq x^*$, and \eqref{eqn:log-var2} then follows by taking limit $m_{x^*}\rightarrow \infty$.

 \section{Opportunistic maximal and maximal realizable $(\alpha=1)$-leakages}\label{appendix:1-leakages}
 We first note that opportunistic maximal $\alpha$-leakage in LHS of \eqref{eqn:opp-alpha-leakage} can be written as (see \cite[Equations~(18)-(20)]{Issaetal})
 \begin{align}
     \tilde{\mathcal{L}}_\alpha^{\text{max}}&(X\rightarrow Y)=\sup_U\frac{\alpha}{\alpha-1}\log\sum_{y\in\text{supp}(Y)}P_Y(y)\nonumber\\
     &\frac{\left(\sum_{u}\left(\sum_xP_{U|XY}(u|x,y)P_{X|Y}(x|y)\right)^\alpha\right)^\frac{1}{\alpha}}{\left(\sum_{u}\left(\sum_xP_{U|XY}(u|x,y)P_{X}(x)\right)^\alpha\right)^\frac{1}{\alpha}}.
\end{align}
Let us define opportunistic maximal and maximal realizable $(\alpha=1)$-leakages as
\begin{align}
& \tilde{\mathcal{L}}^{\text{max}}_1(X\rightarrow Y)=\sup_U\lim_{\alpha\rightarrow 1}\frac{\alpha}{\alpha-1}\log\sum_{y\in\text{supp}(Y)}P_Y(y)\nonumber\\
 &\hspace{12pt}\frac{\left(\sum_{u}\left(\sum_xP_{U|XY}(u|x,y)P_{X|Y}(x|y)\right)^\alpha\right)^\frac{1}{\alpha}}{\left(\sum_{u}\left(\sum_xP_{U|XY}(u|x,y)P_{X}(x)\right)^\alpha\right)^\frac{1}{\alpha}},\\
 &\mathcal{L}_1^{\text{r}-\text{max}}(X\rightarrow Y)=\sup_{U:U-X-Y}\max_{y\in\text{supp(Y)}}\nonumber\\
 &\hspace{12pt}\lim_{\alpha\rightarrow 1}\frac{\alpha}{\alpha-1}\log\frac{\left(\sum_uP_{U|Y}(u|y)^\alpha\right)^{\frac{1}{\alpha}}}{\left(\sum_uP_U(u)^\alpha\right)^{\frac{1}{\alpha}}},
 \end{align}
 by taking the limit first and the supremum next. {\blue{When $X$ and $Y$ are independent, note that the expressions for both the leakages above are equal to zero.}} 
  The following proposition plays a crucial role in proving that these leakages are equal to infinity.
 \begin{proposition}\label{proposition:markov}
{\blue{Let $P_{XY}$ be a probability distribution over a finite alphabet $\mathcal{X}\times\mathcal{Y}$ such that $X$ and $Y$ are not independent. We have}}
 \begin{align}
     \sup_{U:U-X-Y}\left(H(U)-H(U|Y=y)\right)=\infty.
 \end{align}
 \end{proposition}
 \begin{remark}
 It is easy to see that $\sup_{U:U-X-Y}\left(H(U)-H(U|Y)\right)=\sup_{U:U-X-Y}I(U;Y)=I(X;Y)$, since we have $I(U;Y)\leq I(X;Y)$ for every $U$ such that $U-X-Y$ by data processing inequality. However, if we replace the conditional entropy in the objective function with a conditional entropy where the conditioning is on a particular realization of $Y$ (instead of the random variable itself), it is interesting that the supremum blows up to infinity. 
 \end{remark}
 \begin{proof}
  For a fixed $Y=y$, we have
 \begin{align}
&\sup_{U:U-X-Y}\left(H(U)-H(U|Y=y)\right)\nonumber\\
&\hspace{12pt}\geq \sup_{U:U-X-Y, H(X|U)=0}\left(H(U)-H(U|Y=y)\right)\label{eqn:thm141}.
\end{align}
In the RHS of \eqref{eqn:thm141}, we further assume that  $\mathcal{U}=\bigcup_{x\in\mathcal{X}}\mathcal{U}_x$ be the alphabet of $U$ such that 
\begin{align}
    P_{U|X}(u|x)=\begin{cases}
    \frac{1}{n_x}, &u\in\mathcal{U}_x\\
    0, &\text{otherwise}
    \end{cases},
\end{align}
where $n_x=\lvert\mathcal{U}_x\rvert$ is to be fixed later.
This together with the Markov chain $U-X-Y$ gives 
\begin{align}
P_U(u)&=\sum_xP_X(x)P_{U|X}(u|x)=\frac{P_X(x)}{n_x}, \ \text{for}\ u\in\mathcal{U}_x,\label{eqn:thm143}\\
P_{U|Y}(u|y)&=\sum_xP_{X|Y}(x|y)P_{U|X}(u|x)=\frac{P_{X|Y}(x|y)}{n_x},\nonumber\\ 
&\hspace{12pt}\ \text{for}\ u\in\mathcal{U}_x\label{eqn:thm144}.
\end{align}
Continuing \eqref{eqn:thm141}, we get
\begin{align}
    &\sup_{U:U-X-Y}\left(H(U)-H(U|Y=y)\right)\\
     &\geq -\sum_{x\in\mathcal{X}}\sum_{u\in\mathcal{U}_x}P_U(u)\log{P_U(u)}\nonumber\\
     &\hspace{12pt}+\sum_{x\in\mathcal{X}}\sum_{u\in\mathcal{U}_x}P_{U|Y}(u|y)\log{P_{U|Y}(u|y)}\\  
    &=-\sum_{x\in\mathcal{X}}\sum_{u\in\mathcal{U}_x}\frac{P_X(x)}{n_x}\log{\left(\frac{P_X(x)}{n_x}\right)}\nonumber\\
    &\hspace{12pt}+\sum_{x\in\mathcal{X}}\sum_{u\in\mathcal{U}_x}\frac{P_{X|Y}(x|y)}{n_x}\log{\left(\frac{P_{X|Y}(x|y)}{n_x}\right)}\label{eqn:thm145}\\
    &=-\sum_{x\in\mathcal{X}}P_X(x)\log{\left(\frac{P_X(x)}{n_x}\right)}\nonumber\\
    &\hspace{12pt}+\sum_{x\in\mathcal{X}}P_{X|Y}(x|y)\log{\left(\frac{P_{X|Y}(x|y)}{n_x}\right)}\label{eqn:thm142}\\
    &=-\sum_xP_X(x)\log{P_X(x)}+\sum_xP_X(x)\log{n_x}\nonumber\\
    &\hspace{12pt}+\sum_xP_{X|Y}(x|y)\log{P_{X|Y}(x|y)}-\sum_xP_{X|Y}(x|y)\log{n_x}\\
    &=H(X)-H(X|Y=y)\nonumber\\
    &\hspace{12pt}+\sum_x\log{n_x}\left(P_X(x)-P_{X|Y}(x|y)\right)\label{eqn:thm4final}
\end{align}
where \eqref{eqn:thm145} follows from \eqref{eqn:thm143} and \eqref{eqn:thm143}, and \eqref{eqn:thm142} follows because $\lvert\mathcal{U}_x\rvert=n_x$. Note that $1\leq n_x\leq\infty$, for every $x\in\mathcal{X}$. Now choosing 
\begin{align}
    n_x=\begin{cases}
    1, &\text{for}\ x\  \text{s.t.}\ P_X(x)\leq P_{X|Y}(x|y)\\
    \infty, &\text{for}\ x\ \text{s.t.}\ P_X(x)>P_{X|Y}(x|y),
    \end{cases}
\end{align}
implies that the summation term in \eqref{eqn:thm4final} equals infinity.
 \end{proof}
 Then we have the following theorems.
\begin{theorem}\label{thm:11}
{\blue{Let $P_{XY}$ be a probability distribution over a finite alphabet $\mathcal{X}\times\mathcal{Y}$ such that $X$ and $Y$ are not independent. We have}}
 \begin{align}
      \tilde{\mathcal{L}}^{\emph{max}}_1(X\rightarrow Y)=\infty.
 \end{align}
 \end{theorem}
 \begin{proof}
 Using L'Hospital's rule and the fact that,
 \begin{align}\lim_{\alpha\rightarrow 1}\frac{d}{d\alpha}\left(a^\alpha+b^\alpha\right)^\frac{1}{\alpha}=a\log{a}+b\log{b}-(a+b)\log{(a+b)},
 \end{align} 
 for $a,b\geq 0$, we can verify that, for $p_i,a_i,b_i,c_i,d_i\geq 0$, and $a_i+b_i=c_i+d_i=1$, $i\in[1:t]$,
 \begin{align}\label{eqn:oppaleakage1}
     &\lim_{\alpha\rightarrow 1}\frac{\alpha}{\alpha-1}\log{\left(\sum_{i=1}^tp_i\frac{\left(a_i^\alpha+b_i^\alpha\right)^\frac{1}{\alpha}}{\left(c_i^\alpha+d_i^\alpha\right)^\frac{1}{\alpha}}\right)}\nonumber\\
     &\hspace{12pt}=\sum_{i=1}^tp_i\left(a_i\log{a_i}+b_i\log{b_i}-c_i\log{c_i}-d_i\log{d_i}\right).
 \end{align}
 Now we have
 \begin{align}
   &\tilde{\mathcal{L}}^{\text{max}}_1(X\rightarrow Y)\nonumber\\
   &=\sup_U\lim_{\alpha\rightarrow 1}\frac{\alpha}{\alpha-1}\log\sum_{y\in\text{supp}(Y)}P_Y(y)\nonumber\\
   &\hspace{12pt}\frac{\left(\sum_{u}\left(\sum_xP_{U|XY}(u|x,y)P_{X|Y}(x|y)\right)^\alpha\right)^\frac{1}{\alpha}}{\left(\sum_{u}\left(\sum_xP_{U|XY}(u|x,y)P_{X}(x)\right)^\alpha\right)^\frac{1}{\alpha}}\\
   &=\sup_U \sum_{y\in\text{supp}(Y)}\sum_u\bigg(P_{U|Y}(u|y)\log(P_{U|Y}(u|y))\nonumber\\
   &\hspace{12pt}-\tilde{P}_{U|Y}(u|y)\log{(\tilde{P}_{U|Y}(u|y))}\bigg),\label{eqn:oppaleakage2}
   \end{align}
   where \eqref{eqn:oppaleakage2} follows from \eqref{eqn:oppaleakage1} with
   \begin{align}
       P_{U|Y}(u|y)&=\sum_xP_{U|XY}(u|x,y)P_{X|Y}(x|y),\\
       \tilde{P}_{U|Y}(u|y)&=\sum_xP_{U|XY}(u|x,y)P_{X}(x).
   \end{align}
   Continuing \eqref{eqn:oppaleakage2}, we get
   \begin{align}
       &\tilde{\mathcal{L}}^{\text{max}}_1(X\rightarrow Y)\nonumber\\
       &=\sup_U \sum_{y\in\text{supp}(Y)}P_Y(y)\sum_u\bigg(P_{U|Y}(u|y)\log{(P_{U|Y}(u|y)})\nonumber\\
       &\hspace{12pt}-\tilde{P}_{U|Y}(u|y)\log{(\tilde{P}_{U|Y}(u|y))}\bigg)\\
       &=\sup_{(U_y:y\in\mathcal{Y})-X-Y} \sum_{y\in\text{supp}(Y)}P_Y(y)\nonumber\\
       &\hspace{12pt}\sum_u\left(P_{U_y|Y}(u|y)\log{P_{U_y|Y}(u|y))}-P_{U_y}(u)\log{P_{U_y}(u)}\right)\\
       &=\sum_{y\in\text{supp}(Y)}P_Y(y)\sup_{U_y:U_y-X-Y}\nonumber\\
       &\hspace{12pt}\sum_u\left(P_{U_y|Y}(u|y)\log{P_{U_y|Y}(u|y))}-P_{U_y}(u)\log{P_{U_y}(u)}\right)\\
        &=\sum_{y\in\text{supp}(Y)}P_Y(y)\sup_{U:U-X-Y}\nonumber\\
        &\hspace{12pt}\sum_u\left(P_{U|Y}(u|y)\log{P_{U|Y}(u|y))}-P_{U}(u)\log{P_{U}(u)}\right)\\
        &=\sum_{y\in\text{supp}(Y)}P_Y(y)\sup_{U:U-X_-Y}\left(H(U)-H(U|Y=y)\right)\\
        &=\infty\label{eqn:oppaleakage3},
   \end{align}
   where \eqref{eqn:oppaleakage3} follows from Proposition~\ref{proposition:markov}.
 \end{proof}
 \begin{theorem}\label{thm:12}
 {\blue{Let $P_{XY}$ be a probability distribution over a finite alphabet $\mathcal{X}\times\mathcal{Y}$ such that $X$ and $Y$ are not independent. We have}}
 \begin{align}
 \mathcal{L}_1^{\emph{r}-\emph{max}}(X\rightarrow Y)=\infty.
 \end{align}
 \end{theorem}
 \begin{proof}
 Consider
 \begin{align}
     &\mathcal{L}_1^{\text{r}-\text{max}}(X\rightarrow Y)\nonumber\\
     &:=\sup_{U:U-X-Y}\max_{y\in\text{supp}(Y)}\lim_{\alpha\rightarrow 1}\frac{\alpha}{\alpha-1}\log\frac{\left(\sum_uP_{U|Y}(u|y)^\alpha\right)^{\frac{1}{\alpha}}}{\left(\sum_uP_U(u)^\alpha\right)^{\frac{1}{\alpha}}}\\
     &=\sup_{U:U-X-Y}\max_{y\in\text{supp}(Y)}\lim_{\alpha\rightarrow 1}\bigg(\frac{\alpha}{\alpha-1}\log{(\sum_uP_{U|Y}(u|y)^\alpha)^{\frac{1}{\alpha}}}\nonumber\\
     &\hspace{12pt}\frac{\alpha}{\alpha-1}\log{(\sum_uP_U(u)^\alpha)^{\frac{1}{\alpha}}}\bigg)\\
     &=\sup_{U:U-X-Y}\max_{y\in\text{supp}(Y)}\lim_{\alpha\rightarrow 1}\left(H_\alpha(U)-H_\alpha(U|Y-y)\right)\\
     &=\sup_{U:U-X-Y}\max_{y\in\text{supp}(Y)}\left(H(U)-H(U|Y=y)\right)\\
     &=\max_{y\in\text{supp}(Y)}\sup_{U:U-X-Y}\left(H(U)-H(U|Y=y)\right)\\
     &=\infty\label{eqn:propmarkov1},
 \end{align}
 where \eqref{eqn:propmarkov1} follows from Proposition~\ref{proposition:markov}.
 \end{proof}
 {\blue{We remark that Theorem~\ref{thm:12} also follows from Theorem~\ref{thm:11} by noticing that $\mathcal{L}^{r-\text{max}}(X\rightarrow Y)\geq \tilde{\mathcal{L}}^\text{max}(X\rightarrow Y)$.}}
\section{Bregman Divergence}\label{appx:Bregman}
Let $F: \Omega \rightarrow \mathbb{R}$ be a continuously differentiable, strictly convex function on a closed, convex set $\Omega$.
The Bregman divergence associated with $F$~\cite{BREGMAN1967200} for points $p,q \in \Omega$ is the difference between the value of $F$ at $p$ and the value of the first-order Taylor expansion of $F$ around $q$ evaluated at point $p$, i.e.,
\begin{align}
B_{F}(p,q) = F(p) - F(q) - \langle \nabla F(q), p-q \rangle.
\end{align}
Let $p=(p_1,\dots,p_d)$ and $q=(q_1,\dots,q_d)$ be two discrete probability distributions. Consider $F(p)=\frac{\alpha}{\alpha-1}\left(\left(\sum_ip_i^\alpha\right)^\frac{1}{\alpha}-1\right)$ which is a strictly convex function on  the $d$-simplex. The associated Bregman divergence is given by
\begin{align}
    &B_F(p,q)=F(p)-F(q)-\langle \nabla F(q), p-q \rangle\\
    &=\frac{\alpha}{\alpha-1}\bigg[(\sum_ip_i^\alpha)^\frac{1}{\alpha}-1-(\sum_iq_i^\alpha)^\frac{1}{\alpha}+1\nonumber\\
    &\hspace{12pt}-\sum_i(p_i-q_i)(\sum_jq_j^\alpha)^\frac{1-\alpha}{\alpha}q_i^{\alpha-1}\bigg]\\
    &=\frac{\alpha}{\alpha-1}\bigg[(\sum_ip_i^\alpha)^\frac{1}{\alpha}-(\sum_jq_j^\alpha)^\frac{1-\alpha}{\alpha}\nonumber\\
    &\hspace{12pt}\times(\sum_iq_i^\alpha+\sum_ip_iq_i^{\alpha-1}-\sum_iq_i^\alpha)\bigg]\\
    &=\frac{\alpha}{\alpha-1}\bigg[(\sum_ip_i^\alpha)^\frac{1}{\alpha}-(\sum_jq_j^\alpha)^\frac{1-\alpha}{\alpha}\sum_ip_iq_i^{\alpha-1}\bigg]\\
    &=\frac{\alpha}{\alpha-1}\bigg[(\sum_ip_i^\alpha)^\frac{1}{\alpha}-(\sum_ip_i^\alpha)^\frac{1}{\alpha}\nonumber\\
    &\hspace{12pt}\times\sum_i(\frac{p_i^\alpha}{\sum_jp_j^\alpha})^\frac{1}{\alpha}(\frac{q_i^\alpha}{\sum_jq_j^\alpha})^{1-\frac{1}{\alpha}}\bigg]\\
    &=\frac{\alpha}{\alpha-1}\bigg[(\sum_ip_i^\alpha)^\frac{1}{\alpha}(1-\mathrm{e}^{\frac{1-\alpha}{\alpha} D_{\frac{1}{\alpha}}(p^{(\alpha)}||q^{(\alpha)})})\bigg],
\end{align}
where $p^{(\alpha)}:=\left(\frac{p_i^\alpha}{\sum_jp_j^\alpha}:i\in[1:d]\right)$ and $q^{(\alpha)}:=\left(\frac{q_i^\alpha}{\sum_jq_j^\alpha}:i\in[1:d]\right)$ are tilted versions of the distributions $p$ and $q$, respectively. Thus the associated Bregman divergence is related to the R\'{e}nyi divergence of order $\frac{1}{\alpha}$ in the manner 
\begin{align}\label{eqn:bregmanfromalpha}
B_F(p,q)=\frac{\alpha}{\alpha-1}\bigg[(\sum_ip_i^\alpha)^\frac{1}{\alpha}\left(1-\mathrm{e}^{\frac{1-\alpha}{\alpha} D_{\frac{1}{\alpha}}(p^{(\alpha)}||q^{(\alpha)})}\right)\bigg].
\end{align}
Note that $\lim_{\alpha\rightarrow 1}F(p)=\sum_{i=1}^dp_i\log{p_i}$ and the resulting Bregman divergence is equal to relative entropy. 
\section{Proofs of {\nblue{Lemma}}~\ref{fact:positivityrelent}}\label{appndx:proofsof-facts}
Let $A=\{x:\tilde{P}_X(x)>0\}$. Then
\begin{align}
-D(\tilde{P}_X||\tilde{Q}_X)&=\sum_{x\in A}\tilde{P}_X(x)\log\left(\frac{\tilde{Q}_X(x)}{\tilde{P}_X(x)}\right)\\
&\leq \log\left(\sum_{x\in A}\tilde{P}_X(X) \frac{\tilde{Q}_X(x)}{\tilde{P}_X(x)} \right)\label{eqn:jensen's}\\
&= \log\left(\sum_{x\in A}\tilde{Q}_X(x) \right)\\
&\leq \log(1)\label{eqn:factgivenlessthan1}\\
&=0,
\end{align}
where \eqref{eqn:jensen's} follows by Jensen's inequality, \eqref{eqn:factgivenlessthan1} follows since $\sum\limits_x\tilde{Q}_X(x)\leq 1$.

\section{Proof of Lemma~\ref{fact:alphaloss2noneq}}\label{app1}
Let $\mathcal{X}=\{x_1,x_2,\dots,x_n\}$ and $P_X(x_i)=p_i$, for $i\in[1:n]$. Consider $a_{[1:k]}$ such that $a_i=a_j$ for some $i\neq j$. There exists a $b_{[1:k]}$ such that for each $i\in[1:k]$, we have $a_i=b_j$ for some $j$ and $b_r\neq a_j$ for some $r$ and any $j$. Consider
\begin{align}
    \frac{\alpha}{\alpha-1}\text{\small$\left[\sum_{i=1}^np_i\left(1-\mathrm{P}^*\left(\bigcup_{j=1}^k(\hat{X}_j=x_i)\right)^\frac{\alpha-1}{\alpha}\right)\right]$}.\label{eqn:appendix}
\end{align}
Let $\mathcal{A}$ and $\mathcal{B}$ denote the sets of all multiset permutations of $a_{[1:k]}$ and $b_{[1:k]}$, respectively, when $a_{[1:k]}$ and $b_{[1:k]}$ are treated as  multisets. Let $q_{a_1,a_2,\dots,a_k}:=\sum_{r_{[1:k]}\in\mathcal{A}}P_{\hat{X}_{[1:k]}}(r_{[1:k]})$ and $q_{b_1,b_2,\dots,b_k}:=\sum_{r_{[1:k]}\in\mathcal{B}}P_{\hat{X}_{[1:k]}}(r_{[1:k]})$.
Each term out of the $n$ terms in \eqref{eqn:appendix} will either contain both $q_{a_{[1:k]}}$ and $q_{b_{[1:k]}}$ (say, type 1), contain just $q_{b_{[1:k]}}$ alone (say, type 2), or does not contain both (say, type 3). We now construct a new strategy $P_{\hat{X}_{[1:k]}}$ by incorporating the value of $q_{a_{[1:k]}}$ into $q_{b_{[1:k]}}$ making the value of new $q_{a_{[1:k]}}$ equal to zero. Now the values of the terms of type 2 strictly decrease as the $\alpha$-loss function is strictly decreasing in its argument while retaining the values of the terms of types 1 and 3. This leads to a contradiction since $P^*_{X_{[1:k]}}$ is assumed to be an optimal strategy. So, $P_{\hat{X}_{[1:k]}}(a_{[1:k]})=0$. Repeating the same argument as above for all such $a_{[1:k]}$ s.t. $a_i=a_j$, for some $i\neq j$ completes the proof. 

\section{Proof of Lemma~\ref{theorem:probclass1}}\label{app2}
 \noindent\underline{`Only if' part}: Suppose a vector $(t_1,t_2,\dots,t_n)$ is admissible. Then there exists $P_{\hat{X}_{[1:k]}}$ satisfying \eqref{eqn:linsystemk}. Using \eqref{eqn:systemprob}, since $t_i$ is probability of a certain event, we have 
 \begin{align*}
     0\leq t_i \leq 1,\ \text{for}\ i\in[1:n]. 
 \end{align*}
 \underline{`If' part}: Suppose $0\leq t_i\leq 1$, for $i\in[1:n]$. Summing up all the equations in \eqref{eqn:linsystemk} over $i\in[1:n]$ and using $\sum_{i=1}^nt_i=k$, we get
 \begin{align*}
     P_{\hat{X}_{[1:k]}}(a_{[1:k]})=0,\ \text{for all}\ a_{[1:k]}\ \text{s.t.}\ a_i=a_j, \ \text{for some}\ i\neq j.
 \end{align*}
 With this, \eqref{eqn:linsystemk} can be written in the form of system of linear equation only in terms of non-negative variables of the form 
 \begin{align}
 q_{i_1,i_2,\dots,i_k}:=\sum\limits_{\sigma\in S_n}P_{\hat{X}_{[1:k]}}(x_{i_{\sigma(1)}},x_{i_{\sigma(2)}},\dots,x_{i_{\sigma(n)}}),
 \end{align}
 where $i_1,i_2,\dots,i_k$ are all distinct and belong to $[1:n]$. Here the sum is computed over all the permutations $\sigma$ of the set $\{1,2,\dots,n\}$. The set of all such permutations is denoted by $S_n$. With this, the system of equations in \eqref{eqn:linsystemk} can be written in the form $AQ=b$, $Q\geq 0$. Here $A$ is a $n\times\binom{n}{k}$-matrix, where the rows are indexed by $i\in[1:n]$ and columns are indexed by $(i_1,i_2,\dots,i_k)$, where $i_1,i_2,\dots,i_k$ are all distinct and belong to $[1:n]$. In particular, in the column indexed by $(i_1,i_2,\dots,i_k)$, the entry of $A$ corresponding to $i_j^{\text{th}}$ row is $1$, for $j\in[1:k]$. All the remaining entries of the matrix $A$ are zeros. $Q$ is $\binom{n}{k}$-length vector of variables of the form $q_{i_1,i_2,\dots,i_k}$. ${b}$ is an $n$-length vector with $b_i=t_i$. We are interested in the feasibility of the system $AQ=b$, $Q\geq 0$. We use the Farkas' lemma~\cite[Proposition 6.4.3]{Matousek} in linear programming for checking this. It states that the system $A{Q}={b}$ has a non-negative solution if and only if every ${y}\in\mathbbm{R}^n$ with ${y}^\top A\geq 0$ also implies ${y}^\top {b}\geq 0$. For our problem, ${y}^\top A\geq 0$ is equivalent to 
 \begin{align}\label{eqn:farkashyp}
     \sum_{j=1}^ky_{i_j}\geq 0,\ \text{for all distinct}\ i_1,i_2,\dots,i_k\in[1:n]. 
 \end{align}
 Without loss of generality, let us assume that $y_i\leq y_{i+1}$, $i\in[1:n-1]$. Then \eqref{eqn:farkashyp} is equivalent to
 \begin{align}\label{eqn:farkashyp1}
     \sum_{i=1}^ky_i\geq 0.
 \end{align}
 Now consider
 \begin{align}
     &\sum_{i=1}^ny_it_i\nonumber\\
     &=\sum_{i=1}^ky_it_i+y_{k+1}t_{k+1}+\sum_{i=k+2}^ny_it_i\\
     &=\sum_{i=1}^k{y_i}+\sum_{i=1}^ky_i(t_i-1)+y_{k+1}t_{k+1}+\sum_{i=k+2}^ny_it_i\\
     &\geq\sum_{i=1}^ky_i+y_{k+1}\sum_{i=1}^k(t_i-1)+y_{k+1}t_{k+1}+\sum_{i=k+2}^ny_it_i\label{theorem:probclass1proof1}\\
     &\geq\sum_{i=1}^k{y_i}+y_{k+1}\sum_{i=1}^k(t_i-1)+y_{k+1}t_{k+1}+y_{k+1}\sum_{i=k+2}^nt_i\label{theorem:probclass1proof2}\\
     &=\sum_{i=1}^ky_i+y_{k+1}\left(\sum_{i=1}^nt_i-k\right)\\
     &={\sum_{i=1}^ky_i}\label{theorem:probclass1proof3}\\
     &\geq 0\label{theorem:probclass1proof4},
 \end{align}
 where \eqref{theorem:probclass1proof1} follows because $y_i\leq y_{k+1}$ and $t_i-1\leq 0$, for $i\in[1:k]$, \eqref{theorem:probclass1proof2} follows because $y_i\geq y_{k+1}$, for $i\in[k+2:n]$, and \eqref{theorem:probclass1proof3} follows because $\sum_{i=1}^nt_i=k$, \eqref{theorem:probclass1proof4} follows from \eqref{eqn:farkashyp1}. Now using the Farkas' lemma, $AQ=b$, has a non-negative solution, i.e., the vector $(t_1,t_2,\dots,t_n)$ is admissible.

\bibliographystyle{IEEEtran}
\bibliography{Bibliography}

\begin{IEEEbiographynophoto}
{Gowtham R. Kurri} (Member, IEEE) graduated from the International Institute of Information Technology (IIIT), Hyderabad, India, with a B.\ Tech.\ degree in Electronics and Communication Engineering, in 2011. He received his M.Sc. and Ph.D. degrees from the Tata Institute of Fundamental Research, Mumbai, India in 2020. From 2020-2023, he was a Post-Doctoral Researcher at the School of Electrical, Computer and Energy Engineering at Arizona State University. Since February 2023, he has been an Assistant Professor with IIIT Hyderabad, where he is affiliated to the Signal Processing and Communications Research Centre. His research interests are in information theory and statistical machine learning.

From 2011-2012, he worked as an Associate Engineer at Qualcomm India Private Limited, Hyderabad, India. From July to October, 2019, he was a Research Intern in the Blockchain Technology Group at IBM Research, Bangalore, India. 
\end{IEEEbiographynophoto}

\begin{IEEEbiographynophoto}{Lalitha Sankar} (S'02--M'07--SM'13) received the B. Tech. degree from the Indian Institute of Technology, Bombay, the M.S. degree from the University of Maryland, and the Ph.D. degree from Rutgers University. She is currently a Professor in the School of Electrical, Computer, and Energy Engineering at Arizona State University. Her research interests include applying information theory and data science to study reliable, responsible, and privacy-protected machine learning as well as cyber security and resilience in critical infrastructure networks.
She received the National Science Foundation CAREER Award in 2014, the IEEE Globecom 2011 Best Paper Award for her work on privacy of side-information in multi-user data systems, and the Academic Excellence award from Rutgers in 2008.
\end{IEEEbiographynophoto}

\begin{IEEEbiographynophoto}{Oliver Kosut} (S'06--M'10--SM'22)
     received B.S. degrees in electrical engineering and mathematics from the Massachusetts Institute of Technology, Cambridge, MA, USA in 2004, and the Ph.D. degree in electrical and computer engineering from Cornell University, Ithaca, NY, USA in 2010.

Since 2012, he has been a faculty member in the School of Electrical, Computer and Energy Engineering at Arizona State University, Tempe, AZ, USA, where he is an Associate Professor. Previously, he was a Postdoctoral Research Associate at MIT from 2010 to 2012. His research interests include information theory---particularly with applications to security and machine learning---and power systems.

Prof.~Kosut received the NSF CAREER award in 2015. He is an associate editor for the \emph{IEEE Transactions on Information Forensics and Security}. He is an IEEE Information Theory Society Distinguished Lecturer, 2023--2024.
\end{IEEEbiographynophoto}
\end{document}